\documentclass[12pt]{article}

\usepackage[utf8]{inputenc}
\usepackage{amsmath, amssymb}
\usepackage{graphicx}
\usepackage{caption}
\usepackage{subcaption}
\usepackage{color}
\usepackage{hyperref}
\usepackage{geometry}
\usepackage{amsthm}
\newtheorem{lemma}{Lemma}
\newtheorem{assumption}{Assumption}
\newtheorem{definition}{Definition}
\newtheorem{theorem}{Theorem}
\newtheorem{corollary}{Corollary}
\newtheorem{proposition}{Proposition}
\usepackage{setspace}
\usepackage{enumitem}
\usepackage{multirow}
\usepackage{array}
\usepackage{float}
\usepackage{hyperref}
\usepackage{bbm}
\usepackage{booktabs}
\usepackage{authblk}
\hypersetup{
    colorlinks=true,
    linkcolor=blue,
    filecolor=magenta,      
    urlcolor=cyan,
    citecolor = blue,
}
\usepackage{tocloft}
\usepackage{titletoc}

\usepackage{xr}
\allowdisplaybreaks

\usepackage[round]{natbib}   
\title{Stationary Point Constrained Inference via Diffeomorphisms}
\author[1]{Michael Price}
\author[2]{Debdeep Pati}
\author[1]{Ning Ning}

\affil[1]{Department of Statistics, Texas A\&M University}
\affil[2]{Department of Statistics, University of Wisconsin-Madison}
\date{}

\begin{document}

\maketitle

\begin{abstract}
Stationary points or derivative zero crossings of a regression function correspond to points where a trend reverses, making their estimation scientifically important. Existing approaches to uncertainty quantification for stationary points cannot deliver valid joint inference when multiple extrema are present, an essential capability in applications where the relative locations of peaks and troughs carry scientific significance. We develop a principled framework for functions with multiple regions of monotonicity by constraining the number of stationary points. We represent each function in the diffeomorphic formulation as the composition of a simple template and a smooth bijective transformation, and show that this parameterization enables coherent joint inference on the extrema. This construction guarantees a prespecified number of stationary points and provides a direct, interpretable parameterization of their locations. We derive non-asymptotic confidence bounds and establish approximate normality for the maximum likelihood estimators, with parallel results in the Bayesian setting. Simulations and an application to brain signal estimation demonstrate the method’s accuracy and interpretability.
\end{abstract}

\section{Introduction}
\label{sec:intro}
Incorporating shape constraints into nonparametric regression models has proven to be a powerful tool in statistical analysis. From a practical standpoint, such constraints ensures that the fitted model adheres to essential scientific principles. Theoretically, restricting the function class to which the ground truth belongs typically reduces the generalization error relative to standard nonparametric approaches. Applications of shape-constrained nonparametric regression are widespread, including enforcing diminishing marginal returns in economic models \citep{johnson2018shape}, estimating monotonic dose-response relationships in biostatistics \citep{westling2020causal}, and modeling population dynamics in ecology \citep{citores2020modelling}. A rich literature exists on estimating regression functions under structural constraints such as monotonicity or convexity and on studying their asymptotic properties \citep{groeneboom2014nonparametric}. Given its broad applicability and the increased computational feasibility of nonparametric estimation, shape-constrained regression remains an active area of research \citep{samworth2018special}.

We begin with the standard univariate nonparametric regression model given by
\begin{equation}
Y_i = f(X_i) + \epsilon_i,
\end{equation}
where the covariate $X_i$ is independent and identically distributed (i.i.d.) samples drawn from a probability measure $\mathbb{P}_X$, the errors $\epsilon_i$ are i.i.d. as $N(0, \sigma^2)$, and the regression function $f(x)$ belongs to a specified function class. When the univariate ground-truth function $f(x)$ is not monotonic, a simple modeling strategy is to specify the number of \textit{stationary points}, i.e., points at which $f'(x)=0$. Joint estimation and inference of these locations are often of substantive interest. One motivating example is the analysis of electroencephalography (EEG) data to extract event-related potentials (ERPs), where peaks and troughs correspond to distinct stages of brain processing \citep{luck2012event}. However, the literature on inference for the number and location of local extrema remains relatively limited. \cite{davies2001local} developed two algorithms to find the best-fitting curve with the minimal number of local extrema. While their methods identify potential stationary points, they are prone to false positives. In addition, the lack of a likelihood-based formulation precludes principled uncertainty quantification.  

An alternative line of work first estimates the regression function using a flexible nonparametric smoother and then applies a hypothesis-testing procedure to identify potential extrema. \cite{song2006nonparametric} exploit the asymptotic distribution of the derivative of a local quadratic estimator to test for the existence of stationary points at selected locations. \cite{schwartzman2011multiple} and \cite{cheng2019multiple} adopt similar strategies for both one-dimensional and random-field settings, identifying local maxima of the smoothed function and applying multiple-testing corrections to control the false discovery rate or family-wise error rate.

\cite{yu2023bayesian} proposed a semiparametric Bayesian approach that models the unknown function as a Gaussian process, leveraging the joint Gaussianity of the process and its derivative to sample likely zero-derivative locations. A notable advantage of their approach is that the analyst needs only to specify a prior for a single stationary point $t$, while the posterior becomes multimodal with modes centered near true stationary points. Extending this framework, \cite{li2024semiparametric} derived a closed-form posterior for the scalar variable $t$ and demonstrated that it asymptotically converges to a mixture of normal distributions centered at the true stationary points. However, since $t$ is scalar-valued, this approach cannot be directly extended to perform joint inference over multiple stationary points.
Furthermore, it often places excessive posterior mass near the boundaries, necessitating additional post-processing to obtain stable point estimates and credible intervals.

\cite{wheeler2017bayesian} developed a Bayesian model that allows the analyst to constrain the number of stationary points and infer their locations jointly. In their framework, the true function $f(x)$ is modeled as
\begin{align*}
    f(x) &=  \sum_{k = 0}^{K + j - 1}\beta_{k} B^{*}_{j, k}(x), \\
    \text{where} \quad 
    B^{*}_{j, k}(x) &=  C \int_{0}^{x} \left[ \prod_{i = 1}^{M} (\xi - \alpha_{i}) \right] B_{j, k}(\xi) \, d\xi.
\end{align*}
Here, $B_{j, k}$ denotes the standard degree-$j$ B-spline basis constructed with knot set $\Xi$, $\beta_{k} \ge 0$, and $C$ is a fixed integer. The authors showed that $f(x)$ possesses exactly $M$ stationary points located at $\alpha_{1}, \ldots, \alpha_{M}$, treated as model parameters. 
Posterior sampling is performed via a reversible-jump Markov chain Monte Carlo (MCMC) algorithm that treats the knot set as a trans-dimensional parameter. However, the model suffers from non-identifiability: any permutation of $\alpha_{1}, \ldots, \alpha_{M}$ yields the same function. Consequently, the marginal posteriors for the stationary points are identical, and inference requires post-processing to relabel samples.

In a nonparametric density estimation context, \citet{dasgupta2020two} proposed a two-step procedure for estimating univariate densities supported on an interval $[a,b]$. They first obtain an initial estimate $f_{\theta}$ from a simple parametric family $\mathcal{F}_{\theta}$, termed the \textit{template density}. Next, they estimate a monotone transformation $\gamma$ by solving
\begin{align*}
   \hat{\gamma}= \arg\max_{\gamma \in \Gamma} \sum_{i=1}^{n} \log\!\big[f_{\theta}(\gamma(x_{i}))\,\gamma'(x_{i})\big],
\end{align*}
where $\Gamma$ is the set of \textit{diffeomorphisms} on $[a,b]$,
\begin{equation}
\label{eq:diffeomorphism}
    \Gamma = \big\{\gamma : \gamma(a)=a,\; \gamma(b)=b,\; \gamma' > 0\big\}.
\end{equation}
The final estimate is $\hat{f}(x) = f_{\theta}(\hat{\gamma}(x))\,\hat{\gamma}'(x)$. This approach parallels normalizing flows in machine learning \citep{papamakarios2021normalizing}, where simple base densities are transformed to match complex targets. Its theoretical foundation lies in the transitive group action of $\Gamma$ on $\mathcal{F}$, the space of all densities on $[a,b]$: for any $f \in \mathcal{F}$ and template $f_{\theta} \in \mathcal{F}_{\theta} \subset \mathcal{F}$, there exists a unique $\gamma \in \Gamma$ such that $(f_{\theta} \circ \gamma)\,\gamma' = f$. Thus, $\gamma$ corrects model misspecification arising from a simplistic template. Moreover, \citet{dasgupta2020two} established a bijection between $\Gamma$ and the coefficients of an infinite orthogonal basis with $L^{2}$ norm less than $\pi$, allowing practical estimation via truncated bases and constrained optimization.

Building on this framework, \cite{dasgupta2021modality} applied the diffeomorphic transformation approach to modality-constrained density estimation, where the template function is a spline with a prespecified number of local maxima. The template and diffeomorphism are jointly estimated by solving
\begin{align*}
    (\hat{f}, \hat{\gamma}) = \underset{f \in \mathcal{T}, \, \gamma \in \Gamma}{\arg\max} \sum_{i=1}^{n} \log\left[ \frac{f(\gamma(x_{i}))}{\int_{a}^{b}  f(\gamma(x))\,dx } \right],
\end{align*}
where $\mathcal{T}$ denotes a spline family with stationary points at prespecified locations.

We extend this diffeomorphic framework to regression by modeling the true function as $f = g \circ \gamma$, where $g$ belongs to a simple function class $\mathcal{F}_{M}$ and $\gamma$ is a diffeomorphism. When the number and location of stationary points are of primary interest, this approach offers two key advantages. First, similar to \cite{wheeler2017bayesian}, we ensure that the estimated regression function has exactly $M$ stationary points within a specified interval by requiring that $\mathcal{F}_{M}$ consist of functions with $M$ stationary points in a simple parametric form such as piecewise interpolations alternating between increasing and decreasing segments. Second, for the estimated diffeomorphism $\hat{\gamma}$, if the stationary points of the template function $g$ are $b_{1}, \ldots, b_{M}$, then the stationary points of $f$ are estimated as $\left(\hat{\gamma}^{-1}(b_{1}), \ldots, \hat{\gamma}^{-1}(b_{M})\right)$. Hence, unlike \cite{li2024semiparametric}, who treats the stationary point as a scalar parameter, our method facilitates valid multidimensional joint inference on their locations. Moreover, the model is identifiable and requires no post-processing, as $\hat{\gamma}^{-1}$ preserves ordering: if $b_{1} < \cdots < b_{M}$, then $\hat{\gamma}^{-1}(b_{1}) < \cdots < \hat{\gamma}^{-1}(b_{M})$.

As there exists a one-to-one correspondence between the set of diffeomorphisms and the coefficients of an infinite orthogonal basis, high-dimensional statistical theory can be employed to derive non-asymptotic error bounds of the estimator. In particular, assuming Gaussian errors, we invoke the theory of \cite{Spokoiny2012_ParametricEstimation} to obtain non-asymptotic large deviation and confidence bounds for the maximum likelihood estimator (MLE) of the coefficients parameterizing the diffeomorphism. A delta method type argument then yields the joint asymptotic distribution of the stationary point estimates. An analogous analysis yields error bounds for a Bernstein–von Mises theorem in the Bayesian setting.

In this article, we investigate this template-diffeomorphism framework for both shape-constrained regression and the joint inference of stationary point locations. Section~\ref{sec:spcr} introduces the model in detail, explains how to construct simple template classes $\mathcal{F}_{M}$, and provides theoretical justification for the formulation. Section~\ref{sec:theory} derives finite-sample confidence bounds for the MLEs of both the model parameters and stationary points; Bayesian extensions are presented in Appendix \ref{sec:bayesian_case}. Section~\ref{sec:sim} reports simulation studies comparing our method with that of \cite{li2024semiparametric}, and Section~\ref{sec:app_ds} illustrates the approach using EEG data. 
Section~\ref{sec:conclusion} closes the paper with a summary of the main findings and a discussion of their broader implications.
Most proofs are deferred to the Supplementary Materials. An \texttt{R} package \texttt{scrdiff}, which implements the proposed method, is publicly available at \url{https://github.com/mprice747/scrdiff}. All code and scripts used for the simulation studies and real-data analyses are provided at \url{https://github.com/mprice747/scrdiff_analysis}.

\section{Stationary Point Constrained Regression}
\label{sec:spcr}
This section develops the framework for stationary point constrained regression. We first formalize the model structure and assumptions in Section~\ref{sec:The_Model}. Next, in Section~\ref{sec:Parameterizing}, we introduce a flexible parameterization for the stationary point location parameter $\gamma$, which enables efficient and interpretable inference. Finally, Section~\ref{sec:Joint_Estimation} presents the joint estimation and inference procedure for both the function $g$ and the stationary point parameter $\gamma$.

\subsection{Model Setup}
\label{sec:The_Model}
 Consider the input data $X_{i} \stackrel{i.i.d.}{\sim} \mathbb{P}_{X}$, where $\mathbb{P}_{X}$ is a probability measure on the interval $[0, 1]$. We specify that \begin{equation}
    Y_{i} = f(X_{i}) + \epsilon_{i}, \quad\epsilon_{i} \stackrel{i.i.d.}{\sim} N(0, \sigma^2),
\end{equation}
where $f(\cdot)$ is assumed to belong to the function class $\mathcal{F}_M$ in Definition \ref{def:F_M}.
\begin{definition}
\label{def:F_M}
Let $\mathcal{F}_M$ denote the set of all functions $ f : [0, 1] \to \mathbb{R} $ satisfying the following properties:
\begin{itemize}
    \item $ f $ is differentiable on $[0, 1]$.
    \item $ f $ has exactly $ M\in \mathbb{N}^{+} $ stationary points at $ b^*_1, b^*_2, \ldots, b^*_M $, where $0=b^*_{0} < b^*_{1} < ... <b^*_{M + 1} = 1$, i.e.,  $ f'(b^*_k) = 0 $ for $ k = 1, \ldots, M $.
    \item Each stationary point $b_{k}^*$ is a strict local extremum; i.e., the function switches from increasing to decreasing or vice versa. 
\end{itemize}
\end{definition}

 Now, we model $f$ as the composition of two functions $f = g \circ \gamma$. We assume $g$, in which we call the \textit{template function}, belongs to a restricted subclass $\mathcal{G}_{M} \subset \mathcal{F}_{M}$. Next, we let $\gamma \in \Gamma_{c}$, defined as the set of all bijective, continuous functions from $[0, 1]$ to $[0, 1]$, i.e., \begin{equation}
     \Gamma_{c} = \Big\{ \gamma : [0, 1] \to [0, 1] \; : \;\gamma \text{ is bijective and continuous} \Big\}.
 \end{equation}
The theoretical justification for this compositional formulation follows from Theorem~\ref{thm:thm_1} below. We define the following two sets of vectors: \begin{equation}
      \Lambda_{M}^{+} = \Big\{ (\lambda_{0}, \lambda_{1}, \ldots, \lambda_{M + 1}) \ :  \lambda_{2j} < \lambda_{2j + 1}, \lambda_{2j + 1} > \lambda_{2j + 2} \  \text{for}\ 1 \leq 2j + 1 \leq M + 1\Big\},
 \end{equation}
     \begin{equation}
         \Lambda_{M}^{-} = \Big\{ (\lambda_{0}, \lambda_{1}, \ldots, \lambda_{M + 1}) \ :   \lambda_{2j} > \lambda_{2j + 1}, \lambda_{2j + 1} < \lambda_{2j + 2} \  \text{for}\ 1 \leq  2j + 1 \leq M + 1\Big\}.
     \end{equation}
Put simply, $\Lambda_{M}^{+}$  is the set of valid outputs for $f$ at the points $b^*_0, b^*_{1}, \ldots,  b^*_{M+1}$ if the function is increasing on  $[b^*_0, b^*_{1}]$, decreasing on $[b^*_{1}, b^*_{2}]$, etc. $\Lambda_{M}^{-}$ is the set of valid outputs for $f$ at the points $b^*_0, b^*_{1}, \ldots,  b^*_{M+1}$, if the function is decreasing from $[b^*_0, b^*_{1}]$, increasing from $[b^*_{1}, b^*_{2}]$, etc.
 \begin{theorem}
 \label{thm:thm_1} Let $\lambda = (\lambda_{0}, \lambda_{1}, \ldots, \lambda_{M+1}) \in \Lambda_{M}^{+} \cup \Lambda_{M}^{-}$. Define the subset $\mathcal{F}_{\lambda, M} \subset \mathcal{F}_{M}$ such that $f \in \mathcal{F}_{\lambda, M}$ if $f(b^*_{k}) = \lambda_{k}$ for $k \in \{0, 1, \ldots, M, M+1\}$, where $\{b^*_1,\ldots, b^*_{M}\}$ are the stationary points of $f$. For any $f, g \in \mathcal{F}_{\lambda, M}$, there exists a $\gamma \in \Gamma_{c}$ satisfying $f = g \circ \gamma$.
\end{theorem}

The proof is provided in the Supplementary Materials . The strength of this approach lies in the ability to restrict $g$ to a simple subclass, such as the set of cubic splines with stationary points $b_{k}^* = \frac{k}{M + 1}$, while jointly learning $g \circ \gamma$ to approximate any function in $\mathcal{F}_{M}$ accurately. In addition, if the stationary points of $g$ are known $(b_{k}^*)$ and $\gamma$ is a differentiable and strictly increasing function, one can infer the stationary points of $f$ by calculating $\gamma^{-1}(b_{k}^*)$, as  $$f'(\gamma^{-1}(b_{k}^*)) = g'(\gamma \circ \gamma^{-1}(b_{k}^*)) \gamma'(\gamma^{-1}(b_{k}^*))=g'(b_{k}^*) \gamma'(\gamma^{-1}(b_{k}^*)) = 0.$$

\begin{figure}[htbp]
  \centering
  \begin{subfigure}[b]{0.48\textwidth}
    \includegraphics[width=\textwidth]{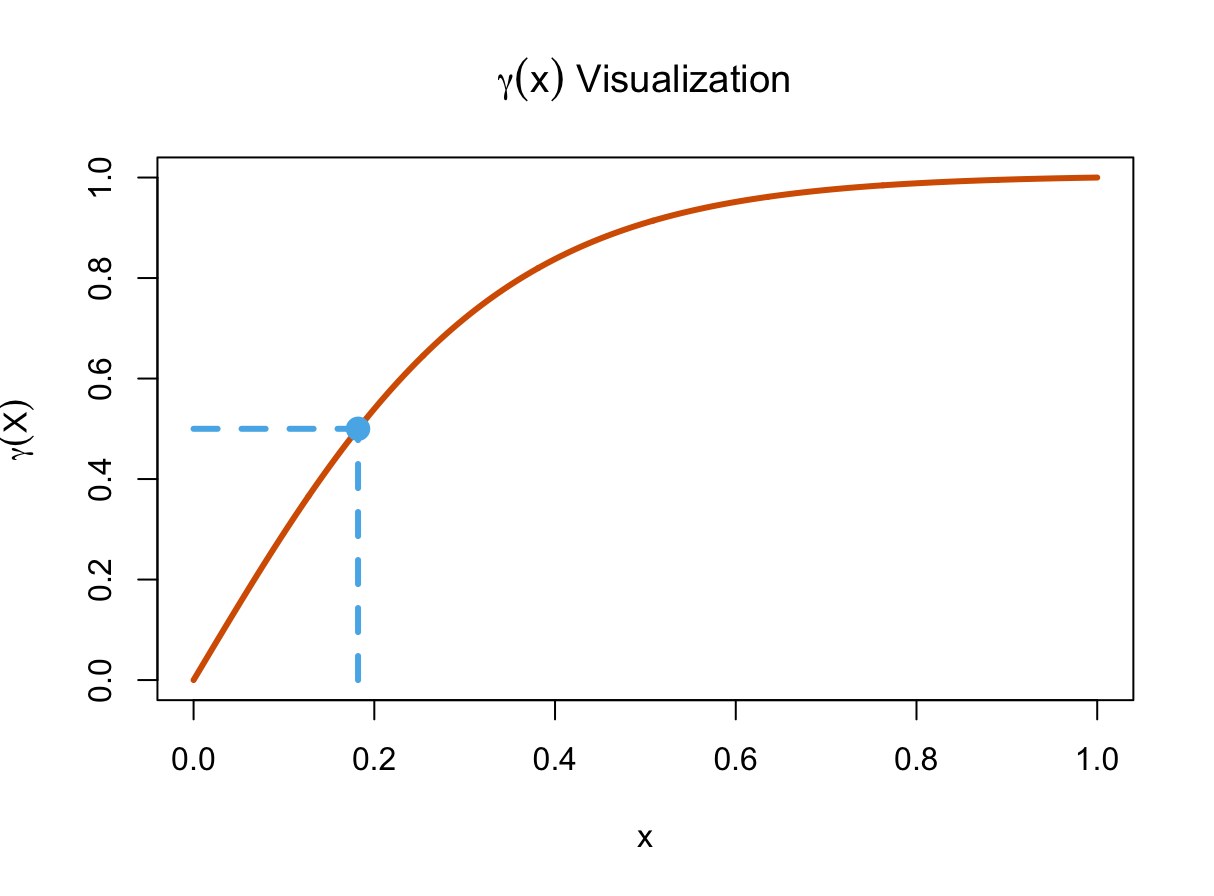}
  \end{subfigure}
  \begin{subfigure}[b]{0.48\textwidth}
    \includegraphics[width=\textwidth]{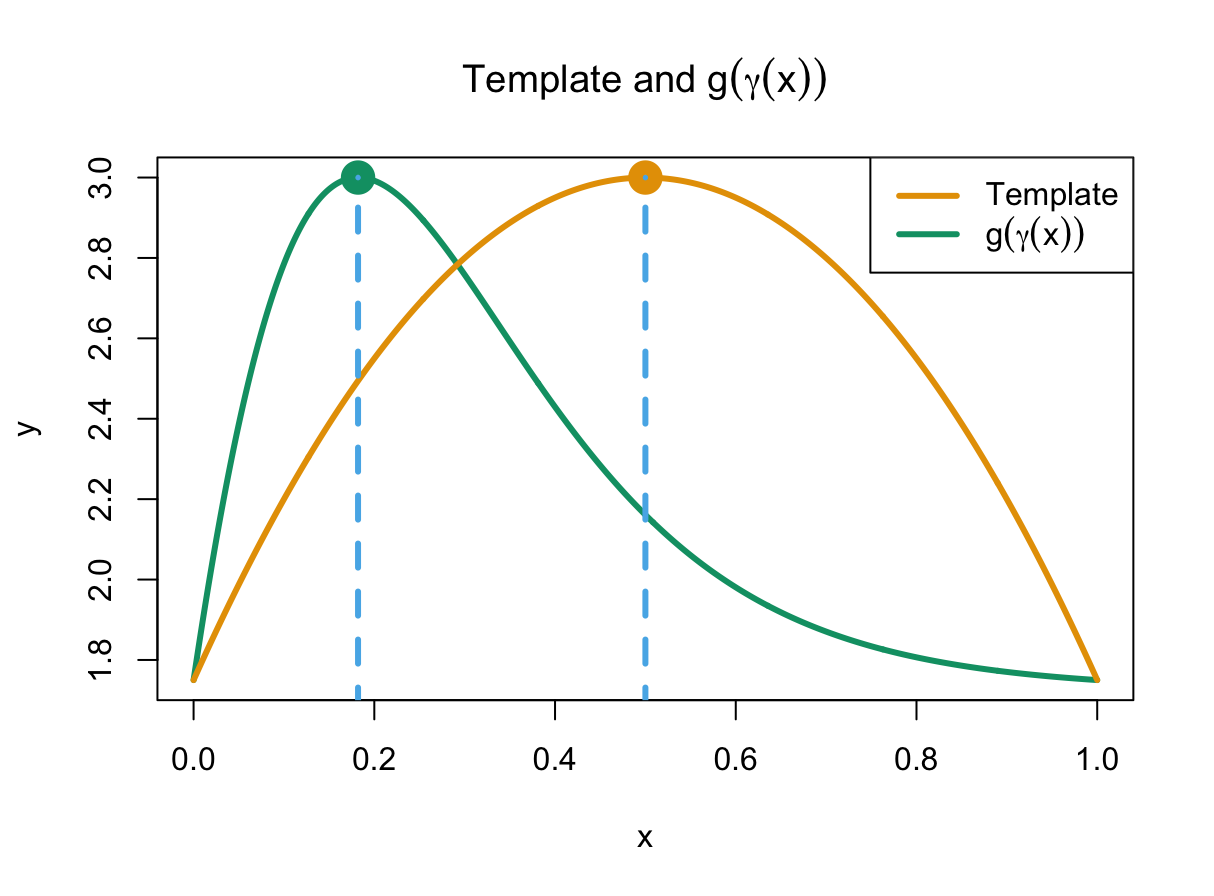}
  \end{subfigure}
  \caption{Visualization of the composition approach. Notice that the stationary points of the composition and template correspond to a specific point $(x, \gamma(x))$ on the graph of $\gamma$.}
\end{figure}

\subsection{Parameterizing $\gamma$}
\label{sec:Parameterizing}
Despite being mathematically pleasing, learning a best fitting $\gamma \in \Gamma_{c}$ proves to be challenging because $\Gamma_{c}$ is a nonlinear manifold  and there exists no low-dimensional parameterization of the set. However, this composition approach proves viable if we assume a suitable $\gamma$ belongs to the set of \textit{diffeomorphisms} $\Gamma$, the set of differentiable mappings on $[0, 1]$ whose derivatives are strictly positive, 
i.e., \begin{equation}
    \Gamma = \Big\{ \gamma : [0, 1] \to [0, 1] \; : \; \gamma(0) = 0, \gamma(1) = 1,  \gamma'(x) > 0 \text{ for all } x \in [0, 1] \Big\}.
\end{equation}  Although $\Gamma$ is not a vector space, \citet{dasgupta2020two} showed that there exists a bijective mapping from $\Gamma$ to a set $V$, where $V$ is defined as the collection of functions in $L^{2}([0, 1])$ that have $L^{2}$ norm less than $\pi$ and are orthogonal to all constant functions on $[0, 1]$. Put mathematically,  \begin{equation}
    V = \left\{ v \in L^{2}([0, 1]) : \left\lVert v \right\lVert < \pi,\, \int_{0}^{1}v(t)dt = 0 \right\}. 
\end{equation}
This bijection from $v\in V$ to $\gamma \in \Gamma$ first requires mapping $v$ to the unit Hilbert sphere, defined as \begin{equation}
    \mathbb{S}_{\infty} = \left\{ q : [0, 1] \rightarrow \mathbb{R} : \int_{0}^{1} q^{2}(t)dt = 1\right\}, 
\end{equation}via the bijective exponential mapping $\exp_{1}(v)$, defined as below:
\begin{equation}
    \exp_{1}(v) : V \rightarrow \mathbb{S}_{\infty},  \ \ \ \  \exp_{1}(v)= \cos(\left\lVert v \right\lVert) + \frac{\sin(\left\lVert v \right\lVert)}{\left\lVert v \right\lVert}v.
\end{equation}   Finally, we obtain $\gamma(t)$ by setting $\gamma(t) = \int_{0}^{t} q^{2}(s) ds$, where $q(s)$ is the output of the exponential map. As $\int_{0}^{0} q^{2}(t)dt = 0$, $\int_{0}^{1} q^{2}(t)dt = 1$, and the integral from $0$ to $t$ of a positive function is a strictly increasing function of $t$, $v$ has been transformed to a diffeomorphism. In fact, any diffeomorphism $\gamma$ can be mapped to a $q \in \mathbb{S}_{\infty}$ by setting $q(s) = \sqrt{\gamma'(s)}$. Subsequently,  $q(s)$ can be transformed to a $v \in V$ via the inverse of the exponential map \citep{dasgupta2020two}.   Conveniently, we can parameterize $V$, and subsequently $\Gamma$,  using any infinite orthogonal basis of $L^{2}([0, 1], \mathbb{R})$, $\sum\limits_{j = 1}^{\infty} \beta_{j} \psi_{j}(t)$, whose elements $\psi_{j}$ are orthogonal to $v(t) = 1$ and coefficients $\boldsymbol{\boldsymbol{\beta}} = (\beta_{1}, \beta_{2}, ..., )$ have Euclidean norm less than $\pi$. Choices for $\{\psi_{j}\}_{j = 1}^{\infty}$ include the Fourier and cosine basis without the constant term. Of course, modeling $V$ with an infinite orthogonal basis is impractical. However, one can truncate this basis to a finite dimension $p$ and can model all diffeomorphisms within this $p$ dimensional subset $\Gamma_{p}$. As $p\rightarrow \infty$, the subset $\Gamma_{p}$ converges to the full set $\Gamma$. 

Putting everything together, assume a suitable diffeomorphism lies in $\Gamma_{p}$. Thus, letting $\boldsymbol{\beta} \in \boldsymbol{\Theta} = \{ \boldsymbol\beta \in \mathbb{R}^{p} : \|\boldsymbol{\beta}\| < \pi \}$ and setting $v_{\boldsymbol{\beta}}(t) = \sum\limits_{j = 1}^{p} \beta_{j} \psi_{j}(t)$, the corresponding diffeomorphism $\gamma_{\boldsymbol{\beta}}$ has the following form:
\begin{align*}
\gamma_{\boldsymbol{\beta}}(t)=&\int_{0}^{t} \left(\cos(\left\lVert v_{\boldsymbol{\beta}} \right\rVert) + \frac{\sin(\left\lVert v_{\boldsymbol{\beta}} \right\rVert)}{\left\lVert v_{\boldsymbol{\beta}} \right\rVert}v_{\boldsymbol{\beta}}(s)\right)^{2} ds\\   \\
=&\int_{0}^{t} \left[\cos^2(\left\lVert \boldsymbol{\beta} \right\rVert) + 2\frac{\cos(\left\lVert \boldsymbol{\beta} \right\rVert)\sin(\left\lVert \boldsymbol{\beta} \right\rVert)}{\left\lVert \boldsymbol{\beta} \right\rVert}v_{\boldsymbol{\beta}}(s) + \frac{\sin^2(\left\lVert \boldsymbol{\beta} \right\rVert)}{\left\lVert \boldsymbol{\beta} \right\rVert^2}v^{2}_{\boldsymbol{\beta}}(s)\right]  ds \\ \\
      = & \ \cos^2(\left\lVert \boldsymbol{\beta} \right\rVert)t + 2\frac{\cos(\left\lVert \boldsymbol{\beta} \right\rVert)\sin(\left\lVert \boldsymbol{\beta} \right\rVert)}{\left\lVert \boldsymbol{\beta} \right\rVert}\int_{0}^{t}v_{\boldsymbol{\beta}}(s)ds + \frac{\sin^2(\left\lVert \boldsymbol{\beta} \right\rVert)}{\left\lVert \boldsymbol{\beta} \right\rVert^2}\int_{0}^{t}v^{2}_{\boldsymbol{\beta}}(s)ds. 
\end{align*} 
Here, 
 $\left\lVert v_{\boldsymbol{\beta}} \right\rVert = \left\lVert \boldsymbol{\beta} \right\rVert$, owing to the orthogonality of the $\psi_j$’s. For the rest of the article, we represent $V$ using the cosine basis with norm less than $\pi$, i.e. \begin{equation}
     V = \text{span} \left\{ \sqrt{2}\cos(k \pi t) : k \in \mathbb{N}^+  \right\} \cap \left\{ v \in L^{2}([0, 1]) : \|v\| < \pi \right\}.
 \end{equation}
 This choice is motivated by the fact that both $\int_{0}^{t} v_{\boldsymbol{\beta}}(s)ds$ and $\int_{0}^{t} v_{\boldsymbol{\beta}}^{2}(s)ds$, and consequently $\gamma_{\boldsymbol{\beta}}$, admit closed-form expressions for all $t \in [0,1]$. In addition, the cosine basis ensures that the gradient of the log-likelihood with respect to $\boldsymbol{\beta}$ remains bounded, which will be useful for our theoretical results.

 \subsection{Joint Estimation and Inference of $g$ and $\gamma$}
\label{sec:Joint_Estimation}

The next challenge is to choose a suitable subclass $\mathcal{G}_{M} \subset \mathcal{F}_{M}$ for the template function. Ideally, each $g \in \mathcal{G}_{M}$ should have the following three properties: \begin{enumerate}
    \item Evaluating $g$ at any $t \in [0, 1]$ should be computationally efficient. 
    \item $g$'s $M$ stationary points are assumed to be explicitly available, either given a priori or obtainable with negligible computational effort. 
    \item $g \circ \gamma_{\boldsymbol{\beta}} \in C^{j}([0, 1])$ for a desired $j$, where $C^j([0, 1])$ denotes the set of $j-$times continuously differentiable functions on $[0, 1]$. 
\end{enumerate}
 
However, estimating functions from  $\mathcal{G}_{M}$, where each element satisfies these properties, is challenging. One strategy is to parameterize each $g_{\lambda} \in \mathcal{G}_{M}$ with a height vector $\lambda \in \Lambda_{M}^{+} \ \text{or} \ \Lambda_{M}^{-}$. Here, $\lambda_{k}$ will be the value of $g$ at prespecified points $b_{0}, b_{1}, ..., b_{M + 1}$. We can then construct $g_{\lambda}$ as a Hermite cubic interpolant using the pairs $(b_{k}, \lambda_{k})$ as interpolation nodes, as specified below:
\begin{equation}
    g_{\lambda}(x) = (2t^{3} - 3t^{2} + 1)\lambda_{k} + (-2t^{3} + 3t^{2})\lambda_{k + 1},   
\end{equation}
with $t = \frac{x - b_{k}}{b_{k + 1} - b_{k}}$ and $x \in [b_{k}, b_{k + 1}]$. 
This formulation is justified by the following proposition.
\begin{proposition}
\label{thm:stationary_points}
    Let $\lambda \in \Lambda_{M}^{+} \ \text{or} \ \Lambda_{M}^{-}$. If $g_{\lambda}$ is a Hermite cubic interpolation with $(b_{k}, \lambda_{k})$ as the nodes, then $g_{\lambda}$ will have exactly M stationary points located at $b_{1}, \ldots, b_{M}$. 
\end{proposition}

\begin{proof}
    The derivative of the interpolation will equal \begin{align*}
g_{\lambda}'(x) = \frac{6}{b_{k}- b_{k + 1}}(t^{2} - t)(\lambda_{k} - \lambda_{k + 1}),
\end{align*} 
with $t = \frac{x - b_{k}}{b_{k + 1} - b_{k}}$ and $x \in [b_{k}, b_{k + 1}]$. 
As $t^2 = t$ if $t = 0 \ \text{or} \ 1$, $g_{\lambda}'(b_{k}) = 0$ for $k = 1, ..., M$. In addition, as $t^{2} < t$ for $t \in (0, 1)$, $g'_{\lambda}(x)$ will be strictly decreasing if $\lambda_{k} > \lambda_{k + 1}$ and $g'_{\lambda}(x)$ will be strictly increasing if $\lambda_{k} < \lambda_{k + 1}$. Thus, $g_{\lambda}$ will have $M$ stationary points located at $b_k$ for $k = 1, ..., M$.
\end{proof}

Proposition \ref{thm:stationary_points} guarantees $f = g_{\lambda} \circ \gamma$ has exactly $M$ stationary points,  making $\gamma_{\boldsymbol{\beta}}^{-1}(b_{k})$ a valid estimate for the $k$-th stationary point. However, the estimated $g_{\lambda}$, and hence the estimated $f$ generally does not have continuous derivatives. If higher order smoothness is desired, this height vector approach can be used together with a smoother interpolation method like B-splines via the following process. First, define a sequence of knots $\Xi = (\xi_{1}, \xi_{2}, ..., \xi_{K})$ where $0 < \xi_{1} < \xi_{2} < ... < \xi_{K}<1$. Next, to enforce the monotonicity for the template on each $[b_{k}, b_{k + 1}]$ interval, generate additional data points $(b_{k, \ell}, \lambda_{k, \ell})$ similar to that of a linear interpolation:
\begin{equation}
    b_{k, \ell} = b_{k} + \delta_{\ell}(b_{k + 1} - b_{k}) \quad\text{and}\quad \lambda_{k, \ell} = \lambda_{k} + \delta_{\ell}(\lambda_{k + 1} - \lambda_{k}),
\end{equation}  
for $\delta_\ell \in [0, 1]$. Finally, let $\Phi(x)$ denote the $m$ degree B-spline basis evaluated at $x$, with knot sequence $\Xi$. 

We can now solve the following constrained least-squares problem, which admits a unique analytical solution provided that a sufficient number of $(b_{k, \ell}, \lambda_{k, \ell})$ pairs are available:
\begin{equation}
    \hat{\boldsymbol{\theta}}_{\lambda} = \underset{\boldsymbol{\theta} \in \mathbb{R}^{K + m + 1} }  {\arg\min} \sum\limits_{k = 0}^{M + 1}\sum\limits_{\ell} (\lambda_{k, \ell} - \boldsymbol{\theta}^T\Phi(b_{k, \ell}))^{2} + \alpha P(\boldsymbol{\theta}),
\end{equation} 
such that 
$$\begin{cases}
        \Phi(b_{k}) = \lambda_{k} \qquad  \text{for} \ k = 0, 1, ..., M + 1, \\
        \Phi'(b_{k}) = 0 \qquad \text{for} \ k = 1, ..., M .
    \end{cases}$$
Here, $P(\boldsymbol{\theta})$ denotes a curvature penalty, such as the second derivative penalty $(\int\limits_{0}^{1} (\boldsymbol{\theta}^T\Phi''(x))^{2}dx) $ or a discrete approximation  like the P-spline penalty term \citep{eilers1996flexible}. The function $g_{\lambda}(x) = (\hat{\boldsymbol{\theta}}_{\lambda})^T \Phi(x)$ serves as the template, ensuring at least $M$ stationary points at $b_k$, for $k = 1, \ldots, M$. A key limitation of this approach is that the estimate $g_{\lambda}(x)$ may have additional points where its derivative is zero, potentially making $\gamma_{\boldsymbol{\beta}}^{-1}(b_k)$ an invalid stationary point estimate. To address this, we propose using a large number of knots ($K \approx 100$), substantial generated data, and a high smoothing parameter ($\alpha > 10^5$) to ensure flexibility while penalizing overly oscillatory templates. However, despite guaranteed smoothness, this method is computationally intensive. Figure \ref{fig:Template} compares the Hermite cubic and B-spline templates for identical $(b_k, \lambda_k)$ pairs.

\begin{figure}[!h]
  \centering
  \begin{subfigure}[b]{0.48\textwidth}
    \includegraphics[width=\textwidth]{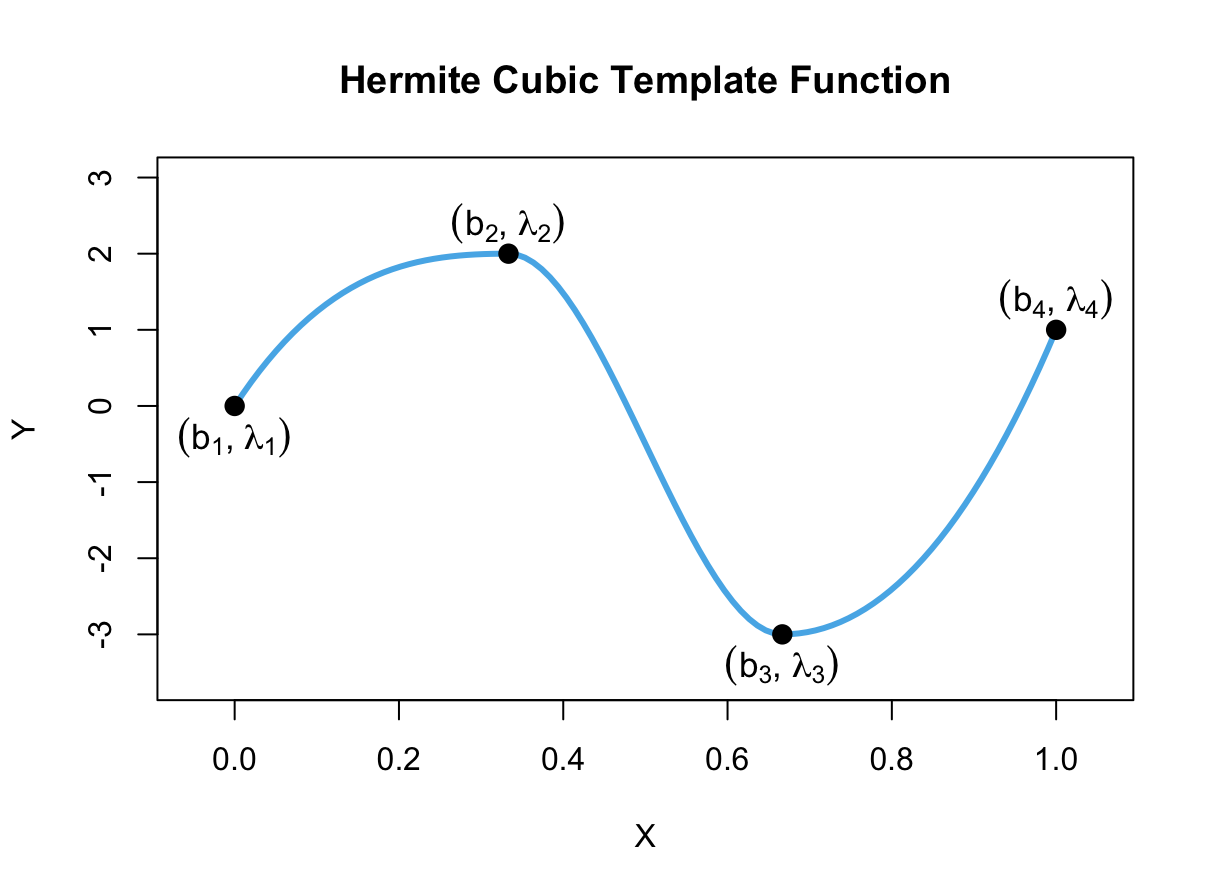}
  \end{subfigure}
  \begin{subfigure}[b]{0.48\textwidth}
    \includegraphics[width=\textwidth]{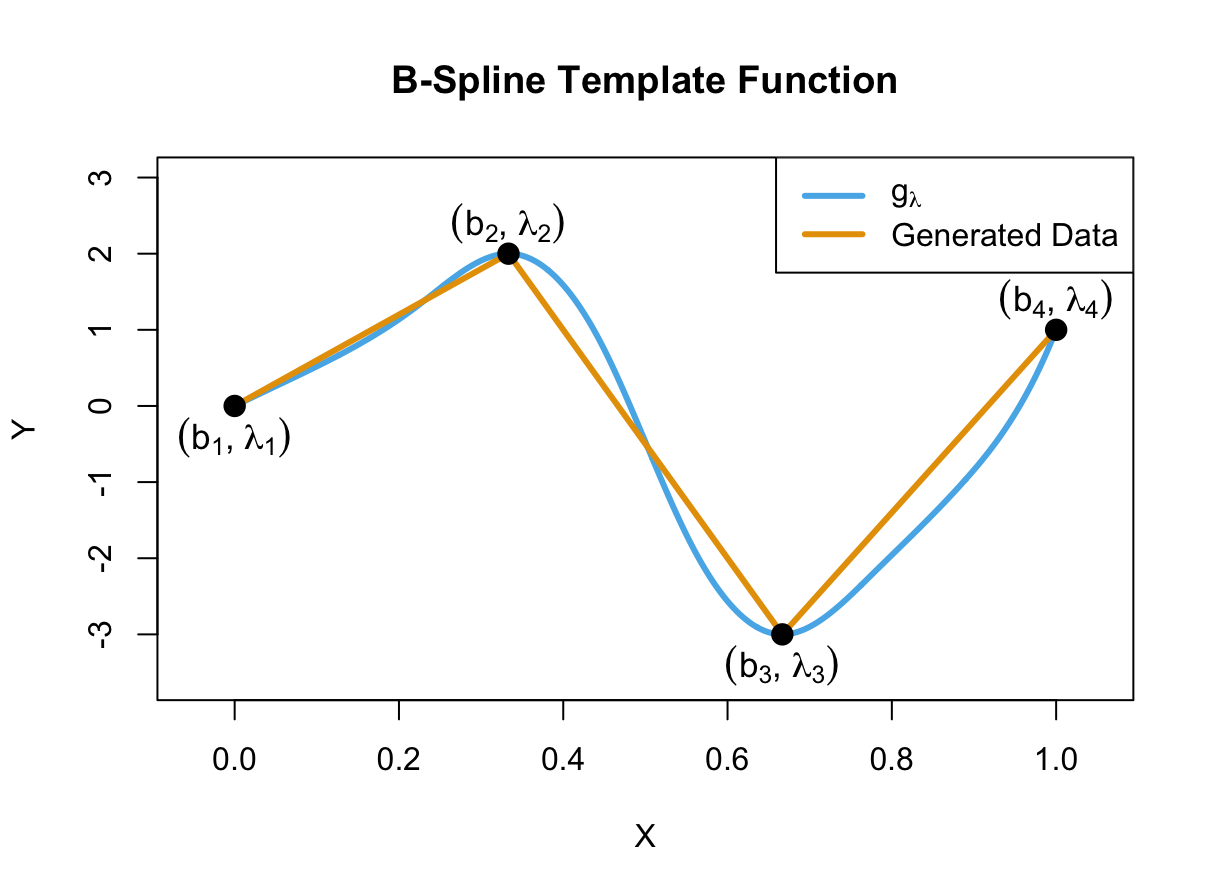}
  \end{subfigure}
  \caption{Template functions for the same set of $(b_{k}, \lambda_{k})$ parameters}
\label{fig:Template}
  \end{figure}

A less computationally intensive approach that ensures smoothness is to define the template function $g$ as a polynomial of degree $M + 1$. The stationary points of $g$, and thus of $f$, can be determined with minimal computational cost by computing the eigenvalues of the companion matrix of its derivative \citep{edelman1995polynomial}. However, this approach only guarantees that $g$ has at most $M$ stationary points in $[0, 1]$, not exactly $M$. To apply the height vector and interpolation strategy, one must first specify whether $\lambda \in \Lambda^+_M$ or $\lambda \in \Lambda^-_M$. Let $\Lambda$ denote the chosen height vector set. Assuming independent normal errors, maximum likelihood estimates $\boldsymbol{\hat{\beta}}$ and $\hat{\lambda}$ are obtained by solving
\begin{equation*}
    \underset{\lambda \in \Lambda,\  \boldsymbol{\beta} \in \boldsymbol{\Theta} } {\arg\min} \sum\limits_{i = 1}^{n}(y_{i} - g_{\lambda}\circ \gamma_{\boldsymbol{\beta}}(x_{i}))^{2}.
\end{equation*}
This is a nonstandard constrained optimization problem due to the irregular geometry of $\Lambda^{+}_{M}$ and $\Lambda^{-}_{M}$. However, the constraint becomes more manageable if we optimize over the following set: \begin{equation}
    \mathcal{L} = \Big\{(\ell_0, \ell_1, \dots, \ell_{M+1}) \in \mathbb{R} \times (0,\infty)^{M+1}
\; :\; \ell_{0} = \lambda_{0},\; \ell_{k} = |\lambda_{k} - \lambda_{k -1}|\Big \}.
\end{equation}
Here, $\lambda_{k}$ are easily reconstructed by performing a cumulative sum. We further remove the positive constraint on $(\ell_{1}, ..., \ell_{M + 1})$ by expressing $\ell_{k} = \exp(l_{k})$, where $l_{k} \in \mathbb{R}$. In addition, the norm constraint ($\|\boldsymbol{\beta}\| < \pi$) can be removed from $\boldsymbol{\beta}$ by setting $\boldsymbol{\beta} = \pi \frac{y}{\sqrt{1+ \left\lVert y\right\lVert ^2}}$, where $y \in \mathbb{R}^{p}$. Therefore, we can apply standard unconstrained optimization techniques like the Broyden–Fletcher–Goldfarb–Shanno (BFGS) quasi-Newton algorithm \citep{nocedal2006numerical} to obtain maximum likelihood estimates. Given the non-convex nature of the loss function, we recommend applying the BFGS algorithm from multiple initial points and selecting the solution with the lowest loss.

To estimate uncertainty, methods such as residual or block bootstrap \citep{kreiss2012bootstrap} can be used to construct approximate confidence intervals for parameters and stationary points. For Bayesian inference, a standard prior, such as a multivariate Gaussian, can be placed on the unconstrained parameters. However, the posterior is often highly multimodal, complicating efficient exploration of the space. To address this, we suggest running BFGS multiple times to identify posterior modes and compute the inverse Hessian at each mode. The inverse Hessians, appropriately scaled,  serve as proposal covariance matrices for multiple MCMC chains. Alternatively, approximate resampling methods like the weighted likelihood bootstrap \citep{newton1994approximate} have proven effective for multimodal posterior estimation. For more complex cases, specialized techniques for multimodal posteriors, such as MCMC tempering \citep{geyer1995annealing} or variational boosting \citep{miller2017variational}, can be employed.

\section{Non-Asymptotic Analysis}
\label{sec:theory}

A key advantage of our model is its explicit parametric form. Assuming the log-likelihood is twice continuously differentiable and that the dimension of the diffeomorphism weight parameter $\boldsymbol{\beta}$ is fixed and known a priori, standard asymptotic theory guarantees consistency and asymptotic normality of the MLE and the posterior for $\boldsymbol{\beta}$ \citep{van2000asymptotic}. Moreover, if the mapping $\Phi$ from $\boldsymbol{\beta}$ to the stationary points is first-order differentiable at the ground truth, the multivariate delta method ensures joint asymptotic normality and consistency of the stationary point estimates. However, this asymptotic analysis is limited for two reasons. First, representing the full set of diffeomorphisms $\Gamma$ requires $\boldsymbol{\beta}$ to be infinite-dimensional, so restricting $\boldsymbol{\beta}$ to a finite dimension $p$ is primarily a matter of analytical or computational convenience. Second, asymptotic theory does not quantify the accuracy of these approximations at finite sample sizes.

In this section, we derive finite-sample confidence bounds for the MLE of $\boldsymbol{\beta}$ and its sampling distribution. Let $\boldsymbol{\hat{\beta}} \in \mathbb{R}^{p}$ denote the MLE based on a sample of size $n$. We first bound the probability that $\boldsymbol{\hat{\beta}}$  lies outside a neighborhood of the true parameter $\boldsymbol{\beta^*} \in \mathbb{R}^{p}$. This quantifies confidence that the estimator is close to the truth for a given $n$. We then bound the distance between scaled versions of the asymptotic normal statistic $\sqrt{n\boldsymbol{v}^{2}_{0}}(\boldsymbol{\hat{\beta}} - \boldsymbol{\beta^*})$ and the normalized score $(n\boldsymbol{v}^{2}_{0})^{-1/2}\sum_{i=1}^{n} \nabla_{\boldsymbol{\beta}} \ell_{i}(\boldsymbol{\beta^*})$, where $\boldsymbol{v}^{2}_{0}$ is the Fisher information matrix for a single observation and $\ell_i(\boldsymbol{\beta^*})$ is the log-likelihood for the $i$-th data point. Since the normalized score is asymptotically standard normal under standard regularity conditions (e.g., twice continuous differentiability of the log-likelihood, positive definite Fisher information matrix, etc.), this bound quantifies the accuracy of the asymptotic MLE distribution at finite samples. These results naturally extend to the case where $p$ grows with $n$, thereby accommodating the infinite-dimensional nature of the diffeomorphism and providing confidence bounds for the MLE of stationary points.
For our first set of results, we assume the following: Let the pairs $(X_{1}, Y_{1})$. $\ldots,$ $ (X_{n}, Y_{n})$ be generated from the model \begin{align*}
    (X_{1}, ..., X_{n}) \stackrel{i.i.d.}{\sim} \mathbb{P}_{X}([0, 1]) \quad\text{and}\quad
    Y_{i}|X_{i} \stackrel{ind.}{\sim} N(f_{\boldsymbol{\beta^*}}(X_{i}), \sigma^{2}).
\end{align*} 
We let $f_{\boldsymbol{\beta^*}}(X_{i}) = g(\gamma_{\boldsymbol{\beta^*}}(X_{i}))$, where $g$ is the template function and $\gamma_{\boldsymbol{\beta^*}}$ is the diffeomorphism parameterized by $\boldsymbol{\beta^*}$, using the parameterization discussed in Section~\ref{sec:Parameterizing}. We make the following assumptions about $\boldsymbol{\beta^*}$, $g$, $\sigma^2$, and $\mathbb{P}_{X}$. 
\begin{assumption}[Correct Specification]
\label{ass:CS}
Let $g$ be a fixed, thrice continuously differentiable function in $\mathcal{F}_{M}$, and assume that $g \neq f_{\beta^*}$. In addition, suppose $\boldsymbol{\beta^*} \in \boldsymbol{\Theta}$, where $\boldsymbol{\Theta} = \{\boldsymbol{\beta} \in \mathbb{R}^{p} : 0 < \left\lVert \boldsymbol{\beta} \right\rVert < \pi\}$. Finally, let $\sigma^2$ be a fixed positive constant.
\end{assumption}
 \begin{assumption}[Positivity and Independence of $\mathbb{P}_{X}$]
 \label{ass:PI_Px}
For every compact subinterval $[a, b] \subseteq [0, 1]$ with $a < b$, the measure $\mathbb{P}_{X}([a, b]) > 0$. In addition, $\mathbb{P}_{X}$ is independent of both $g$ and $\boldsymbol{\beta^*}$.
\end{assumption}

\begin{assumption}[Identifiability]
\label{ass:Identifiability}
    The Fisher information matrix for a single observation at $\boldsymbol{\beta^*}$, denoted by $\boldsymbol{v}_{0}^{2}$, is given by $ \sigma^{-2}\mathbb{E}\left[\left( \nabla_{\boldsymbol{\beta}}  f_{\boldsymbol{\beta^*}}(x)\right) \left(\nabla_{\boldsymbol{\beta}}  f_{\boldsymbol{\beta^*}}(x) \right)^T \right]$ and is positive definite.
\end{assumption}

We assume the above to ensure that $\boldsymbol{\beta^*}$ is the unique maximizer of the expected log-likelihood and that the log-likelihood is twice continually differentiable. To derive the finite error bounds of the maximum likelihood estimate, we primarily rely on the theory of \cite{Spokoiny2012_ParametricEstimation}. They proved that if the gradient of the log-likelihood is a sub-exponential random vector and the log-likelihood admits a second-order Taylor expansion around the ground truth, then within a local neighborhood of $\boldsymbol{\beta^*}$, the log-likelihood is bounded between two Gaussian-like quadratic processes.
This result is used to derive finite-sample error bounds for the asymptotic normal approximation of MLE's within this local neighborhood. In addition, if the negative expected log-likelihood is lower bounded by a quadratic function of $\boldsymbol{\beta}$, it is possible obtain a large deviation bound for the MLE. Because these bounds involve functions of the sample size $n$ and parameter dimension $p$, the theory of \cite{Spokoiny2012_ParametricEstimation} can also easily be extended to the case where the parameter dimension grows with the sample size.  Since our goal is to derive finite-sample error bounds for the estimator while accommodating the model’s infinite-dimensional structure through a sequence of increasingly rich finite-dimensional approximations, adopting this framework is well justified.
\cite{PanovSpokoiny2015_FiniteSampleBvM} used a similar approach for the Bayesian case and derived an error bound for the asymptotic Bernstein-von Mises Gaussian approximation of the posterior. For brevity, we will focus strictly on the maximum likelihood estimate case in this section and expand to the Bayesian case in the appendix.

\subsection{Finite Sample Theory for $\boldsymbol{\beta}$ and Stationary Points}
\label{sec:FST_Beta}
  To derive our finite-sample bounds, we must state precisely what we mean by ``local neighborhoods of $\boldsymbol{\beta^*}$". Let $\mathbb{K}(\boldsymbol{\beta}, \boldsymbol{\beta^*})$  be the Kullback–Leibler (KL) divergence between the models parameterized by $\boldsymbol{\beta}$ and $\boldsymbol{\beta^*}$ and let \begin{align}
      H({\boldsymbol{\beta}}) = \frac{\partial}{\partial \boldsymbol{\beta} \partial \boldsymbol{\beta}^T}\mathbb{K}(\boldsymbol{\beta}, \boldsymbol{\beta^*}).
  \end{align} This is the Hessian of $\mathbb{K}(\boldsymbol{\beta}, \boldsymbol{\beta^*})$ at $\boldsymbol{\beta} \in \boldsymbol{\Theta}$. Now $H({\boldsymbol{\beta^*}}) = \boldsymbol{v}_{0}^{2}$,  the Fisher information matrix at $\boldsymbol{\beta^*}$, which we assumed to be positive definite (Assumption \ref{ass:Identifiability}). In addition, by Assumption \ref{ass:CS}, the second partial derivatives in the Hessian are continuous. Thus, there exists a closed ball $\bar{B}(\boldsymbol{\beta^*}, \tau_0)$ of positive radius $\tau_0 > 0$ centered at $\boldsymbol{\beta^*}$ such that, for all $\boldsymbol{\beta} \in \bar{B}(\boldsymbol{\beta^*}, \tau_0)$, the Hessian of $\mathbb{K}(\boldsymbol{\beta}, \boldsymbol{\beta^*})$ evaluated at $\boldsymbol{\beta}$ is positive definite.
  Consequently, $\mathbb{K}(\boldsymbol{\beta}, \boldsymbol{\beta^*})$ is strictly convex for $\boldsymbol{\beta} \in \bar{B}(\boldsymbol{\beta^*}, \tau_{0}) $ on the closed ball $\bar{B}(\boldsymbol{\beta^*}, \tau_{0})$.
  This leads to the following definition.
  \begin{definition}[Local Sets of $\boldsymbol{\beta^*}$]
  \label{def:local_sets}
  Let $\tau_{0}$ be the radius of a closed ball on which $\mathbb{K}(\boldsymbol{\beta}, \boldsymbol{\beta^*})$ is strictly convex. The local sets of $\boldsymbol{\beta^*}$ are defined as \begin{equation}
      \Theta(u) =  \{ \boldsymbol{\beta} \in \Theta : \left\lVert\boldsymbol{v}_{0}(\boldsymbol{\beta} - \boldsymbol{\beta^*})\right\lVert \leq u\},
  \end{equation}
with $u \leq u_{0} = \frac{\tau_{0} }{\sqrt{\lambda_{\max}}}$, where $\boldsymbol{v}_{0}$ is the square root matrix of $\boldsymbol{v}_{0}^{2}$ and $\lambda_{\max}$ is the maximum eigenvalue of $\boldsymbol{v}_{0}^{2}$.   
  \end{definition}

We define $u_{0}$ in this way because $\boldsymbol{\beta} \in \Theta(u_{0})$ if and only if $\boldsymbol{\beta} \in \bar{B}(\boldsymbol{\beta^*}, \tau_{0})$.  Intuitively, $\Theta(u_{0})$ defines the neighborhood where the expected log-likelihood will resemble that of an identifiable linear Gaussian model, whose negative expected log-likelihood is a strictly convex function. Within this neighborhood, we can establish our finite-sample error bounds for the MLE. We define the following two quantities: \begin{equation}
    \lambda_{\min}(u_{0}) = \inf_{\boldsymbol{\beta} \in \Theta(u_{0})} \lambda_{\min}(H(\boldsymbol{\beta})) \quad\text{and}\quad \kappa(u_{0}) =  \underset{\boldsymbol{\beta} \notin \Theta(u_{0}) } {\inf} \mathbb{K}(\boldsymbol{\beta}, \boldsymbol{\beta^*}).
\end{equation}
 A small $\lambda_{\min}(u_{0})$ means that the expected log-likelihood has low curvature in a certain direction at a certain point, so nearby parameter values along that direction are difficult to distinguish. A small $\kappa(u_{0})$ means there may exist a parameter value far outside the local set $\Theta(u_{0})$ that provides a reasonable fit.
 Due to $H(\boldsymbol{\beta})$ being positive definite for all $\boldsymbol{\beta} \in \Theta(u_0)$, the parameter space being bounded and each $\boldsymbol{\beta}$ defining a unique diffeomorphism, both of these quantities are positive. With $\lambda_{\max}$ being the maximum eigenvalue of $\boldsymbol{v}^{2}_{0}$, we  define \begin{equation}
    b = \min\left(\frac{1}{2}\frac{\lambda_{\min}(u_{0})}{\lambda_{\max}},\ \frac{1}{4\pi^2}  \frac{\kappa(u_{0})}{\lambda_{\max}} \right).
\end{equation} 
It can be shown that $\mathbb{K}(\boldsymbol{\beta}, \boldsymbol{\beta^*})$ is lower bounded by the positive function $b\left\lVert\boldsymbol{v}_{0}(\boldsymbol{\beta} - \boldsymbol{\beta^*})\right\lVert^{2}$.  With this quadratic lower bound, we derive our finite sample error bonds. We explore the derivation of these constants and how they prove the subsequent theorems in Section~\ref{sec:Aux_Res}.

With $\Theta(u)$ and $b$ defined, we state our first important result, which establishes an upper bound for the probability that the MLE is not within the local neighborhood $\Theta(u)$ for $u \leq u_{0}$, and is analogous to the fact that the MLE is a consistent estimator.  \begin{theorem}
\label{thm:Consistency}
  Suppose Assumptions~\ref{ass:CS}-\ref{ass:Identifiability} hold.  Let $u \leq u_{0}$, $p \geq 2$ be the parameter dimension, $v_{0}$ be a fixed value dependent on $p$ and $0 < \alpha < 1$. Then if 
      $  n^{1/2}ub \geq 6v_{0}\sqrt{-\ln(\alpha) + 2p}$, \begin{equation}
          \mathbb{P}\left(\boldsymbol{\hat{\beta}} \notin \Theta(u)\right) \leq \alpha.
      \end{equation}
\end{theorem}

The positive value $v_{0}$ is dependent on the parameter dimension $p$ and we discuss it in greater detail in Section~\ref{sec:Aux_Res}. This theorem states that for the MLE to be within the local set $\Theta(u)$ with high confidence, $n^{1/2}$ needs to surpass a certain threshold dependent on $u$ and $b$. Because $u$ is chosen by the user, the constant $b$ is the primary factor determining the efficiency of the estimator. If either $\lambda_{\min}(u_{0})$ or $\kappa(u_{0})$ are small, then the analyst needs to collect a large amount of data in order to obtain estimates near the ground truth. In the case where the sample size $n$ is allowed to grow, then if $u \geq Cn^{-1/2}$ for a $C > 0$ derived from the theorem above, then this large deviation bound remains. Thus, the neighborhood can shrink at a rate of $n^{-1/2}$, which establishes root-$n$ consistency of the MLE. 

Now we aim to establish finite sample error bounds for the asymptotic normal approximation for the sampling distribution of the MLE. Let $u$ be small enough such that $\rho(u) + \delta(u) < 1$, where $\rho(u)$  and $\delta(u)$ are positive multiples of $u$ and are defined explicitly in the Supplementary Materials.  We define the following quantities: \begin{equation}
         \boldsymbol{f_{\epsilon}} = (1 - \rho(u) - \delta(u))\boldsymbol{v}_{0}^{2}\quad\text{and} \quad\boldsymbol{\xi}_{\boldsymbol{\epsilon}} = (n\boldsymbol{f_{\epsilon}})^{-1/2} \sum\limits_{i = 1}^{n} \nabla_{\boldsymbol{\beta}} \ell_{i}(\boldsymbol{\beta^*}).
\end{equation}
 As stated before, $\ell_{i}(\boldsymbol{\beta})$ is the log-likelihood for a single data point at $\boldsymbol{\beta}$. The quantities are scaled versions of the Fisher information matrix and normalized score, which is zero mean and asymptotically normal under the traditional i.i.d setup. They are used in the following theorem, which provides a bound for how much the quantity $\sqrt{n\boldsymbol{f_{\epsilon}}}(\boldsymbol{\hat{\beta}} - \boldsymbol{\beta^*})$ deviates from the random vector $\boldsymbol{\xi_{\epsilon}}$.

\begin{theorem} 
\label{thm:Normality}
Suppose Assumptions~\ref{ass:CS}-\ref{ass:Identifiability} are satisfied. If $u \leq u_{0}$ such that $\rho(u) + \delta(u) < 1$, then it holds on a random set $C_{\boldsymbol{\epsilon}}(\sqrt{n}u)$ and a random value $\Delta_{\boldsymbol{\epsilon}}(u)$ that \begin{equation}
\Bigl\|\sqrt{n\boldsymbol{f_{\epsilon}}}\left(\boldsymbol{\hat{\beta}} - \boldsymbol{\beta^*}\right) - \boldsymbol{\xi_{\epsilon}}\Bigr\|^{2} \leq 2\Delta_{\boldsymbol{\epsilon}}(u),
\end{equation}
where the probability of $C_{\boldsymbol{\epsilon}}(\sqrt{n}u)$ and the distribution of $\Delta_{\boldsymbol{\epsilon}}(u)$ are dependent on $u$.

\end{theorem}

If $u$ is chosen to be sufficiently small, then the value $\Delta_{\boldsymbol{\epsilon}}(u)$ will also be small with high probability. In addition, if $n$ is sufficiently large, then it is possible to show that $\mathbb{P}(C_{\boldsymbol{\epsilon}}(\sqrt{n}u))$ will be near 1. Thus, for sufficiently small $u$ and large enough $n$, the probability that the value $\sqrt{n\boldsymbol{f_{\epsilon}}}(\boldsymbol{\hat{\beta}} - \boldsymbol{\beta^*})$, which closely approximates the normalized MLE, significantly deviates from $\boldsymbol{\xi_{\epsilon}}$ is small. As for large $n$,  $\boldsymbol{\xi_{\epsilon}}$ is approximately standard normal, this theorem establishes that the sampling distribution of the MLE is approximately normal for a large enough sample size, which matches standard asymptotic theory. 

 Theorems~\ref{thm:Consistency} and~\ref{thm:Normality} can easily be extended to establish MLE normality in the growing parameter case, where $p$ is a function of $n$, all constants are now dependent on $p$, in which we indicate with a subscript, and $\boldsymbol{\beta^{*}_{\textit{p}}}$  is a sequence of ground truth parameters. However, for the error bounds to remain small as the dimension increases, we require additional assumptions on the parameter dimension function $p = p(n)$ and the radii of the local sets $u_{p}$.
 \begin{assumption}[Conditions on $u_p$]
 \label{ass:u_p}
We require the following:
\begin{enumerate}[label=(\roman*),itemjoin={{\quad}}, itemjoin*={{; \quad}}]

    \item $u_p \to 0$,
    \item $u_p < \dfrac{(\tau_0)_p}{\sqrt{(\lambda_{\max})_p}}$,
    \item $\rho_p(u_p) + \delta_p(u_p) \to 0$,
    \item $\sqrt{n}\,u_p \to \infty$.
\end{enumerate}
\end{assumption}
Essentially, we require $u_{p}$ to decrease in size while ensuring the local sets remain in the neighborhood where $\mathbb{K}(\boldsymbol{\beta}_{p}, \boldsymbol{\beta^{*}_{\textit{p}}})$ is strictly convex. In addition, we would like $u_{p}$ to decrease such that the scaling factor $\rho_p(u_p) + \delta_p(u_p)$ of Theorem \ref{thm:Normality} disappears. Next, we require the following on the function $p(n)$: 

\begin{assumption}[Growth Rate of $p(n)$] 
\label{ass:p_n}
 Let $\alpha_{p} \rightarrow 0$ and $p(n)$ be an increasing function such that $p(n) \rightarrow \infty $ and 
    \begin{enumerate}[label=(\roman*),itemjoin={{\quad}}, itemjoin*={{; \quad}}]
\item $ n^{1/2}(u_{p})(b_{p})  \geq   6(v_{0})_{p}\sqrt{-\ln(\alpha_{p}) + 2p(n)}$,
\item $\rho(u_{p})\left(1 +\sqrt{-\ln(\alpha_{p}) + 2p(n)}\right)  \rightarrow 0$,
\item If $Y$ is a $N(0, \sigma^2)$ random variable, then  $\mathbb{P}\left(\left\lVert Y \right\lVert(v_{0})_{p}  \leq \sqrt{n}u_{p}\right) \rightarrow 1 $. 
    \end{enumerate}
\end{assumption}
The first condition ensures that the MLE $\boldsymbol{\hat{\beta}}_{p}$ is located in $\Theta(u_{p})$ with high probability. The next two assumptions ensure that the $\mathbb{P}\left(C_{\boldsymbol{\epsilon}, p}(\sqrt{n}u_{p})\right)$ is near $1$ and $\Delta_p(u_{p})$ is near $0$ with high probability for large $n$ and $p$. With these assumptions established, we apply our finite sample bounds to the growing parameter case.  
\begin{theorem}
\label{thm:Growing_Parameter}
   Assume Assumptions~\ref{ass:CS}-\ref{ass:Identifiability} and Assumptions~\ref{ass:u_p}-\ref{ass:p_n} hold. Let $p \geq 2$  and $\boldsymbol{\beta^{*}_{\textit{p}}}$ converge to an infinite dimensional limit $\boldsymbol{\beta}$, where $\left\lVert \boldsymbol{\beta} \right\lVert < \pi $. Let $\boldsymbol{\hat{\beta}}_{p}$ be the MLE in dimension $p$ and let $u_{p}$, $\alpha_{p}$ and a function $p=p(n)$  satisfy all the conditions from Assumptions~\ref{ass:u_p} and~\ref{ass:p_n}. Then  
\begin{equation}
             \mathbb{P}\left(\boldsymbol{\hat{\beta}}_{p} \notin \Theta_p(u_{p})\right) \leq \alpha_{p} \rightarrow 0
    \end{equation}
and on the random set $C_{\boldsymbol{\epsilon}, p}(\sqrt{n}u_{p})$ \begin{equation}
\Bigl\|\sqrt{n(\boldsymbol{f}_{\boldsymbol{\epsilon}})_{p}}\left(\boldsymbol{\hat{\beta}}_{p} -\boldsymbol{\beta^{*}_{\textit{p}}}\right) - \left(\boldsymbol{\xi}_{\boldsymbol{\epsilon}}\right)_{p}\Bigr\|^{2}  \leq 2\Delta_{\boldsymbol{\epsilon}, p}\left(u_{p}\right), 
\end{equation}
where $\mathbb{P}\left(C_{\boldsymbol{\epsilon}, p}(\sqrt{n}u_{p})\right) \rightarrow 1$ and  $\Delta_{\boldsymbol{\epsilon}, p}\left(u_{p}\right)  \rightarrow 0$ in probability as $p \rightarrow \infty$. 
\end{theorem}

This theorem can be viewed as an extension of the preceding results, allowing the relevant constants to vary with $p$. It characterizes the rate at which $p = p(n)$ may increase while still preserving the finite-sample consistency and asymptotic normality of the estimators. This rate depends critically on how the quantities $u_p$ and $b_p$ evolve with $p$. Owing to the complexity of the diffeomorphic parameterization, deriving these rates analytically is challenging and typically requires numerical investigation. In particular, if $u_p$ and $b_p$ decay rapidly to zero, then $p(n)$ must grow more slowly, implying that substantially larger sample sizes are needed to increase the dimension of $\boldsymbol{\beta}$ without incurring excessive error.

Under the conditions of Theorem~\ref{thm:Growing_Parameter}, the estimator $\boldsymbol{\hat{\beta}}_{p}$ converges in probability to an infinite-dimensional ground truth parameter $\boldsymbol{\beta}^{*}$, and the induced diffeomorphism $\boldsymbol{\gamma}_{\boldsymbol{\hat{\beta}}}$ converges to its infinite-dimensional counterpart. Moreover, as $\mathbb{P}\!\left(C_{\boldsymbol{\epsilon}, p}(\sqrt{n}u_{p})\right) \to 1$ and $\Delta_{\boldsymbol{\epsilon}, p}(u_p) \to 0$ in probability, the distribution of $\boldsymbol{\hat{\beta}}_{p}$ for large $p$ becomes asymptotically close to that of $(\boldsymbol{\xi}_{\boldsymbol{\epsilon}})_{p}$. If $(\boldsymbol{\xi}_{\boldsymbol{\epsilon}})_{p}$ is bounded in probability, then it is approximately normally distributed. Consequently, the asymptotic properties of the maximum likelihood estimator are preserved even in the presence of an infinite-dimensional diffeomorphic structure.

\subsection{Finite Sample Error Bounds for the Stationary Points}
\label{sec:FST_SPs}

If we again assume the ground truth function $f \in \mathcal{F}_{M}$ is of the form $g \circ \gamma_{\boldsymbol{\beta^*}}$ and $b_{1}, \dots, b_{M}$ are the stationary points of the template function $g$, the function \begin{equation}
\boldsymbol{\Phi} : \mathbb{R}^{p} \to \mathbb{R}^{M}, \qquad 
\boldsymbol{\Phi}(\boldsymbol{\beta}) := 
\left( \gamma_{\boldsymbol{\beta}}^{-1}(b_{1}), \, \dots, \, \gamma_{\boldsymbol{\beta}}^{-1}(b_{M}) \right)
\end{equation}
maps $\boldsymbol{\beta^*}$ to the stationary points of $f$. If we input $\boldsymbol{\hat{\beta}}$ into $\boldsymbol{\Phi}$,  we output the maximum likelihood estimate of the stationary points. Now, assume that $f= g \circ \gamma_{\boldsymbol{\beta^*}} $ where Assumptions~\ref{ass:CS},~\ref{ass:PI_Px} and~\ref{ass:Identifiability} from the previous subsection apply. In order to derive finite sample error bounds for $\boldsymbol{\boldsymbol{\Phi}}(\boldsymbol{\hat{\beta}})$, we need the following assumption: \begin{assumption}[Strict Positivity of the Derivative of $\gamma_{\boldsymbol{\beta^*}}$]
\label{ass:pos_deriv_gamma}
    The value $\inf\limits_{x \in [0, 1]} \gamma_{\boldsymbol{\beta^*}}'(x) = m$, a positive number.
\end{assumption}
With this assumption, we define the compact set \begin{equation}
    D_{m_{0}} = \{ \boldsymbol{\beta} \in \boldsymbol{\Theta} \mid \inf\limits_{x \in [0, 1]} \gamma_{\boldsymbol{\beta}}'(x) \geq  m_{0} \},
\end{equation}where $m > m_{0} > 0$. As shown in the Supplementary Materials, $\boldsymbol{\Phi}$ is a Lipschitz function on the set $D_{m_{0}}$.  Because $\boldsymbol{\beta^*} \in D_{m_{0}}(\boldsymbol{\beta})$, $\gamma'_{\boldsymbol{\beta}}(x)$ is a continuous function of $\boldsymbol{\beta}$ and $\gamma'_{\boldsymbol{\beta}}(x)$ is a uniformly continuous function of $x$, $\boldsymbol{\Phi}$ is Lipschitz in a neighborhood around $\boldsymbol{\beta^*}$. Thus, there exists a $\phi > 0$ such that if $\left\lVert\boldsymbol{\beta} - \boldsymbol{\beta^*}\right\lVert \leq \phi$, then  $\left\lVert\Phi(\boldsymbol{\beta}) - \Phi(\boldsymbol{\beta^*})\right\lVert \leq L_{\phi}\left\lVert\boldsymbol{\beta} - \boldsymbol{\beta^*} \right\lVert $ for some $L_{\phi} > 0$. This leads to the following result, which provides establishes conditions for $\Phi(\boldsymbol{\hat{\beta}})$ to be near the ground truth with high probability. 
\begin{corollary}
 Assume Assumptions \ref{ass:CS}-\ref{ass:Identifiability}, and \ref{ass:pos_deriv_gamma} hold. Let $\lambda_{\min}$ be the minimum eigenvalue of $\boldsymbol{v}_{0}^{2}$, $u \leq \min(u_{0}, \sqrt{\lambda_{\min}} \phi)$, and $0 < \alpha < 1$. If the conditions of Theorem~\ref{thm:Consistency} are satisfied, then \begin{align*}
        \mathbb{P}\left( \left\lVert\Phi(\boldsymbol{\hat{\beta}}) - \Phi(\boldsymbol{\beta^*})\right\lVert \leq \frac{L_{\phi}}{\sqrt{\lambda_{\min}}} u\right) > 1 - \alpha.
    \end{align*}
\end{corollary}

\begin{proof}
    If $\boldsymbol{\hat{\beta}} \in \Theta(u)$, then \begin{align*}
     \left\lVert\boldsymbol{\hat{\beta}} - \boldsymbol{\beta^*} \right\lVert  \leq \frac{1}{\sqrt{\lambda_{\min}}} \left\lVert \boldsymbol{v}_{0}(\boldsymbol{\hat{\beta}} - \boldsymbol{\beta^*}) \right\lVert   \leq \frac{u}{\sqrt{\lambda_{\min}}} \leq \phi .
     \end{align*}
     Therefore, we conclude that $\left\lVert\Phi(\boldsymbol{\hat{\beta}}) - \Phi(\boldsymbol{\beta^*})\right\lVert \leq L_{\phi}\left\lVert\boldsymbol{\hat{\beta}} - \boldsymbol{\beta^*}\right\lVert  \leq \frac{L_{\phi}}{\sqrt{\lambda_{\min}}}u$. 
Thus, \\ $\mathbb{P}(\boldsymbol{\hat{\beta}} \in \Theta(u)) \leq  \mathbb{P}\left( \left\lVert\Phi(\boldsymbol{\hat{\beta}}) - \Phi(\boldsymbol{\beta^*})\right\lVert \leq \frac{L_{\phi}}{\sqrt{\lambda_{\min}}} u\right)$. As the conditions for Theorem \ref{thm:Consistency} are satisfied, the result follows. 
\end{proof}
Now, because $\Phi$ is Lipschitz on the set $D_{m_{0}}(\boldsymbol{\beta})$, $\Phi$ is first order differentiable almost everywhere in this set by Rademacher's theorem. Thus, assuming that $\boldsymbol{\beta^*}$ does not fall within the negligible set where $\boldsymbol{\Phi}$ fails to be differentiable, we apply the following theorem: 

\begin{theorem}
\label{thm:delta_method}
Assume Assumptions \ref{ass:CS}-\ref{ass:Identifiability}, and \ref{ass:pos_deriv_gamma} are satisfied. Let $\Phi : \mathbb{R}^{p} \to \mathbb{R}^{m}$ be a differentiable mapping at $\boldsymbol{\beta^*} \in D_{m_{0}}(\boldsymbol{\beta})$ with Jacobian $J_\Phi(\boldsymbol{\beta^*}) \neq 0$.  
Define $d_{\boldsymbol{\epsilon}} = \sqrt{\boldsymbol{f}_{\epsilon}}$ and suppose the conditions of Theorem \ref{thm:Normality} are satisfied. Then, on the random set $C_{\boldsymbol{\epsilon}}(\sqrt{n} u)$, we have
\begin{align*}
    \Bigl\| \sqrt{n} \bigl(\Phi(\boldsymbol{\hat{\beta}}) - \Phi(\boldsymbol{\beta^*}) \bigr) 
    - J_\Phi(\boldsymbol{\beta^*}) d_{\boldsymbol{\epsilon}}^{-1} \boldsymbol{\xi}_{\epsilon} \Bigr\|^2
    \;\leq\; 2\, \sigma_1\bigl(J_\Phi(\boldsymbol{\beta^*}) d^{-1}_{\boldsymbol{\epsilon}}\bigr)\, \Delta_{\boldsymbol{\epsilon}}(u) + o_{p}(1),
\end{align*}
    where $\sigma_1\left(J_\Phi(\boldsymbol{\beta^*})d^{-1}_{\boldsymbol{\epsilon}}\right)$ is the largest singular value of $J_\Phi(\boldsymbol{\beta^*}) d^{-1}_{\boldsymbol{\epsilon}}$.
\end{theorem}

The proof of this theorem follows from an argument analogous to the multivariate delta method together with an application of Theorem~\ref{thm:Normality}. Since the singular value $\sigma_1(J_\Phi(\boldsymbol{\beta}^{*}))$ is fixed, the result implies that when $u$ is sufficiently small and $n$ is large, the distribution of $\sqrt{n}\,(\Phi(\boldsymbol{\hat{\beta}})-\Phi(\boldsymbol{\beta}^{*}))$ is well approximated by that of $J_\Phi(\boldsymbol{\beta}^{*})\, d_{\boldsymbol{\epsilon}}^{-1}\, \boldsymbol{\xi}_{\boldsymbol{\epsilon}}$. Because $\boldsymbol{\xi}_{\boldsymbol{\epsilon}}$ is approximately standard normal for small $u$ and large $n$, the vector $J_\Phi(\boldsymbol{\beta}^{*})\, d_{\boldsymbol{\epsilon}}^{-1}\boldsymbol{\xi}_{\boldsymbol{\epsilon}}$ is approximately normal with mean $0$ and covariance $J_\Phi(\boldsymbol{\beta}^{*})\, d_{\boldsymbol{\epsilon}}^{-2}\, J_\Phi(\boldsymbol{\beta}^{*})^{T}$. Since $d_{\boldsymbol{\epsilon}}^{-2}$ coincides with the inverse Fisher information for a single observation (up to a mild rescaling), this result mirrors the classical multivariate delta method. Thus, unlike existing approaches discussed in the introduction, we obtain a joint asymptotic normal distribution for the MLE of the stationary points, along with a finite-sample bound quantifying the accuracy of this approximation. As in Theorems~\ref{thm:Consistency} and~\ref{thm:Normality}, the quality of the approximation depends on both $u$ and $b$. Finally, these results extend to the growing-parameter setting provided that the strict positivity of $\gamma'_{\boldsymbol{\beta}^{*}}$ and the necessary differentiability conditions hold for all $\boldsymbol{\beta}^{*}_{p}$.

\subsection{Auxiliary Results}
\label{sec:Aux_Res}
The theorems in the preceding subsections are direct consequences of the i.i.d.\ theory developed by \cite{Spokoiny2012_ParametricEstimation}. In essence, once the log-likelihood satisfies a set of regularity conditions, finite-sample error bounds for the MLE follow immediately. We investigate these conditions thoroughly in the Supplementary Materials, but provide here a concise summary of the key results. 

The first requirement is to verify that the gradient of the stochastic component $\zeta_{i}(\boldsymbol{\beta}) = \ell_{i}(\boldsymbol{\beta}) - \mathbb{E}\ell_{i}(\boldsymbol{\beta})$, where $\ell_{i}(\boldsymbol{\beta})$ is the log-likelihood contribution of a single observation, is a sub-exponential random vector for every $i = 1, \ldots, n$. This is not obvious, since the log-likelihood of our model is non-concave. However, the crucial fact that $\left\lVert \boldsymbol{\beta} \right\rVert < \pi$ leads to the following lemma.

\begin{lemma}
\label{lem:gradient_lem}
 Suppose Assumptions~\ref{ass:CS}-\ref{ass:Identifiability} hold.
       Then $f_{\boldsymbol{\beta}}(x) = g(\gamma_{\boldsymbol{\beta}}(x))$ is a thrice continuously differentiable function with respect to $\boldsymbol{\beta}$ for all $x \in [0, 1]$. In addition, the norm of the $\nabla_{\boldsymbol{\beta}} f_{\boldsymbol{\beta}}(x)$ is bounded by the function  $C(p \ln p)$ where $C > 0$.
\end{lemma}

The lemma follows from standard results in multivariable calculus, and we provide a complete proof in the Supplementary Materials. Since the gradient of $f_{\boldsymbol{\beta}}$ is a bounded random vector, it follows that $\zeta_{i}(\boldsymbol{\beta})$ is sub-Gaussian, and hence sub-exponential. This is formalized in the next lemma.

\begin{lemma}
\label{lem:subexp}
Suppose Assumptions~\ref{ass:CS}--\ref{ass:Identifiability} hold. Consider $\zeta_{i}(\boldsymbol{\beta})$ with $\boldsymbol{\beta} \in \boldsymbol{\Theta}$ and $i \in \{1, \ldots, n\}$. Then, for any $\boldsymbol{\gamma} \in \mathbb{R}^{p}$,
\begin{align*}
\mathbb{E}\exp\!\left(\lambda \boldsymbol{\gamma}^{T} \nabla \zeta_{i}(\boldsymbol{\beta})\right)
\leq 
\exp\!\left(\frac{v_{0}^{2}\lambda^{2}}{2} \left\lVert \boldsymbol{v}_{0}\boldsymbol{\gamma} \right\rVert^{2}\right),
\end{align*}
where $v_{0}^{2} = \max\!\left(1, \frac{(Cp\ln p)^{2}}{\sigma^{2}\lambda_{\min}}\right)$, $Cp\ln(p)$ is the uniform bound on $\left\lVert\nabla_{\boldsymbol{\beta}} f_{\boldsymbol{\beta}}(x)\right\rVert$, and $\lambda_{\min}$ is the minimum eigenvalue of $\boldsymbol{v}_{0}^{2}$.
\end{lemma}

This bound is derived using the law of iterated expectations together with the Cauchy--Schwarz inequality. The next requirement is to verify that $\mathbb{K}(\boldsymbol{\beta}, \boldsymbol{\beta}^{*}) = \mathbb{E}\!\left[\ell_{i}(\boldsymbol{\beta}^{*}) - \ell_{i}(\boldsymbol{\beta})\right]$ is bounded below by a quadratic function of the form $b \left\lVert \boldsymbol{v}_{0}(\boldsymbol{\beta}^{*}-\boldsymbol{\beta}) \right\rVert^{2}$ for some positive constant $b$. If $\boldsymbol{\beta} \in \Theta(u_{0})$, the Hessian of $\mathbb{K}(\boldsymbol{\beta}, \boldsymbol{\beta}^{*})$ at $\boldsymbol{\beta}$ is positive definite. Since the second partial derivatives are continuous on $\Theta(u_{0})$ and $\Theta(u_{0})$ is closed, the value $\lambda_{\min}(u_{0}) = \inf_{\boldsymbol{\beta} \in \Theta(u_{0})} \lambda_{\min}(H(\boldsymbol{\beta}))$ is strictly positive. Applying a first-order Taylor expansion of $\mathbb{K}(\boldsymbol{\beta}, \boldsymbol{\beta}^{*})$ at $\boldsymbol{\beta}^{*}$ and using the integral remainder, we conclude that for $\boldsymbol{\beta} \in \Theta(u_{0})$,
\begin{align*}
\mathbb{K}(\boldsymbol{\beta}, \boldsymbol{\beta}^{*})
\geq \frac{\lambda_{\min}(u_{0})}{2} \left\lVert \boldsymbol{\beta} - \boldsymbol{\beta}^{*} \right\rVert^{2}
\geq \frac{\lambda_{\min}(u_{0})}{2\lambda_{\max}}
\left\lVert \boldsymbol{v}_{0}(\boldsymbol{\beta} - \boldsymbol{\beta}^{*}) \right\rVert^{2}.
\end{align*}
For $\boldsymbol{\beta} \notin \Theta(u_{0})$, we rely on the following lemma, proved in the Supplementary Materials.

\begin{lemma}
\label{lem:kappa}
Define $\kappa(u_{0}) = \inf_{\boldsymbol{\beta} \notin \Theta(u_{0})} \mathbb{K}(\boldsymbol{\beta}, \boldsymbol{\beta}^{*})$. If Assumptions~\ref{ass:CS}--\ref{ass:Identifiability} hold, then $\kappa(u_{0}) > 0$.
\end{lemma}
Since $\left\lVert \boldsymbol{\beta} \right\rVert < \pi$, we have
\begin{align*}
\left\lVert \boldsymbol{v}_{0}(\boldsymbol{\beta} - \boldsymbol{\beta}^{*}) \right\rVert^{2}
\leq \lambda_{\max} \left\lVert \boldsymbol{\beta} - \boldsymbol{\beta}^{*} \right\rVert^{2}
\leq 4\pi^{2}\lambda_{\max}.
\end{align*}
Thus, for $\boldsymbol{\beta} \notin \Theta(u_{0})$,
\begin{align*}
\mathbb{K}(\boldsymbol{\beta}, \boldsymbol{\beta}^{*})
\geq \kappa(u_{0})
\geq \frac{1}{4\pi^{2}} \frac{\kappa(u_{0})}{\lambda_{\max}}
\left\lVert \boldsymbol{v}_{0}(\boldsymbol{\beta} - \boldsymbol{\beta}^{*}) \right\rVert^{2}.
\end{align*}
Therefore, with 
\begin{align*}
b = \min\!\left( \frac{1}{2}\frac{\lambda_{\min}(u_{0})}{\lambda_{\max}}, \; \frac{1}{4\pi^{2}}\frac{\kappa(u_{0})}{\lambda_{\max}} \right),
\end{align*}
we obtain the lower quadratic bound $\mathbb{K}(\boldsymbol{\beta}, \boldsymbol{\beta}^{*}) \geq b \left\lVert \boldsymbol{v}_{0}(\boldsymbol{\beta} - \boldsymbol{\beta}^{*}) \right\rVert^{2}$. Together with the second-order differentiability of the log-likelihood with respect to $\boldsymbol{\beta}$, these results establish all of the conditions required by the theory of \cite{Spokoiny2012_ParametricEstimation} for i.i.d.\ models. Consequently, all the theorems in the previous subsection follow.

\section{Simulation}
\label{sec:sim}
In this section, we evaluate the performance of our method in estimating a function’s true stationary points. We compare its accuracy with that of the derivative-constrained Gaussian process (DGP) model of \citet{li2024semiparametric}, which is the only existing approach that permits direct sampling of likely stationary points and provides publicly available  code. The DGP method first fits a Gaussian process to the observed data, substitutes the fitted parameters into an analytical one-dimensional posterior for the stationary points, and subsequently samples likely stationary points using inverse transform sampling. Although the model proposed by \citet{wheeler2017bayesian} also allows direct sampling, the provided implementation is restricted to functions with only two stationary points. 

Because the DGP framework is fully Bayesian, we adopt a comparable setting for our diffeomorphism-based model by assigning priors to all parameters and sampling from the posterior distribution. The two approaches (DGP and the diffeomorphism based model) are then compared based on (i) the accuracy of the posterior mean estimates of stationary points and (ii) the empirical coverage probabilities of the true stationary points under their respective posterior distributions. For all simulation studies, we employ the unconstrained parameterization described in Section~\ref{sec:Joint_Estimation}. This formulation enables the use of a multivariate normal prior on the model parameters and facilitates efficient posterior computation via standard Bayesian algorithms. The template function is constructed using Hermite cubic interpolation with stationary points located at $b_{k} = \frac{k}{M + 1}$. Although this approach does not guarantee global smoothness, it yields highly accurate and computationally efficient results in practice.

Because $\gamma_{\boldsymbol{\beta}}$ is a nonlinear transformation of $\boldsymbol{\beta}$, the resulting posterior distribution is often highly multimodal. Consequently, simple Metropolis–Hastings samplers with multivariate normal proposals tend to perform poorly. To ensure adequate posterior exploration, we employ a two-stage strategy. First, multiple runs of the BFGS quasi-Newton algorithm are used to locate local modes of the log-posterior and compute their associated inverse Hessians. Next, several Metropolis–Hastings chains are initialized, each using a multivariate normal proposal with covariance $c(2.4/p^2)H^{-1}(\boldsymbol{\theta}_m)$, where $p$ is the parameter dimension, $H(\boldsymbol{\theta}_m)$ is the Hessian matrix of the log-posterior at mode $\boldsymbol{\theta}_m$, and $c$ is a positive scaling constant selected to achieve an average acceptance rate of approximately 23\%. This proposal choice follows the theoretical results of \citet{gelman1997weak}. Since many modes occur in regions of negligible posterior mass, we further apply a reweighting step in which each proposal covariance matrix is selected with probability proportional to the posterior density at its corresponding mode.

\subsection{Computational Considerations}

Before presenting the simulation results, we highlight several practical considerations for applying our method to real data. A notable drawback of both the DGP model and the hypothesis-testing-based strategies discussed in the Introduction is their tendency to produce false positives. These methods do not restrict the number of points at which the derivative can vanish, often leading to the identification of spurious stationary points in flat regions or in response to random noise. Our diffeomorphism approach largely mitigates this issue by explicitly constraining the number of stationary points. Consequently, it prioritizes regions where the function exhibits clear transitions between increasing and decreasing behavior.

To illustrate this difference, consider a scenario in which the true function contains an extended flat region. Figure~\ref{fig:flat_func} compares the posterior distributions of stationary points under both methods. The DGP method places most of its posterior mass over the slowly decreasing segment of the curve in $[0.75, 1]$, incorrectly identifying it as a stationary point. In contrast, the diffeomorphism method concentrates nearly all its posterior mass near the single true stationary point, correctly reflecting the underlying structure of the function.

\begin{figure}[htbp]
  \centering
  \begin{subfigure}[b]{0.48\textwidth}
    \includegraphics[width=\textwidth]{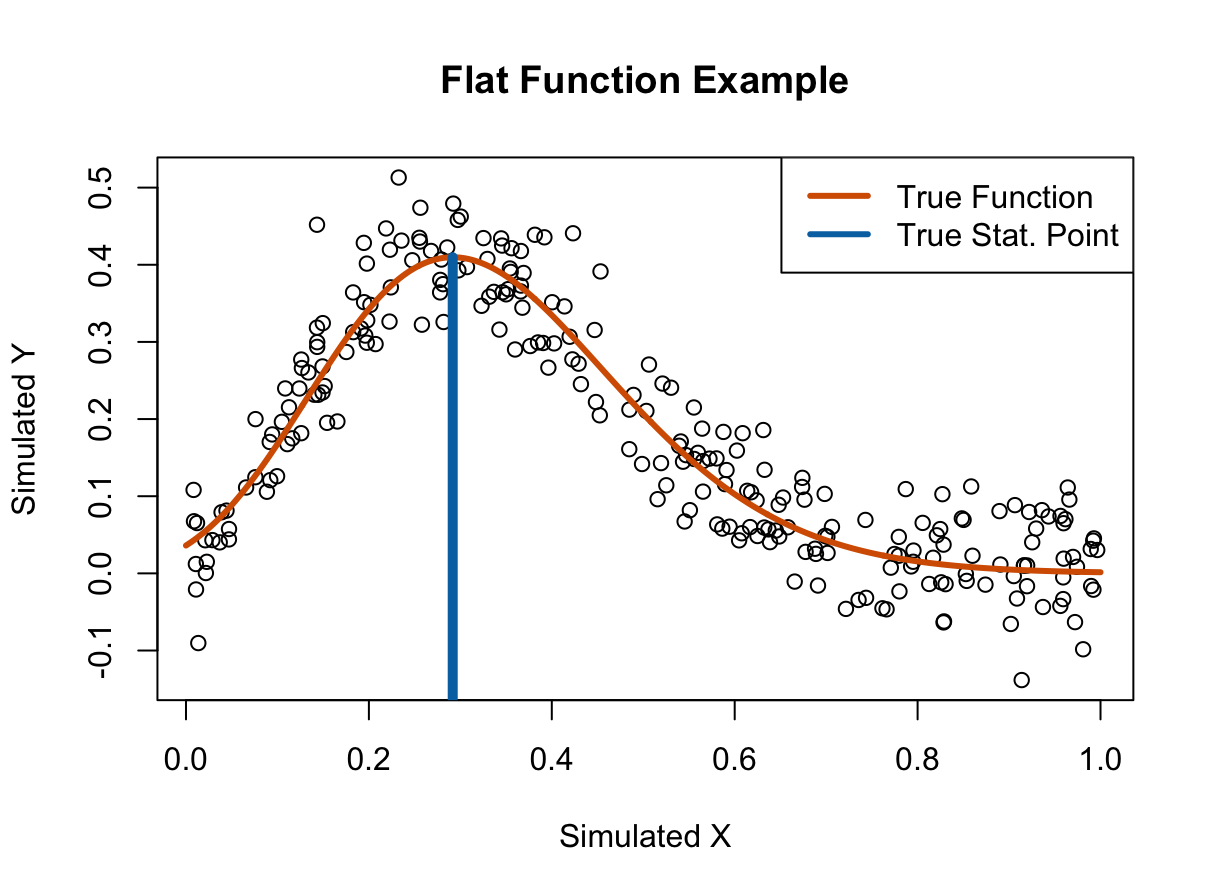}
  \end{subfigure}
  \begin{subfigure}[b]{0.48\textwidth}
    \includegraphics[width=\textwidth]{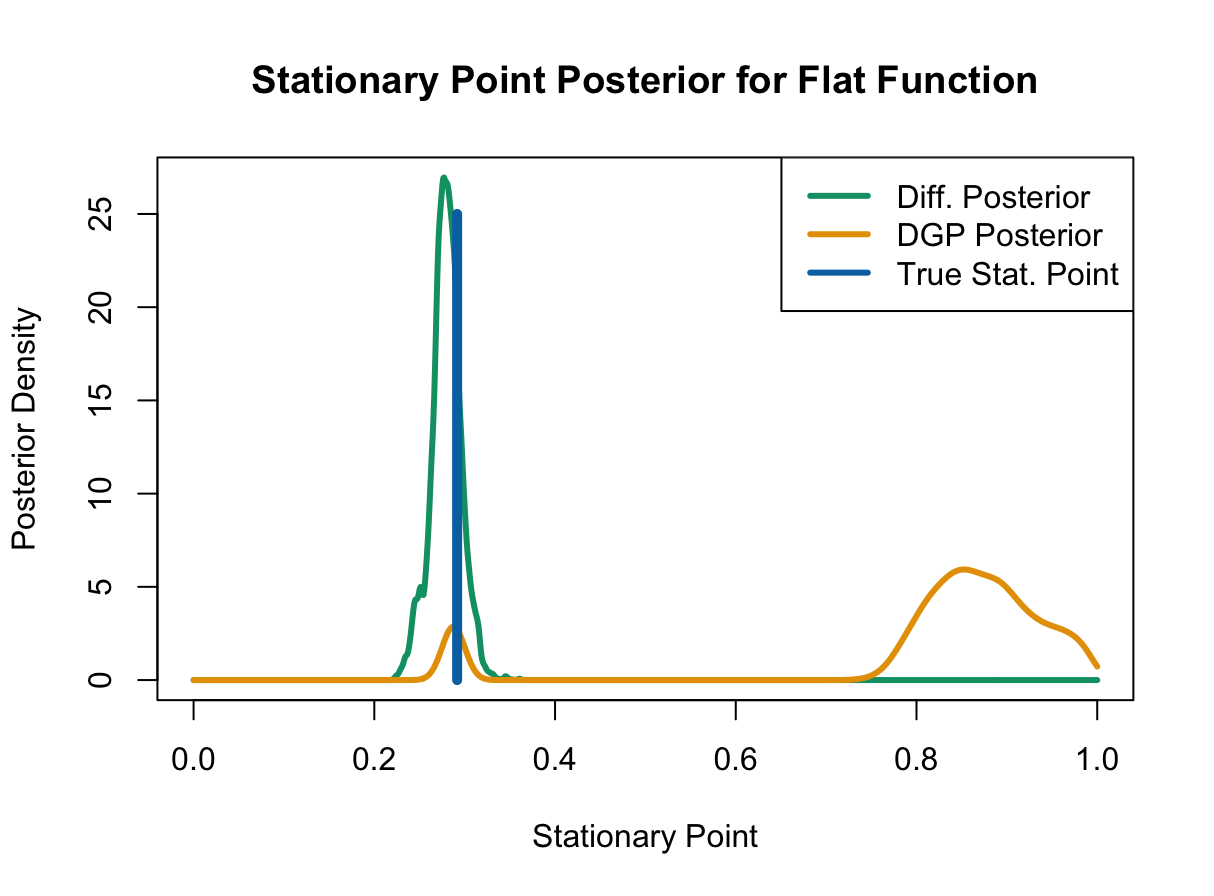}
  \end{subfigure}
  \caption{Comparison of the diffeomorphism and DGP methods for a function with a flat region. The DGP method places posterior mass over a broad, nearly flat region, whereas the diffeomorphism method concentrates around the true stationary point.}
  \label{fig:flat_func}
\end{figure}

For small sample sizes, both the sampling distribution and the posterior of the stationary points can exhibit skewness and multimodality. In such cases, asymptotic approximations should be avoided, and simulation-based inference methods such as bootstrap or MCMC should be employed to obtain reliable uncertainty quantification. As the sample size increases, both the sampling and posterior distributions tend to become approximately Gaussian, consistent with the asymptotic theory developed in Section~\ref{sec:theory}. 

Finally, since the full space of diffeomorphisms is infinite-dimensional, a finite-dimensional basis approximation is required in practice. In our experience, basis sizes between 5 and 15 are adequate for most applications. The optimal basis size can be selected using standard model selection tools such as AIC, cross-validated prediction error, or WAIC in the Bayesian setting \citep{mcelreath2018statistical}. If either the sampling distribution or the posterior concentrates around an incorrect stationary point, we recommend increasing the basis size or adjusting the stationary points of the template function accordingly.

\subsection{Simulation Study}

The first simulation evaluates the performance of both the diffeomorphism and DGP methods on data with a single true stationary point. This design produces a one-dimensional posterior for each method, allowing a direct comparison. The data-generating model is 
\begin{align*}
    X_{i} &\stackrel{i.i.d.}{\sim} \text{Unif}([0,1]),\\ 
    Y_{i} &= 1 + \sin(2X_{i}) + \cos(3X_{i}) + 3X_{i} - 2X_{i}^2 + \epsilon_{i}, \quad \epsilon_{i} \sim N(0, 0.25^2),
\end{align*}
which has a true stationary point at $b = 0.39973$. Sample sizes of $n = 50, 100, 200, 300$ were considered, with 100 simulations per sample size. The evaluation focuses on coverage probabilities of the stationary point, accuracy of the posterior mean estimates, and marginal posterior variances. 

Coverage probabilities were computed using 90\%, 95\%, and 99\% contiguous highest posterior density (HPD) intervals derived from the sampled stationary points. Posterior mean accuracy was assessed using the root mean squared error (RMSE): 
\begin{eqnarray*}
\text{RMSE} = \sqrt{\frac{1}{100}\sum_{j=1}^{100} (\bar{b}_j - b)^2},
\end{eqnarray*}
where $\bar{b}_j$ denotes the posterior  estimate from the $j$th simulation. Mean bias was calculated as
\begin{eqnarray*}
\text{Bias} = \frac{1}{100} \sum_{j=1}^{100} (\bar{b}_j - b),
\end{eqnarray*}
and posterior variance was summarized as the average posterior standard deviation across simulations.

For the diffeomorphism method, basis sizes of 4, 5, 6, and 7 were used for sample sizes $n = 50, 100, 200, 300$, respectively, in accordance with the increasing parameter dimension framework described in Section~3.2. Table~\ref{tab:sim1_coverage} presents coverage probabilities for both methods.

\begin{table}[ht]
\centering
\footnotesize
\begin{tabular}{lcccccc}
\toprule
 & \multicolumn{2}{c}{90\% HDP} & \multicolumn{2}{c}{95\% HDP} & \multicolumn{2}{c}{99\% HDP} \\
\cmidrule(lr){2-3} \cmidrule(lr){4-5} \cmidrule(lr){6-7}
 & Diffeomorphism & DGP & Diffeomorphism & DGP & Diffeomorphism & DGP \\
\midrule
$n = 50$ & \textbf{0.96} & 0.82 & \textbf{0.99} & 0.9 & \textbf{1.00} & 0.96 \\
$n = 100$ & \textbf{0.94} & 0.75 & \textbf{0.98} & 0.87 & \textbf{1.00} & 0.99 \\
$n = 200$ & \textbf{0.97} & 0.86 & \textbf{1.00} & 0.91 & \textbf{1.00} & 0.97 \\
$n = 300$ & \textbf{0.94} & 0.78 & \textbf{0.99} & 0.9 & \textbf{1.00} & 0.97 \\
\bottomrule
\end{tabular}
\caption{Coverage probabilities of 90\%, 95\%, and 99\% HPD intervals for Simulation 1. The diffeomorphism method consistently exceeds nominal coverage levels, whereas the DGP method tends to under-cover the true stationary point.}
\label{tab:sim1_coverage}
\end{table}

Across all simulations, the diffeomorphism method maintained coverage above nominal levels, while the DGP method often undercovered, particularly for smaller sample sizes. Figure~\ref{fig:sim1_hdp} illustrates the first fifty 90\% HPD intervals for $n = 100$. Intervals from the diffeomorphism method are generally wider and centered on the true stationary point, whereas the DGP method produces narrower intervals systematically shifted above 0.4. Occasionally, the DGP posterior exhibits excessive mass near the boundaries, generating extremely long HPD intervals, whereas the diffeomorphism method is more consistent.

\begin{figure}[htbp]
  \centering
  \begin{subfigure}[b]{0.48\textwidth}
    \includegraphics[width=\textwidth]{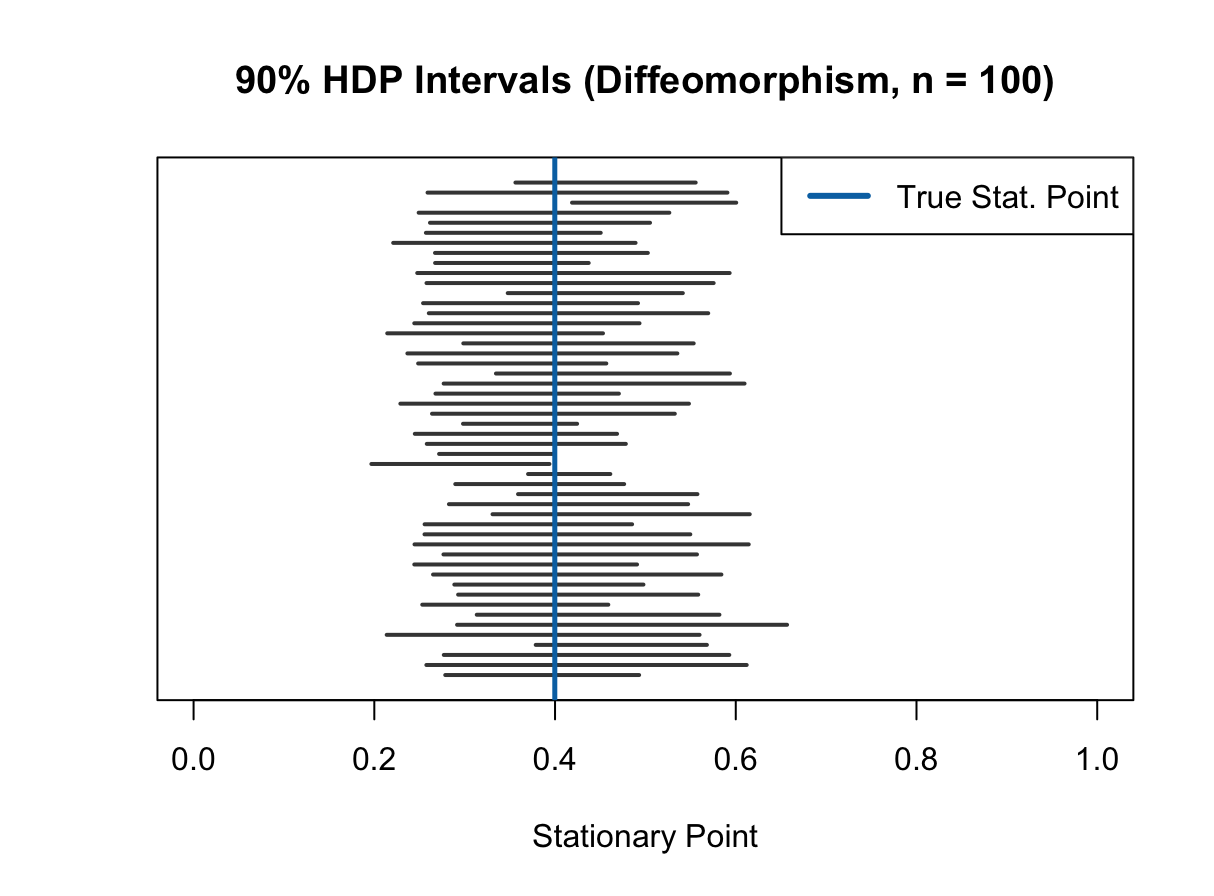}
    \caption{Diffeomorphism method}
  \end{subfigure}
  \begin{subfigure}[b]{0.48\textwidth}
    \includegraphics[width=\textwidth]{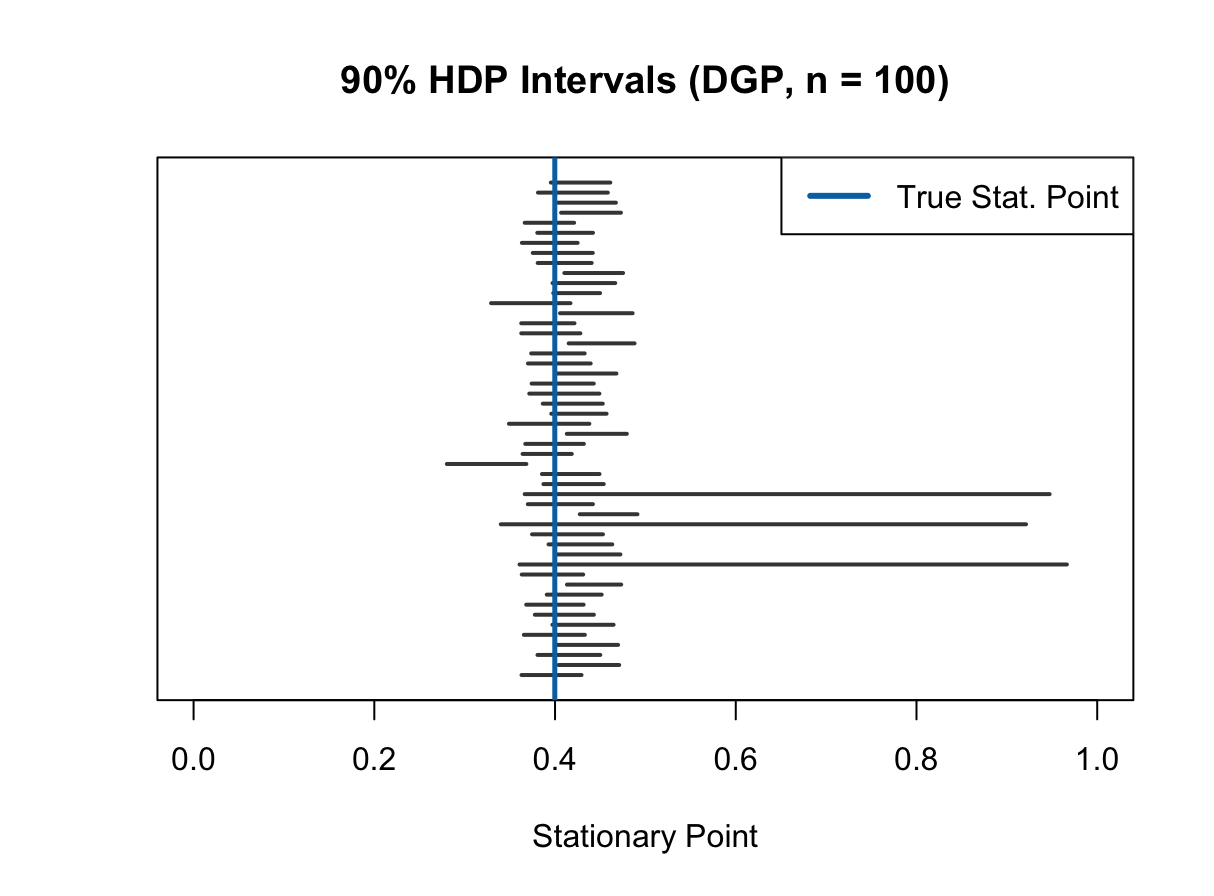}
    \caption{DGP method}
  \end{subfigure}
  \caption{First fifty 90\% HPD intervals of the stationary point for $n = 100$. The diffeomorphism method produces intervals consistently centered near the true stationary point, while the DGP intervals are narrower and slightly biased.}
  \label{fig:sim1_hdp}
\end{figure}

Table~\ref{tab:sim1_metrics} reports RMSE, mean bias, and average posterior standard deviation across sample sizes. While the DGP method exhibits lower RMSE and smaller posterior variance, there exists a noticeable positive bias for the smaller sample sizes. The diffeomorphism method displays higher variability, but its bias remains consistent across all sample sizes. Figure~\ref{fig:sim1_boxplots} shows boxplots of the posterior mean bias and posterior standard deviation for $n = 100$, highlighting the trade-off between variability and bias. For point estimation, the DGP method is generally preferred; however, for accurate uncertainty quantification, the diffeomorphism approach is superior.

\begin{table}[ht]
\centering
\footnotesize
\begin{tabular}{lcccccc}
\toprule
 & \multicolumn{2}{c}{RMSE} & \multicolumn{2}{c}{Mean Bias} & \multicolumn{2}{c}{Avg. Posterior SD} \\
\cmidrule(lr){2-3} \cmidrule(lr){4-5} \cmidrule(lr){6-7}
 & Diffeomorphism & DGP & Diffeomorphism & DGP & Diffeomorphism & DGP \\
\midrule
$n = 50$ & 0.0570 & \textbf{0.0540} & \textbf{0.0197} & 0.0386 & 0.0968 & \textbf{0.0656} \\
$n = 100$ & 0.0471 & \textbf{0.0311} & $\boldsymbol{9.52 \times 10^{-3}}$ &  0.0191 & 0.0813 & \textbf{0.0407} \\
$n = 200$ & 0.0413 & \textbf{0.0176} & 0.015 & $\boldsymbol{8.95 \times 10^{-3}}$ & 0.0673 & \textbf{0.0278} \\
$n = 300$ & 0.0378 & \textbf{0.0161} & 0.00899 & \textbf{0.00825} & 0.0577 & \textbf{ 0.0158} \\
\bottomrule
\end{tabular}
\caption{RMSE, mean bias, and average posterior standard deviation for Simulation 1. The DGP method produces tighter posteriors but exhibits systematic positive bias, while the diffeomorphism method shows higher variance with negligible bias.}
\label{tab:sim1_metrics}
\end{table}

\begin{figure}[htbp]
  \centering
  \begin{subfigure}[b]{0.48\textwidth}
    \includegraphics[width=\textwidth]{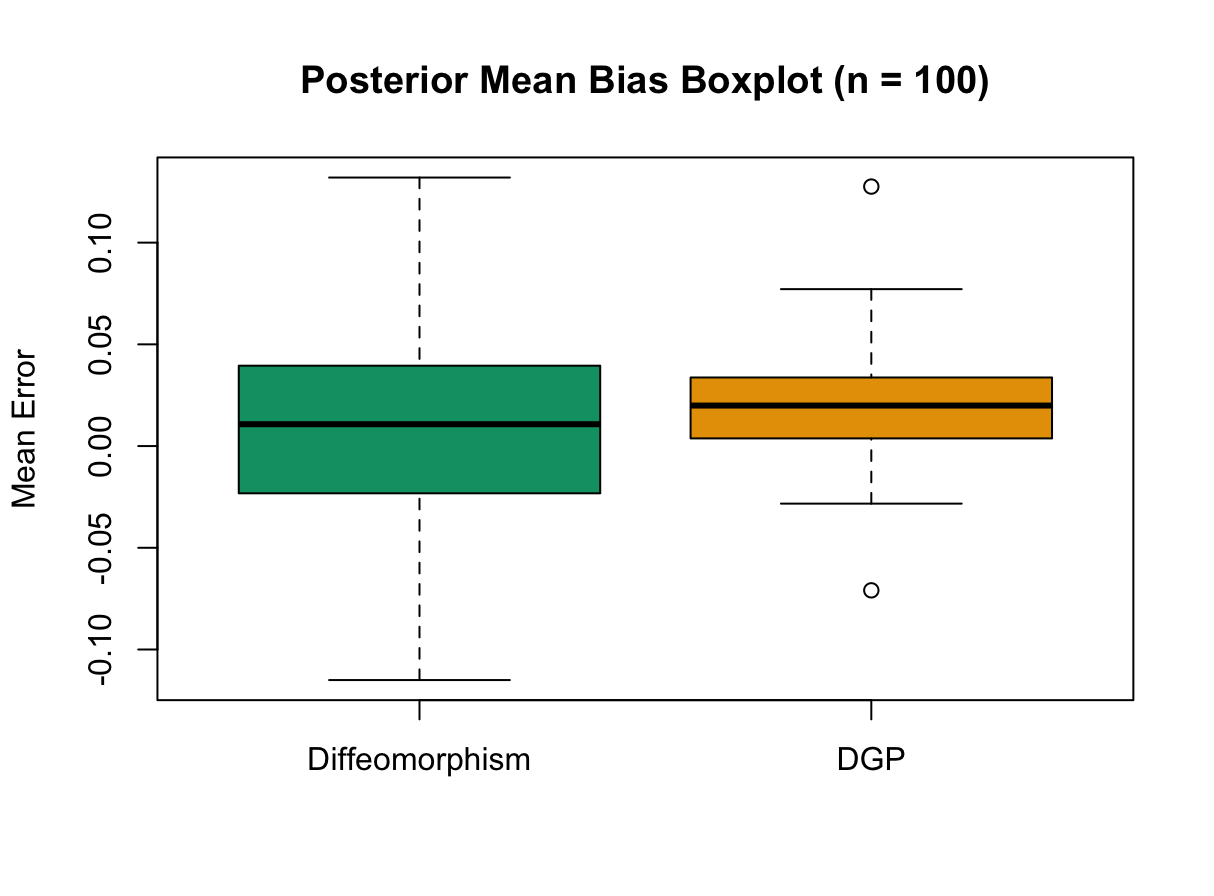}
    \caption{Posterior Mean Bias}
  \end{subfigure}
  \begin{subfigure}[b]{0.48\textwidth}
    \includegraphics[width=\textwidth]{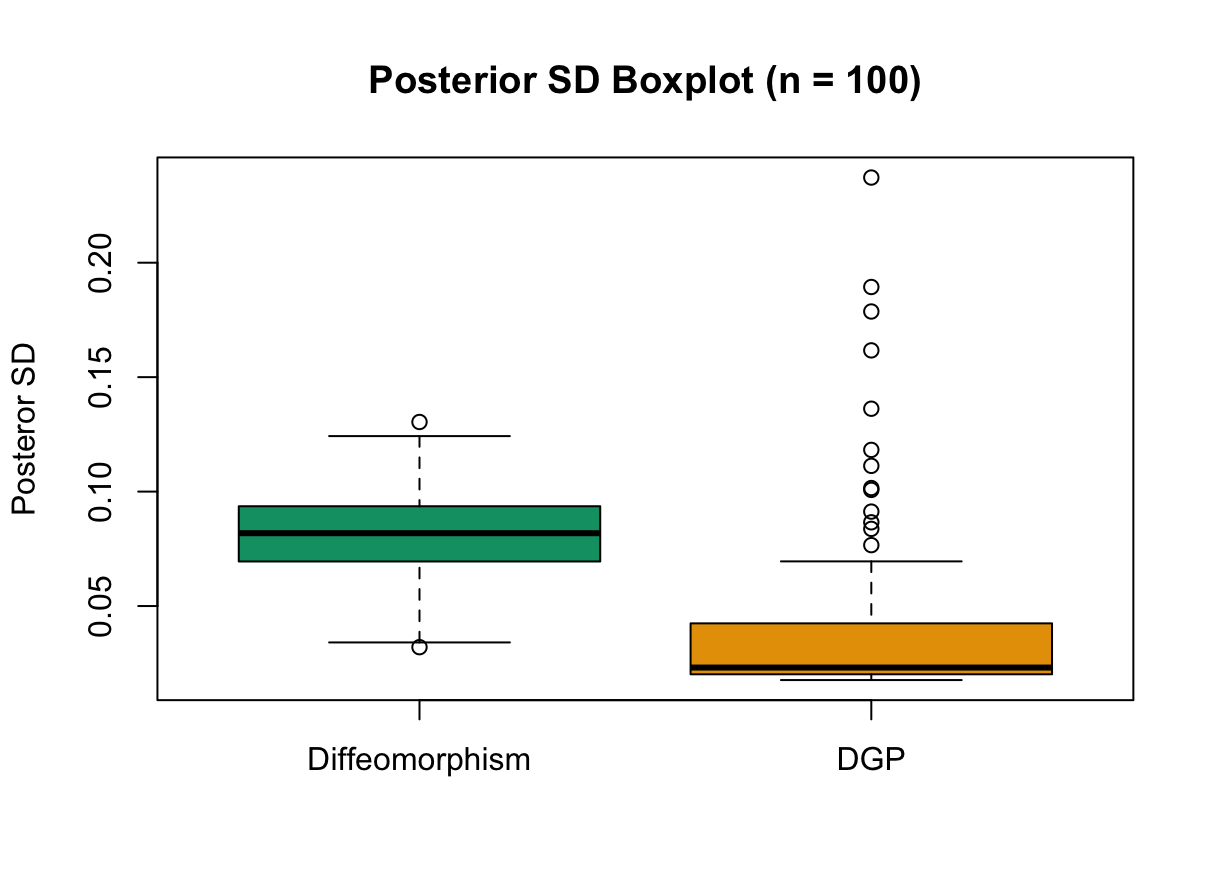}
    \caption{Posterior standard deviation}
  \end{subfigure}
  \caption{Boxplots of Posterior Mean bias and posterior standard deviation for $n = 100$ across 100 simulations. The diffeomorphism method exhibits wider posteriors with minimal bias, whereas the DGP method yields tighter posteriors but with positive bias.}
  \label{fig:sim1_boxplots}
\end{figure}

The increased variability yet consistent, low bias of the diffeomorphism method arises because the stationary point posterior is an explicit function of the model parameters, allowing each data point to influence the posterior. In contrast, the DGP method computes the stationary point posterior after fitting a Gaussian process, producing tighter posteriors that inherit bias from the fitted model. Variability in the diffeomorphism posterior can be reduced by estimating the template function first and treating it as fixed during posterior sampling of $\boldsymbol{\beta}$.

\subsection{Simulation with Two Stationary Points}

The second simulation investigates a function with two stationary points. Only the diffeomorphism method is considered, as it provides a joint posterior with dimensions corresponding directly to the stationary points, requiring no post-processing. The data-generating model is 
\begin{align*}
    X_{i} &\stackrel{i.i.d.}{\sim} \text{Unif}([0,1]),\\ 
    Y_{i} &= 1.4 X_{i} + \sin(3.3 X_{i}) + \cos(4 X_{i}) + \epsilon_{i}, \quad \epsilon_{i} \sim N(0, 0.15^2),
\end{align*}
with true stationary points at 0.24204 and 0.68735. Simulations were conducted 100 times for sample sizes $n = 50, 100, 200, 300$, with basis sizes 7, 8, 9, and 10, respectively. Coverage probabilities, RMSE, mean bias, and average posterior standard deviation were recorded for each stationary point (Table~\ref{tab:sim2_metrics}).

\begin{table}[ht]
\centering
\footnotesize
\begin{tabular}{lcccccc}
 \multicolumn{7}{c}{Stationary Point 1}  \\
\toprule
 & 90\% HPD & 95\% HPD & 99\% HPD & RMSE  & Mean Bias & Avg. Posterior SD \\
\midrule
$n = 50$  & 0.90 & 0.95 & 0.99 & 0.06530 & 0.04755 & 0.07459 \\
$n = 100$ & 0.90 & 0.95 & 0.98 & 0.04915 & 0.03704 & 0.05983 \\
$n = 200$ & 0.87 & 0.95 & 1 &  0.04186 & 0.02813 & 0.04822\\
$n = 300$ & 0.87 & 0.92 & 0.98 & 0.04048 & 0.02642 & 0.04416 \\
\bottomrule
\end{tabular}

\vspace{0.4em}

\begin{tabular}{lcccccc}
 \multicolumn{7}{c}{Stationary Point 2}  \\
\toprule
 & 90\% HPD & 95\% HPD & 99\% HPD & RMSE  & Mean Bias & Avg. Posterior SD \\
\midrule
$n = 50$  & 0.99 & 1.00 & 1.00 & 0.05179 & -0.02801 & 0.07192 \\
$n = 100$ & 0.92 & 0.95 & 1.00 & 0.04311 & -0.0222 & 0.0608 \\
$n = 200$ & 0.89 & 0.94 & 1.00 & 0.03946 & -0.01977 & 0.05015 \\
$n = 300$ & 0.86 & 0.93 & 0.98 & 0.03606 &  -0.01748 & 0.04459 \\
\bottomrule
\end{tabular}
\caption{Performance metrics for Simulation 2. Coverage probabilities remain near nominal levels, while RMSE and posterior standard deviations decrease with increasing sample size and basis size, reflecting posterior contraction and reduced bias.}
\label{tab:sim2_metrics}
\end{table}
In addition, the joint coverage of the 95\% Bonferroni corrected intervals were 0.98, 0.98, 0.97 and 0.94 for sample sizes $n = 50, 100, 200$ and $300$ respectively. 
All coverage probabilities remain near their nominal levels. As sample size and basis size increase, both RMSE and average posterior standard deviations decrease, indicating posterior contraction, while absolute mean bias also declines due to improved approximation of the function’s structure by higher-dimensional bases.

\section{Applied Data Analysis }
\label{sec:app_ds}
Electroencephalography (EEG) is a widely employed tool in cognitive neuroscience for measuring the brain’s electrical responses to specific stimuli. Due to high trial-to-trial variability, repeated stimulus presentation is necessary, and the resulting EEG segments are averaged to obtain an event-related potential (ERP). Peaks and troughs of an ERP correspond to distinct neural processes elicited by the stimulus; their latencies and amplitudes indicate the timing and magnitude of brain responses \citep{luck2012event}. Consequently, estimation and inference on the locations of stationary points are directly applicable in this context.

Following \cite{li2024semiparametric}, the EEG dataset available at \url{http://dsenturk.bol.ucla.edu/supplements.html} was analyzed. It comprises signals from a single individual with autism spectrum disorder (ASD) across 72 trials. The original experiment aimed to assess how children with ASD learn patterns in continuous sequences of colors and shapes, with ERP latencies and amplitudes serving as markers of implicit learning \citep{hasenstab2015identifying, jeste2015electrophysiological}. Each time series contains 250 observations, and the ERP is obtained by averaging across trials. While the N1 and P3 peaks, occurring between 100 and 250 milliseconds, are of primary interest, the analysis considers the entire ERP to demonstrate the method’s robustness.

The averaged ERP reveals approximately ten stationary points. A basis size of 15 was employed, with noninformative independent normal priors specified for the unconstrained parameters. Hermite cubic spline interpolation was used for the template family, with prior guesses for stationary point locations set at the node points of the template. Multiple Metropolis-Hastings chains were run, using scaled inverse Hessians at posterior modes as proposal covariances. Figure~\ref{fig:stacked} (top panel) displays the ERP data alongside the maximum a posteriori (MAP) estimate of the underlying signal. Color-coded segments indicate Bonferroni-adjusted 95\% highest posterior density (HPD) intervals for the stationary points. The bottom panel presents kernel density estimates of each stationary point’s marginal posterior distribution.

\begin{figure}[t!]
    \centering
    \includegraphics[width=0.9\linewidth]{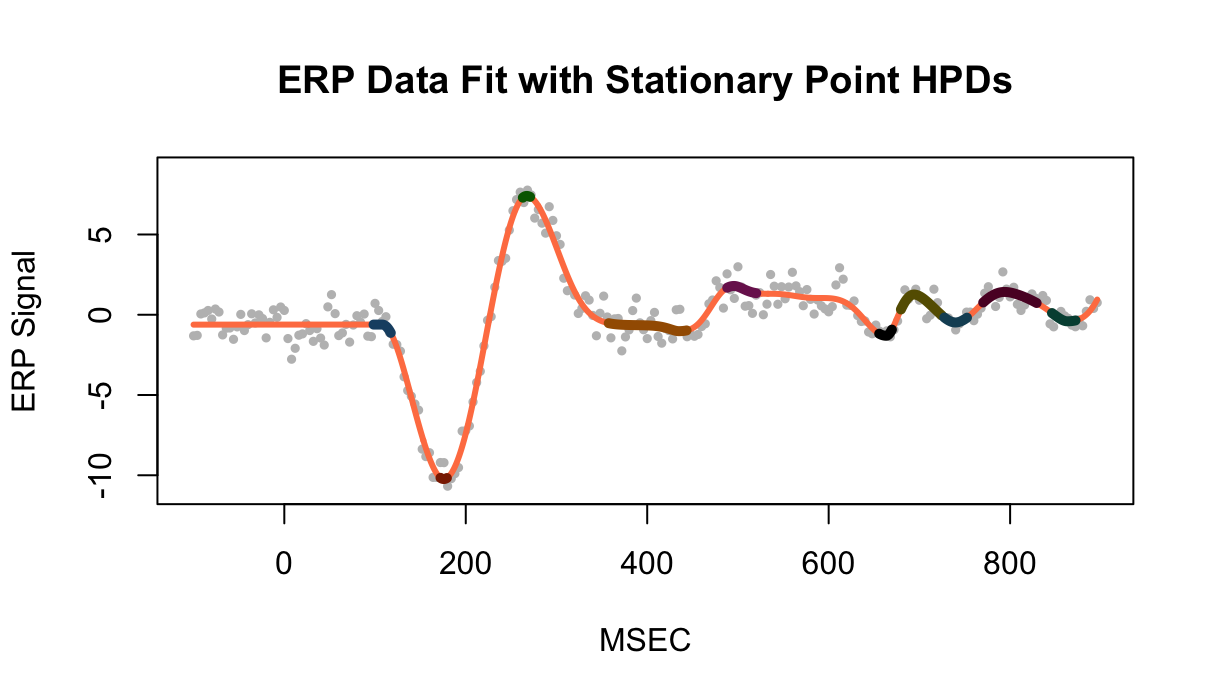}
    \vspace{1em}
    \includegraphics[width=0.9\linewidth]{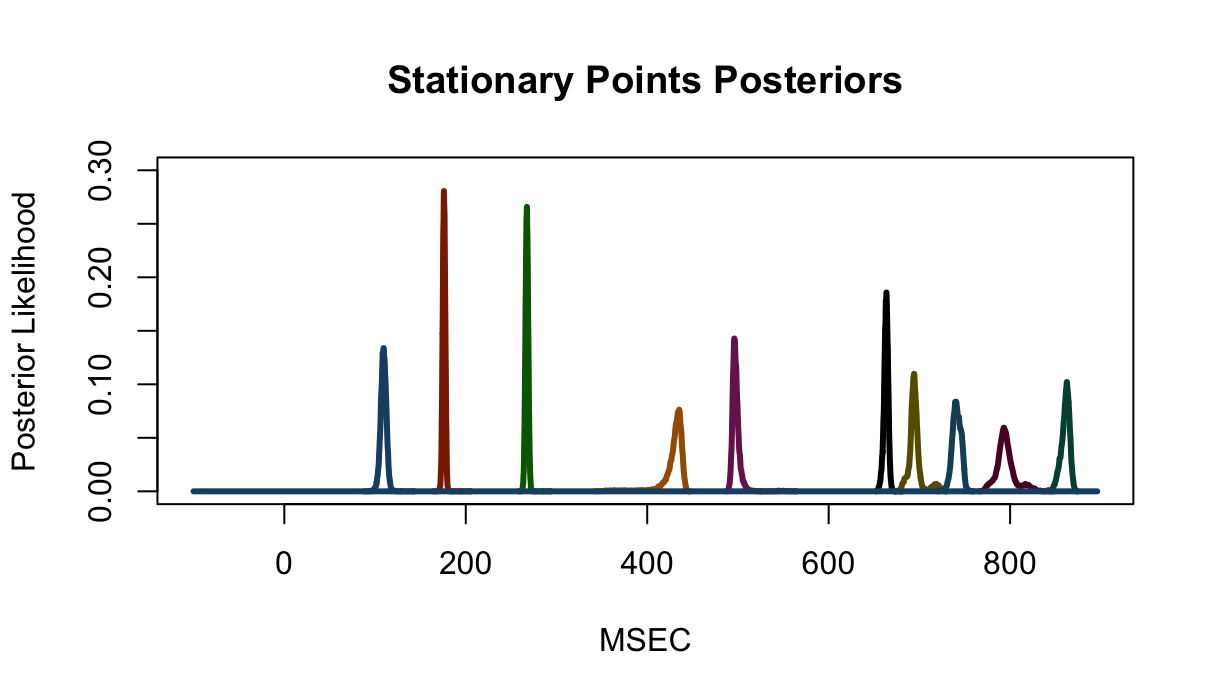}
    \caption{Estimation of ERP stationary points using the diffeomorphism-based method. 
    \textit{Top:} Original ERP data (gray) with the MAP estimate of the underlying signal (orange). 
    Colored segments indicate Bonferroni-adjusted 95\% HPD intervals for each stationary point. 
    \textit{Bottom:} Kernel density estimates of the marginal posterior distributions for each stationary point.}
    \label{fig:stacked}
\end{figure}

Fixing the number of stationary points enhances robustness, preventing the identification of spurious signals and isolating the most prominent features. Most marginal posteriors are approximately bell-shaped and concentrate near the true stationary points. Negligible posterior mass is assigned to flat regions between $-100$ and $99$ milliseconds, where derivatives are nearly zero, indicating a preference for genuine changepoints. Bonferroni-corrected HPD intervals for each stationary point are listed in Table~\ref{tab:sp_stacked_bars}.

\begin{table}[t!]
\centering
\renewcommand{\arraystretch}{1.2}

\begin{tabular}{lccccc}
\multicolumn{6}{c}{\textbf{Stationary Points 1–5}} \\[0.2em]
\hline
 & SP 1 & SP 2 & SP 3 & SP 4 & SP 5 \\
\hline
Lower & 98.3072  & 171.5270 & 262.8431 & 357.4698 & 488.1361  \\
Upper & 117.8728 & 180.0963 & 271.5853 & 444.2518 & 520.6181 \\
\hline
\end{tabular}

\vspace{1em}

\begin{tabular}{lccccc}
\multicolumn{6}{c}{\textbf{Stationary Points 6–10}} \\[0.2em]
\hline
 & SP 6 & SP 7 & SP 8 & SP 9 & SP 10 \\
\hline
Lower & 655.0799 & 679.1362 & 727.06263  & 770.00145 & 845.9251  \\
Upper & 669.8123 & 724.5678 & 753.0716 & 829.9239 & 872.2931 \\
\hline
\end{tabular}

\vspace{0.6em}
\caption{Bonferroni-adjusted 95\% highest posterior density (HPD) intervals for each stationary point estimated from the averaged ERP signal. Stationary points are labeled SP1–SP10.}
\label{tab:sp_stacked_bars}
\end{table}

For comparison, \citep{li2024semiparametric} analyzed the 100–350 millisecond window containing N1 and P3, constructing 95\% HPD intervals for the second (N1) and third (P3) stationary points. Their intervals were $[174.58, 178.32]$ and $[266.58, 270.93]$, whereas the diffeomorphism-based method produced $[173.39, 178.64]$ and $[263.19, 268.24]$. The slightly wider intervals observed here reflect estimation across the entire ERP, illustrating the method’s ability to accurately recover local changepoints without segment restriction.

\section{Conclusion}
\label{sec:conclusion}
This work introduces a template-diffeomorphism framework for estimating univariate functions and their stationary points. The approach enables joint inference on stationary point locations without identifiability or label-switching issues. The model admits a parametric representation in which stationary points are smooth functions of the parameters, allowing application of high-dimensional, non-asymptotic theory to derive finite-sample confidence bounds for both parameters and stationary points. Although theoretical analysis focused on the frequentist maximum likelihood estimator, the framework naturally extends to Bayesian inference, as demonstrated in simulations.

Cosine bases were employed to parameterize the diffeomorphism; alternative bases such as Legendre polynomials or wavelets may be advantageous for capturing abrupt transitions or stationary points distant from the template. The parameter space is typically multimodal, exhibiting irregular likelihood surfaces. A basic MCMC and optimization strategy were sufficient here, but advanced computational methods for uncertainty quantification in non-concave settings warrant further investigation.

The theoretical results extend broadly: if the model is identifiable, gradients are bounded, and the parameter space is compact, finite-sample confidence bounds remain valid even under non-concave likelihoods. In the growing parameter setting, no assumptions on vector evolution with dimension were imposed, though structural assumptions could yield additional insight into MLE and posterior behavior. Additionally, while the template function was treated as fixed, further theory is required for finite-sample error bounds when the template itself is estimated, as such functions lie outside standard nonparametric classes.

Overall, the template-diffeomorphism framework offers a flexible and powerful tool for nonparametric regression. The composition inherits derivative properties from the template, while the diffeomorphism corrects for model misspecification. Diffeomorphic representations thus provide a promising avenue for shape-constrained nonparametric regression.

\clearpage

\bibliography{refs}

\appendix

\section{Bayesian Case}
\label{sec:bayesian_case}
We briefly discuss how to obtain finite-sample error bounds for the Bernstein-von-Mises approximation for the posterior of $\boldsymbol{\beta}$. Our approach follows the framework  developed by \cite{PanovSpokoiny2015_FiniteSampleBvM}, in which they extend the results of \cite{Spokoiny2012_ParametricEstimation} to the Bayesian case and allow the presence a nonparametric nuisance parameter. We adopt the same assumptions from Section \ref{sec:theory}, where $\boldsymbol{\beta^*} \in \mathbb{R}^{p}$ and the only estimable parameter. Therefore, all conditions the i.i.d. theory of \cite{Spokoiny2012_ParametricEstimation}, and consequently those of \cite{PanovSpokoiny2015_FiniteSampleBvM}, will be satisfied, and we apply their results directly. 

The primary quantity of interest is $\Delta(r_{0}, x)$, which generally defines the quality of the approximation of the expected log-likelihood with a second order Taylor expansion. We define $r_{0} = \sqrt{n}u_{0}$, with $u_{0}$ defined as in Definition \ref{def:local_sets} and $x$ is taken to be a large positive number. As stated in Theorem 4 in \cite{PanovSpokoiny2015_FiniteSampleBvM}, $\Delta(r_{0}, x)$ is generally of the order $Cp^3/n$, where $C$ depends on $x$, in the i.i.d. case. This leads directly to the following theorem, which measures how well the normalized posterior matches a standard Gaussian distribution for a finite sample $n$. 
\begin{theorem}
\label{thm:bayesian}
  Let the prior be uniform on $\boldsymbol{\beta} \in \boldsymbol{\Theta}$, $\bar{\boldsymbol{\beta}}_{n}$ be the posterior mean and $\boldsymbol{V}_{n}^2$ be the posterior covariance matrix at sample size $n$. Then, there exists a data dependent event $\Omega( x)$ such that $\mathbb{P}(\Omega( x)) \geq 1 - 4e^{-x}$ and \begin{align*}
        \exp\left(-2\Delta(r_{0}, x) - 8e^{-x} \right)\left(\mathbb{P}\left(\boldsymbol{\gamma} \in A \right) - \tau \right) - e^{-x} \leq & \mathbb{P}\left( \boldsymbol{V}_{n}^{-1}(\boldsymbol{\beta} - \bar{\boldsymbol{\beta}}_{n}) \in A \ \middle| \ \boldsymbol{Y} \right) \\
        \leq &  \exp\left(-\Delta(r_{0}, x) + 5e^{-x} \right)\left(\mathbb{P}\left(\boldsymbol{\gamma} \in A \right) + \tau \right), 
    \end{align*}
where $\boldsymbol{Y}$ is the observed data, $A$ is any measurable set in $\mathbb{R}^{p}$, $\boldsymbol{\gamma}$ is a standard Gaussian vector in $\mathbb{R}^{p}$, and $\tau$ is a function of $\Delta(r_{0}, x)$ and $p$.
\end{theorem}

If $x$ and $n$ are large such that $e^{-x}$,  $\Delta(r_{0}, x)$ and $\tau$ are near $0$, then for almost any observed dataset $\boldsymbol{Y}$, $\mathbb{P}\left( \boldsymbol{V}_{n}^{-1}(\boldsymbol{\beta} - \bar{\boldsymbol{\beta}}_{n}) \in A \ \middle| \ \boldsymbol{Y} \right) \approx \mathbb{P}\left(\boldsymbol{\gamma} \in A \right)$. Thus, the rate at which the posterior is approximately normal is dependent on fast $\Delta(r_{0}, x)$ converges to 0. This result naturally extends to the growing parameter case, where $p$ is now a function of $n$. The result above still holds as long as $\Delta(r_{0}, x) = C_{p(n)}\left(p(n)^3/n\right)$ goes to $0$ as $p(n) \rightarrow \infty$. Finite-sample bounds can also be derived to quantify how well the posterior mean approximates the true parameter $\boldsymbol{\beta^*}$. Combining these bounds with Theorem \ref{thm:bayesian}, we can then apply an argument analogous to that of Theorem \ref{thm:delta_method} to obtain a finite-sample error bound for the normal approximation of the posterior distribution of stationary points.

\newpage
\section{Supplementary Materials}

\renewcommand{\theequation}{S.\arabic{equation}}
\setcounter{equation}{0}

\startcontents[supp]
\printcontents[supp]{l}{1}{\section*{Supplementary Table of Contents}}

\subsection{Proof of Theorem~\ref{thm:thm_1}}
\label{sec:proof_thm1}
\begin{proof}
Let $f, g \in \mathcal{F}_{\lambda, M}$. We will show that there exists a continuous, bijective mapping $\gamma: [0,1] \to [0,1]$ such that $g \circ \gamma = f$.  Let $b_{1}, \ldots, b_{M}$ be the stationary points of $g$ and $b^*_{1}, \ldots, b^*_{M}$ be the stationary points of $f$. Define $b_{0} = b^*_{0} = 0$ and $b_{M+1} = b^*_{M+1} = 1$. For each $k = 0, \ldots, M$, note that both $g$ and $f$ are monotonic on $[b_k, b_{k+1}]$ and $[b^*_k, b^*_{k+1}]$, respectively, and their images are $[\lambda_k, \lambda_{k+1}]$. Since $g$ is monotonic on $[b_k, b_{k+1}]$, its restriction to this interval is invertible, and we denote this inverse by $g_{k}^{-1}: [\lambda_k, \lambda_{k+1}] \to [b_k, b_{k+1}]$. Define the function $\gamma: [0,1] \to [0,1]$ by
\begin{center}
    $\gamma(x) = g_{k}^{-1}(f(x)), \quad \text{if } x \in [b^*_k, b^*_{k+1}], \quad \text{for } k = 0, \ldots, M.$
\end{center}

Clearly, for each $k = 0, \ldots M$, if $x \in [b^*_k, b^*_{k+1}]$, then $g \circ \gamma(x) =  g \circ g_{k}^{-1}(f(x)) = f(x)$. Thus, for all $x \in [0, 1]$, $g \circ \gamma(x) = f(x)$. 
Next, since both $g$ and $f$ are continuous functions, $g_{k}^{-1}(f(x))$ is continuous on the interval $(b^*_{k}, b^*_{k+ 1})$, right continuous at $b^*_{k}$ and left continuous at $b^*_{k+ 1}$ for each $k$. In addition, clearly $b_{k + 1} = g_{k}^{-1}(f(b^*_{k + 1})) = g_{k + 1}^{-1}(f(b^*_{k + 1}))$  for all $k = 0, ..., M - 1$. Thus, because at these endpoints the piecewise functions are equal and continuous, $\gamma(x)$ must also be continuous on $[0, 1]$. 

Finally, for $k = 0, \ldots, M$, $f$ is a bijective mapping from $[b_{k}^{*}, b_{k + 1}^{*})$ to $[\lambda_{k}, \lambda_{k + 1})$ and $g_{k}^{-1}(x)$ is a bijective mapping from $[\lambda_{k}, \lambda_{k + 1})$ to $[b_{k}, b_{k+ 1})$ for  $k = 0, \ldots, M$. Because the composition of bijective functions is bijective, $\gamma(x)$ is a bijective mapping from $[b_{k}^{*}, b_{k + 1}^{*})$ to $[b_{k}, b_{k + 1})$. As  $\bigcap\limits_{k = 0}^{M} [b_{k}^{*}, b_{k + 1}^{*}) = \emptyset$ and $\bigcap\limits_{i = 0}^{M} [b_{k}, b_{k + 1}) = \emptyset$, $\gamma(x)$ is also a bijective mapping from $[0, 1)$ to $[0, 1)$. Finally, $\gamma(1) = g_{M}^{-1}(f(1)) = g_{M}^{-1}(\lambda_{M + 1}) = 1$, showing we have found a continuous, bijective mapping $\gamma: [0,1] \to [0,1]$ such that $g \circ \gamma = f$. 
\end{proof}

\subsection{Bounding the Gradient of $\gamma_{\boldsymbol{\beta}}$ and Proof of Lemma~\ref{lem:gradient_lem}}
\label{sec:gradient_beta}

Define the function space \begin{align}
    V = \left\{ v \in \mathbb{L}^{2}([0, 1]) : \left\lVert v \right\lVert < \pi, \int_{0}^{1} v(t)dt = 0\right\},
\end{align}
 where $\left\lVert  \,\cdot\,  \right\lVert$ is the $\mathbb{L}^{2}$ norm. \cite{dasgupta2020two} showed that any diffeomorphisms of $[0, 1]$ $\gamma$ can be represented as \begin{align}
     \gamma(t) = \int_{0}^{t}\left[\cos(\left\lVert v \right\lVert) + \frac{\sin(\left\lVert  v \right\lVert)}{\left\lVert  v \right\lVert} v(s)\right]^{2} dt.
 \end{align}
Because $v$ is infinite dimensional, we instead approximate it with a $p$ dimensional subset $V^p$.  Let $\boldsymbol{\beta} = (\beta_{1}, ..., \beta_{p}) \in \mathbb{R}^{p}$ where $ \left\lVert\boldsymbol{\beta}\right\lVert < \pi $, then any element of $V^{p}$ can be represented as the sum of orthogonal basis functions, $v_{\boldsymbol{\beta}}(s) = \sum\limits_{j = 1}^{p} \beta_{j} b_{j}(s) $. Thus, the corresponding diffeomorphism is equal to \begin{align*}
\gamma_{\boldsymbol{\beta}}(t)= \cos^2(\left\lVert \boldsymbol{\beta}\right\lVert)t + 2\frac{\cos(\left\lVert\boldsymbol{\beta} \right\lVert)\sin(\left\lVert \boldsymbol{\beta}\right\lVert)}{\left\lVert \boldsymbol{\beta}\right\lVert}\int_{0}^{t}v_{\boldsymbol{\beta}}(s)ds + \frac{\sin^2(\left\lVert\boldsymbol{\beta}\right\lVert)}{\left\lVert\boldsymbol{\beta}\right\lVert^2}\int_{0}^{t}v^{2}_{\boldsymbol{\beta}}(s)ds. 
\end{align*} 
For the rest of this work, we choose the cosine basis, where $b_{j}(t) = \sqrt{2}\cos(j \pi t)$. Thus, \begin{align*}
    \int_{0}^t v_{\boldsymbol{\beta}}(s)ds = \sqrt{2} \sum\limits_{j = 1}^{p}\beta_{j}\frac{\sin(\pi j t)}{\pi j}
\end{align*}
and  
\begin{align*}
    \int_{0}^{t}v^{2}_{\boldsymbol{\beta}}(s)ds
    = 2\left[ \sum\limits_{j = 1}^{p} \beta_{j}^{2}\left (\frac{\sin(2\pi j t)}{4 \pi j} + \frac{t}{2}\right) + \sum\limits_{j \neq k} \beta_{j}\beta_{k} \left( \frac{\sin(\pi(k + j) t)}{2\pi(k + j)} + \frac{\sin(\pi(k - j) t)}{2\pi(k - j)}   \right) \right].
\end{align*}
Now, we will need to establish upper bounds for the integral, which will be used to prove that the gradient $\nabla_{\boldsymbol{\beta}}f_{\boldsymbol{\beta}}(x)$ is bounded for all $\boldsymbol{\beta} \in \boldsymbol{\Theta}$ and $x \in [0, 1]$. First,   \begin{equation}
\label{eq:bound_1}
    \left|\int_{0}^{t}v_{\boldsymbol{\beta}}(s)ds\right| \leq   \sqrt{2}\sum\limits_{j = 1}^{p}|\beta_{j}| \left| \frac{\sin(\pi j t)}{\pi j}\right|   \leq   \sqrt{2}\pi \sum\limits_{j = 1}^{p}\left| \frac{1}{\pi j}\right| 
    \leq   \sqrt{2}(1 + \ln(p)).
\end{equation}
Next, similarly, \begin{align*}
    \left|\sum\limits_{j = 1}^{p}\beta_{j}^{2}\left (\frac{\sin(2\pi j t)}{4 \pi j} + \frac{t}{2}\right)\right|  \leq   \sum\limits_{j = 1}^{p} |\beta_{j}^{2}| \left ( \left| \frac{\sin(2\pi j t)}{4 \pi j}\right| + \left|\frac{t}{2}\right|\right)   \leq  \pi^{2}\left(\frac{1}{4\pi}(1 + \ln(p))  +\frac{p}{2} \right).
\end{align*}
In addition, if $j$ is kept constant, \begin{align*}
    \left| \sum\limits_{j \neq k} \beta_{j}\beta_{k} \left( \frac{\sin(\pi(k + j) t)}{2\pi(k + j)} + \frac{\sin(\pi(k - j) t)}{2\pi(k - j)}   \right)\right| 
    \leq&  \pi^2 \sum\limits_{ k \neq j }\left| \frac{1}{2\pi( k + j)}   \right| + \left|  \frac{1}{2\pi( k -j )} \right| \\
    \leq & \frac{\pi}{2}\sum\limits_{k = 1}^{p}\left|\frac{1}{k}\right| + \sum\limits_{k > j} \left|\frac{1}{k - j}\right| \\
    \leq & \pi(2 + 2\ln(p)).
\end{align*}
Thus, we can bound the integral of $v^{2}_{\boldsymbol{\beta}}(s)$ by a function of $p$: \begin{equation}
\label{eq:bound_2}
   \left| \int_{0}^{t}v^{2}_{\boldsymbol{\beta}}(s) ds\right| \leq\pi^{2}\left(\frac{1}{4\pi}(1 + \ln(p))  +\frac{p}{2}  \right) +   p\pi(2 + 2\ln(p)).
\end{equation}

Now, we consider $\nabla_{\boldsymbol{\beta}}\ f_{\boldsymbol{\beta}}(x)$, which equals
  \begin{align*}
g'\!\bigl(\gamma_{\boldsymbol{\beta}}(x)\bigr) \Biggl[
\nabla_{\boldsymbol{\beta}}\!\left(
2\,\frac{\cos\!\bigl(\|\boldsymbol{\beta}\|\bigr)\sin\!\bigl(\|\boldsymbol{\beta}\|\bigr)}{\|\boldsymbol{\beta}\|}
      \right)
      \left( \int_{0}^{t} v_{\boldsymbol{\beta}}(s)\,ds \right) +
2\,\frac{\cos\!\bigl(\|\boldsymbol{\beta}\|\bigr)\sin\!\bigl(\|\boldsymbol{\beta}\|\bigr)}{\|\boldsymbol{\beta}\|}
      \left( \nabla_{\boldsymbol{\beta}} \int_{0}^{t} v_{\boldsymbol{\beta}}(s)\,ds \right) +  \\[0.8em] \nabla_{\boldsymbol{\beta}}\!\left(
        \frac{\sin^{2}\!\bigl(\|\boldsymbol{\beta}\|\bigr)}{\|\boldsymbol{\beta}\|^{2}}
      \right)
      \left( \int_{0}^{t} v_{\boldsymbol{\beta}}^{2}(s)\,ds \right) + \frac{\sin^{2}\!\bigl(\|\boldsymbol{\beta}\|\bigr)}{\|\boldsymbol{\beta}\|^{2}}
      \left( \nabla_{\boldsymbol{\beta}} \int_{0}^{t} v_{\boldsymbol{\beta}}^{2}(s)\,ds \right) +\nabla_{\boldsymbol{\beta}} \cos^{2}\!\bigl(\|\boldsymbol{\beta}\|\bigr)\,  
\Biggr].
\end{align*}
We would like to upper bound this gradient by a function of $p$. Considering the fact that $\left\lVert \boldsymbol{\beta} \right\lVert < \pi$ and that $\left\lVert  \frac{\boldsymbol{\beta}}{\left\lVert  \boldsymbol{\beta}  \right\lVert} \right\lVert = 1$, we can provide an upper bound for each of the below components:
\begin{align}
\label{eq:bound_3}
2\,\frac{\cos\!\left(\|\boldsymbol{\beta}\|\right)\sin\!\left(\|\boldsymbol{\beta}\|\right)}{\|\boldsymbol{\beta}\|}\leq2, \quad
    \frac{\sin^{2}\!\left(\|\boldsymbol{\beta}\|\right)}{\|\boldsymbol{\beta}\|}  \leq 1, \quad 
    \nabla_{\boldsymbol{\beta}} \cos^{2}\!\left(\|\boldsymbol{\beta}\|\right)\leq  1,
\end{align}
\begin{align}
    \label{eq:bound_4}
    \nabla_{\boldsymbol{\beta}} \!\left( 2\,\frac{\cos\!\left(\|\boldsymbol{\beta}\|\right)\sin\!\left(\|\boldsymbol{\beta}\|\right)}{\|\boldsymbol{\beta}\|} \right)\leq 1.7448, \quad \text{and}\quad
    \nabla_{\boldsymbol{\beta}} \!\left( \frac{\sin^{2}\!\left(\|\boldsymbol{\beta}\|\right)}{\|\boldsymbol{\beta}\|^{2}} \right) \leq  0.5402.
\end{align}
For the integrals' gradient, $ \left(\nabla_{\boldsymbol{\beta}} \int_{0}^{t} v_{\boldsymbol{\beta}}(s)ds\right)_{j}  \leq \frac{\sqrt{2}}{\pi j }$, which implies  \begin{equation}
\label{eq:bound_5}
         \left\lVert \nabla_{\boldsymbol{\beta}} \int_{0}^{t} v_{\boldsymbol{\beta}}(s)\,ds \right\rVert \leq \sqrt{\frac{2}{\pi^2}\sum\limits_{j = 1}^{\infty} \frac{1}{j^2}} = \frac{1}{\pi}
\end{equation} 
and \begin{align*}
      \left|\nabla_{\boldsymbol{\beta}} \int_{0}^{t} v^{2}_{\boldsymbol{\beta}}(s)ds\right|_{j} \leq &2\left[2\pi\left(\frac{1}{4 \pi} + \frac{1}{2}\right) + \pi \sum\limits_{j \neq k}\left| \frac{1}{2\pi( k + j)}   \right| + \left|  \frac{1}{2\pi( k -j )} \right| \right] \\
      \leq & 2\left[2\pi\left(\frac{1}{4 \pi}  +  \frac{1}{2}\right)  +  1 + \ln(p) \right].
\end{align*}
Therefore, by the Cauchy-Schwarz inequality,  \begin{equation}
\label{eq:bound_6}
    \left\lVert  \nabla_{\boldsymbol{\beta}} \int_{0}^{t} v_{\boldsymbol{\beta}}(s)ds  \right\rVert \leq 2 \sqrt{p}\left[2\pi\left(\frac{1}{4 \pi}  +  \frac{1}{2}\right)  +  1 + \ln(p) \right].
\end{equation}

By Assumption \ref{ass:CS}, $g$ is a continuously differentiable function and thus, the derivative $g'\!\bigl(\gamma_{\boldsymbol{\beta}}(x)\bigr)$ is bounded on $x \in [0, 1]$ by the extreme value theorem. In addition, each term in the gradient of $f_{\boldsymbol{\beta}}(x)$ can be bounded by a function of $p$, as shown in equations  (\ref{eq:bound_1}), (\ref{eq:bound_2}), (\ref{eq:bound_3}), (\ref{eq:bound_4}), (\ref{eq:bound_5}) and (\ref{eq:bound_6}), showing the gradient's norm is bounded for all $x \in [0, 1]$ and $\boldsymbol{\beta} \in \boldsymbol{\Theta}$. The dominating term comes from $\left| \int_{0}^{t}v^{2}_{\boldsymbol{\beta}}(s)ds \right|$, which is $p\ln(p)$. Thus, $\left\lVert \nabla_{\boldsymbol{\beta}} f_{\boldsymbol{\beta}}(x) \right\lVert \leq Cp\ln(p)$, where $C > 0$.

Finally, we need to show that $f$ is thrice continuously differentiable with respect to $\boldsymbol{\beta}$ on the set $\boldsymbol{\Theta}$. First, clearly, $\int_{0}^{t}v_{\boldsymbol{\beta}}(s)ds$ and $\int_{0}^{t}v^{2}_{\boldsymbol{\beta}}(s)ds $ are three times differentiable with respect to $\boldsymbol{\beta}$, as $v_{\boldsymbol{\beta}}(s)$ is simply a linear combination of each of the $\boldsymbol{\beta}$ components ($\boldsymbol{\beta}_{j}$). Now, the first, second and third differential of the exponential mapping and the norm of $\boldsymbol{\beta}$ will contain terms with $\left\lVert \boldsymbol{\beta} \right\lVert$ in the denominator. However, because $\boldsymbol{\Theta}$ does not include $\boldsymbol{\beta} = 0$, $\left\lVert\boldsymbol{\beta} \right\lVert$ will never equal 0. Thus, the first, second and third differential of $f$ with respect to $\boldsymbol{\beta}$ will be continuous for all $x \in [0, 1]$. This leads to the following lemma, which will be helpful in deriving finite sample error bounds. \setcounter{lemma}{0} \begin{lemma}
\label{lem:gradient_lem}
 Suppose Assumptions~\ref{ass:CS}-\ref{ass:Identifiability} hold.
       Then $f_{\boldsymbol{\beta}}(x) = g(\gamma_{\boldsymbol{\beta}}(x))$ is a thrice continuously differentiable function with respect to $\boldsymbol{\beta}$ for all $x \in [0, 1]$. In addition, the norm of the $\nabla_{\boldsymbol{\beta}} f_{\boldsymbol{\beta}}(x)$ is bounded by the function  $C(p \ln p)$ where $C > 0$.
\end{lemma} 

\begin{proof}
    Clearly, $g$, $g'$, $g''$ and $g'''$, are continuous on $[0, 1]$ by Assumption \ref{ass:CS}. The components of the first, second and third differential with respect to $\boldsymbol{\beta}$ involve sums and products of the differentials of $\int_{0}^{t}v_{\boldsymbol{\beta}}(s)ds$, $\int_{0}^{t}v^{2}_{\boldsymbol{\beta}}(s)ds $, the exponential mapping, $\left\lVert \boldsymbol{\beta} \right\lVert$ and $g$, which are all continuous with respect to $\boldsymbol{\beta} \in \boldsymbol{\Theta}$ for all $x \in [0, 1]$.  As sums and products of continuous functions are continuous, $f_{\boldsymbol{\beta}}(x)$ is a thrice continuously differentiable function with respect to $\boldsymbol{\beta}$ on $\boldsymbol{\Theta}$. In addition, we showed earlier that the norm of the gradient can be bounded by $Cp\ln(p)$ where $C > 0$. 
\end{proof}

\subsection{Proofs of Lemmas~\ref{lem:subexp} and~\ref{lem:kappa} }
\label{sec:lemmas}

The next two sections detail several lemmas which will help us show the i.i.d. conditions from \cite{Spokoiny2012_ParametricEstimation}, which will help derive finite sample bounds for the maximum likelihood estimator. We first provide proofs of Lemmas~\ref{lem:subexp} and \ref{lem:kappa} from the main manuscript. First, we define \begin{align}
   \zeta(\boldsymbol{\beta}) = L(\boldsymbol{\beta}) - \mathbb{E}(L(\boldsymbol{\beta})) = \sum\limits_{i = 1}^{n} \ell_{i}(\boldsymbol{\beta}) - \mathbb{E}(\ell_{i}(\boldsymbol{\beta})),
\end{align}
where $L(\boldsymbol{\beta})$ is the log-likelihood and $\ell_{i}(\boldsymbol{\beta})$ is the log-likelihood for an individual data point. For our model, \begin{align}
    \ell_{i}(\boldsymbol{\beta}) = & -\frac{1}{2}\ln(2 \pi \sigma^2) - \frac{1}{2\sigma^{2}}(y_i - f_{\boldsymbol{\beta}}(x_{i}))^{2} + \text{const}, \\
    \nabla  \ell_{i}(\boldsymbol{\beta}) = & \ \sigma^{-2}(y_{i} - f_{\boldsymbol{\beta}}(x_{i}))\nabla  f_{\boldsymbol{\beta}}(x_{i}).
\end{align}
The constant refers to distribution of $X$,  which is not dependent on the value of $\boldsymbol{\beta}$ (Assumption \ref{ass:PI_Px}). By Assumption~\ref{ass:CS}, $\mathbb{E}(y_{i} | x_{i}) = f_{\boldsymbol{\beta}^*}(x_{i})$. Clearly, by Lemma~\ref{lem:gradient_lem}, $\ell_{i}(\boldsymbol{\beta})$ is a twice continuously differentiable function. Therefore, $\nabla \mathbb{E}(l_{i}(\boldsymbol{\beta})) =  \mathbb{E}( \nabla l_{i}(\boldsymbol{\beta}))$. Plugging in the value of $\nabla  \ell_{i}(\boldsymbol{\beta})$ into the expectation and using the law of iterated expectations, we can conclude: \begin{align*}
    \nabla \mathbb{E}(l_{i}(\boldsymbol{\beta})) =\mathbb{E}\left[ \sigma^{-2}(f_{\boldsymbol{\beta}^*}(x_{i}) - f_{\boldsymbol{\beta}}(x_{i}))\nabla  f_{\boldsymbol{\beta}}(x_{i})  \right].
\end{align*}
Now, the two quantities of interest in Spokoiny's theory are \begin{align}
    \nabla \zeta_{i}(\boldsymbol{\beta}) 
    =   \sigma^{-2}(y_{i} - f_{\boldsymbol{\beta}}(x_{i}))\nabla  f_{\boldsymbol{\beta}}(x_{i}) - \mathbb{E}\left[ \sigma^{-2}(f_{\boldsymbol{\beta}^*}(x_{i}) - f_{\boldsymbol{\beta}}(x_{i}))\nabla  f_{\boldsymbol{\beta}}(x_{i})  \right], 
\end{align}
and the KL-Divergence between the models parameterized by $\boldsymbol{\beta}$ and $\boldsymbol{\beta}^*$,
\begin{align}
    \mathbb{K}(\boldsymbol{\beta}, \boldsymbol{\beta}^*) = & -\mathbb{E}\left[ \ell_{i}(\boldsymbol{\beta}) - \ell_{i}(\boldsymbol{\beta}^*)   \right] 
    =  \frac{1}{2\sigma^2}\mathbb{E}\left[ (f_{\boldsymbol{\beta}}(x)  - f_{\boldsymbol{\beta^*}}(x))^2  \right ].
\end{align}
The first lemma of this section shows that $ \nabla \zeta_{i}(\boldsymbol{\beta})$ is a sub-Gaussian (therefore sub-exponential) random vector for all $\boldsymbol{\beta} \in \boldsymbol{\Theta}$. Here, $\boldsymbol{v}_{0}$ is the square root of $\boldsymbol{v}_{0}^{2}$, which exists because $\boldsymbol{v}_{0}^{2}$ is positive definite by Assumption \ref{ass:Identifiability}.
\begin{lemma}
\label{lem:subexp}
 Suppose Assumptions~\ref{ass:CS}-\ref{ass:Identifiability} are satisfied.
    Consider $\zeta_{i}(\boldsymbol{\beta})$ with $\boldsymbol{\beta} \in \boldsymbol{\Theta}$ and $i \in \{1, \ldots n \}$.  Then, for any $\boldsymbol{\gamma} \in \mathbb{R}^{p}$, \begin{align*}
        \mathbb{E}\exp\left(\lambda \boldsymbol{\gamma}^{T} \nabla \zeta_{i}(\boldsymbol{\beta})\right) \leq \exp\left(\frac{v_{0}^{2}\lambda^{2}}{2} \left\lVert\boldsymbol{v}_{0}\boldsymbol{\gamma}\right\lVert^{2}\right),
    \end{align*}
where $v_{0}^{2} =\max \left( 1, \frac{(Cp\ln p)^2}{\sigma^{2}\lambda_{\min}}\right) $. The constant $Cp\ln(p)$ is the upper bound of  $ \left\lVert\nabla_{\boldsymbol{\beta}} f_{\boldsymbol{\beta}}(x)\right\lVert$  and $\lambda_{\min}$ is the minimum eigenvalue of $\boldsymbol{v}_{0}^{2}$. 
\end{lemma}

\begin{proof}

With $\ell_{i}(\boldsymbol{\beta})$ be the log likelihood for a single data point at $\boldsymbol{\beta}$, let $\zeta_{i}(\boldsymbol{\beta}) = \ell_{i}(\boldsymbol{\boldsymbol{\beta}}) - \mathbb{E}\left( \ell_{i}(\boldsymbol{\beta})\right)$. As the model is second order differentiable with respect to $\boldsymbol{\beta}$, the gradient of $\zeta_{i}(\boldsymbol{\beta})$ is equal to
\begin{align*}
  \nabla_{\boldsymbol{\beta}} \zeta_{i}(\boldsymbol{\beta}) 
    =   \sigma^{-2}(y_{i} - f_{\boldsymbol{\beta}}(x_{i}))\nabla  f_{\boldsymbol{\beta}}(x_{i}) - \mathbb{E}\left[ \sigma^{-2}(f_{\boldsymbol{\beta}^*}(x_{i}) - f_{\boldsymbol{\beta}}(x_{i}))\nabla  f_{\boldsymbol{\beta}}(x_{i})  \right].
\end{align*}
Notice that if the gradient is conditioned on $x_{i}$, the expectation becomes a constant. By rearranging terms and using the fact that the residual $y_{i} - f_{\boldsymbol{\beta}^*}(x_{i}) \sim N(0, \sigma^{2})$, we can conclude that for any $\boldsymbol{\gamma} \in \mathbb{R}^{p}$, \begin{align*}
    \boldsymbol{\gamma}^T \nabla_{\boldsymbol{\beta}} \zeta_i(\boldsymbol{\beta}) \mid x_i \sim \mathcal{N}\left(0, \sigma^{-2} \left( \boldsymbol{\gamma}^T \nabla_{\boldsymbol{\beta}} f_{\boldsymbol{\beta}}(x_i) \right)^2 \right).
\end{align*}
Now, our goal is to upper bound the moment generating function (MGF) of $ \boldsymbol{\gamma}^T \nabla_{\boldsymbol{\beta}} \zeta_i(\boldsymbol{\beta})$. To begin, we apply the law of iterated expectations, conditioning on $x_{i}$. \begin{align*}
    \mathbb{E}\left[ \exp\left( \lambda \boldsymbol{\gamma}^T \nabla_{\boldsymbol{\beta}} \zeta_i(\boldsymbol{\beta}) \right) \right] = \mathbb{E}\left[ \mathbb{E}\left[ \exp\left( \lambda \boldsymbol{\gamma}^T \nabla_{\boldsymbol{\beta}} \zeta_i(\boldsymbol{\beta}) \right) \mid x_i \right] \right].
\end{align*}
Since $ \boldsymbol{\gamma}^T \nabla_{\boldsymbol{\beta}} \zeta_i(\boldsymbol{\beta}) \mid x_i $ follows a normal distribution with mean zero and variance $ \sigma^{-2} \left( \boldsymbol{\gamma}^T \nabla_{\boldsymbol{\beta}} f_{\boldsymbol{\beta}}(x_i) \right)^2 $, the conditional MGF is that of a Gaussian random variable. For a normal random variable $ Z \sim N(0, \tau^2) $, the MGF is $ \mathbb{E}[\exp(\lambda Z)] = \exp(\lambda^2 \tau^2 / 2) $. Thus,
\begin{align*}
    \mathbb{E}\left[ \exp\left( \lambda \boldsymbol{\gamma}^T \nabla_{\boldsymbol{\beta}} \zeta_i(\boldsymbol{\beta}) \right) \mid x_i \right] = \exp\left( \frac{\lambda^2}{2\sigma^2} \left( \boldsymbol{\gamma}^T \nabla_{\boldsymbol{\beta}} f_{\boldsymbol{\beta}}(x_i) \right)^2 \right).
\end{align*}
Taking the expectation over $ x_i $, we obtain 
\begin{align*}
    \mathbb{E}\left[ \exp\left( \lambda \boldsymbol{\gamma}^T \nabla_{\boldsymbol{\beta}} \zeta_i(\boldsymbol{\beta}) \right) \right] = \mathbb{E}\left[ \exp\left( \frac{\lambda^2}{2\sigma^2} \left( \boldsymbol{\gamma}^T \nabla_{\boldsymbol{\beta}} f_{\boldsymbol{\beta}}(x_i) \right)^2 \right) \right].
\end{align*}
Applying the Cauchy-Schwarz inequality to $ \boldsymbol{\gamma}^T \nabla_{\boldsymbol{\beta}} f_{\boldsymbol{\beta}}(x_i)$, we have \begin{align*}
    \left( \boldsymbol{\gamma}^T \nabla_{\boldsymbol{\beta}} f_{\boldsymbol{\beta}}(x_i) \right)^2 \leq \|\boldsymbol{\gamma}\|^2 \|\nabla_{\boldsymbol{\beta}} f_{\boldsymbol{\beta}}(x_i)\|^2.
\end{align*}
In addition, by the results of Lemma~\ref{lem:gradient_lem}, $\|\nabla_{\boldsymbol{\beta}} f_{\boldsymbol{\beta}}(x_i)\|^2 \leq (Cp \ln p)^{2}$. As the exponential is an increasing function, we can now conclude that: 
\begin{align*}
    \mathbb{E}\left[ \exp\left( \frac{\lambda^2}{2\sigma^2} \left( \boldsymbol{\gamma}^T \nabla_{\boldsymbol{\beta}} f_{\boldsymbol{\beta}}(x_i) \right)^2 \right) \right]  \leq \exp\left( \frac{\lambda^2}{2\sigma^2} \|\boldsymbol{\gamma}\|^2 (C p \ln p)^2 \right).
\end{align*}
Finally, $\boldsymbol{v}_{0}^{2}$ is a positive definite matrix, and thus has a positive minimum eigenvalue $\lambda_{\min}$. Thus, with $\boldsymbol{v}_{0}$ being the matrix square root of $\boldsymbol{v}_{0}^{2}$ with minimum eigenvalue $\sqrt{\lambda_{\min}}$, we obtain \begin{align*}
    \left\lVert \boldsymbol{\gamma}  \right\lVert^{2} \leq \frac{1}{\lambda_{\min}}\left\lVert\boldsymbol{v}_{0} \boldsymbol{\gamma}  \right\lVert^{2}.
\end{align*}
Thus, we can conclude the following, which is what we wanted to prove: \begin{align*}
      \mathbb{E}\left[ \exp\left( \frac{\lambda^2}{2\sigma^2} \left( \boldsymbol{\gamma}^T \nabla_{\boldsymbol{\beta}} f_{\boldsymbol{\beta}}(x_i) \right)^2 \right) \right]  \leq \exp\left( \frac{\lambda^2}{2} \frac{(C p \ln p)^2}{\sigma^2 \lambda_{\min}} \|\boldsymbol{v}_{0}\boldsymbol{\gamma}\|^2  \right) = \exp\left( \frac{v_{0}^{2}\lambda^2}{2} \|\boldsymbol{v}_{0}\boldsymbol{\gamma}\|^2  \right).
\end{align*}
\end{proof}

Lemma~\ref{lem:kappa} of regards $\boldsymbol{\beta}$'s outside $\Theta(u_{0})$, which is helpful in deriving a lower bound for the KL-Divergence.
 
\begin{lemma}
    Define $\kappa(u_{0}) = \underset{\boldsymbol{\beta} \notin \Theta(u_{0}) } {\inf} \mathbb{K}(\boldsymbol{\beta}, \boldsymbol{\beta^*})$. If Assumptions~\ref{ass:CS}-\ref{ass:Identifiability} hold, then $\kappa(u_{0})$ is greater than 0.
\end{lemma}

\begin{proof}
    Assume not. Thus, for every $\delta > 0$, then $\exists \boldsymbol{\beta} \notin \Theta(u_{0})$, $\mathbb{K}(\boldsymbol{\beta}, \boldsymbol{\beta^*}) \leq \delta $. With $\delta_{i} = \frac{1}{i}$, create a sequence $\boldsymbol{\beta}_{i}$ such that $\mathbb{K}(\boldsymbol{\beta}_{i}, \boldsymbol{\beta^*}) \leq \delta_{i}$. Since $0< \left\lVert \boldsymbol{\beta}_{i} \right\lVert < \pi$ for all $\boldsymbol{\beta}_{i}$, there exists a subsequence $\boldsymbol{\beta}_{i_{k}} $ that converges to a $\boldsymbol{\beta}$ where $0 \leq\left\lVert \boldsymbol{\beta} \right\lVert \leq \pi $ by the Bolzano-Weierstrass Theorem. Consider the two possible cases: \\ \\
\underline{Case 1}: $\boldsymbol{\beta}_{i_{k}} \rightarrow \boldsymbol{\beta}$, where $0 <\left\lVert \boldsymbol{\beta} \right\lVert < \pi $ \\
Now, since $\mathbb{K}(\boldsymbol{\beta}, \boldsymbol{\beta^*})$ is a continuous function of $\boldsymbol{\beta}$, $\mathbb{K}(\boldsymbol{\beta}_{i_{k}}, \boldsymbol{\beta^*}) \rightarrow 0$ implies $\mathbb{K}(\boldsymbol{\beta}, \boldsymbol{\beta^*}) = 0$. Thus, $\boldsymbol{\beta}_{i_{k}}$ must converge to the unique minimizer $\boldsymbol{\beta}^{*}$, which exists by Lemma~\ref{lem:unique_max}. However, this violates the condition that $\boldsymbol{\beta}_{i_{k}} \notin \Theta(u_{0})$. \\ \\
\underline{Case 2}: $\boldsymbol{\beta}_{i_{k}} \rightarrow \boldsymbol{\beta}$ where $\left\lVert \boldsymbol{\beta} \right\lVert = 0$ or  $\left\lVert \boldsymbol{\beta} \right\lVert = \pi $. \\
$\mathbb{K}(\boldsymbol{\beta}, \boldsymbol{\beta^*})$ remains a continuous function on the set $\{ \boldsymbol{\beta} : 0  \leq\left\lVert \boldsymbol{\beta} \right\lVert \leq \pi \}$. Thus, if $\boldsymbol{\beta}_{i_{k}} \rightarrow \boldsymbol{\beta}$, this implies that $\mathbb{K}(\boldsymbol{\beta}, \boldsymbol{\beta^*}) = 0$. However, if $\left\lVert  \boldsymbol{\beta}  \right\lVert = 0$ or $\left\lVert  \boldsymbol{\beta}  \right\lVert = \pi$, then the $\gamma$ function parametrized by $\boldsymbol{\beta}$ is the identity function. This is because $\sin(0) = \sin(\pi) = 0$ and $\cos(0) = \cos(\pi) = 1$, making the diffeomorphism simply $\gamma(t) = \cos^{2}(\left\lVert  \boldsymbol{\beta}  \right\lVert) t = t$. Thus, this implies that  the $\gamma$ function parametrized by $\boldsymbol{\beta^{*}}$ is also the identity function, which violates the assumption $g \neq f_{\boldsymbol{\beta}^*}$ (Assumption \ref{ass:CS}).  \\ \\
Because a contradiction has arisen from both cases, the lemma must be true. 
\end{proof}

\subsection{Additional Lemmas for Finite Sample Theory}

\noindent The next lemma establishes that identifiability of the model despite the non-concavity of the log-likelihood. 

\newtheorem{slemma}{Lemma}       
\renewcommand{\theslemma}{S.\arabic{slemma}}

\begin{slemma}
\label{lem:unique_max}
   If Assumptions~\ref{ass:CS}-\ref{ass:Identifiability} hold, then $\boldsymbol{\beta^*}$ is unique minimizer of $\mathbb{K}(\boldsymbol{\beta}, \boldsymbol{\beta^*})$.
\end{slemma} 

\begin{proof}
     Clearly, $\mathbb{K}(\boldsymbol{\beta^*}, \boldsymbol{\beta^*}) = 0$,  Now, consider for $\boldsymbol{\beta} \neq \boldsymbol{\beta^*}$, $d_{\boldsymbol{\beta}}(x) = \gamma_{\boldsymbol{\beta^*}}(x) - \gamma_{\boldsymbol{\beta}}(x) $. Now,
 for all $\boldsymbol{\beta} \in \boldsymbol{\Theta}$, $d_{\boldsymbol{\beta}}(0) = d_{\boldsymbol{\beta}}(1) = 0$ and  $d_{\boldsymbol{\beta}}(x)$ is always a continuous, differentiable function. As each $\boldsymbol{\beta}$ defines a unique diffeomorphism, by Rolle's Theorem, then there exists an interval $[a, b]$ such that $0 \leq a < b \leq 1$,  $d_{\boldsymbol{\beta}}(a) = 0$, $d_{\boldsymbol{\beta}}'(b) = 0,$ and $d_{\boldsymbol{\beta}}(x)$ is always positive or negative for all $x \in (a, b]$.  Thus, $\gamma_{\boldsymbol{\beta^*}}(a) = \gamma_{\boldsymbol{\beta}}(a)$ and either $\gamma_{\boldsymbol{\beta^*}}(x) > \gamma_{\boldsymbol{\beta}}(x) $ or $\gamma_{\boldsymbol{\beta^*}}(x) < \gamma_{\boldsymbol{\beta}}(x) $ for $x \in (a, b]$. 
 
 Due to $\boldsymbol{\gamma}_{\boldsymbol{\beta}}$ being a continuous bijection of $[0, 1]$ with a derivative greater than 0 for all $\boldsymbol{\beta} \in \boldsymbol{\Theta}$, $[\gamma_{\boldsymbol{\beta^*}}(a), \gamma_{\boldsymbol{\beta^*}}(b)]$ and $[\gamma_{\boldsymbol{\beta}}(a), \gamma_{\boldsymbol{\beta}}(b)]$ are intervals of length greater than 0. In addition, for $f_{\boldsymbol{\beta^*}} = g\circ \gamma_{\boldsymbol{\beta^*}}$, $g \in \mathcal{F}_{M}$, and thus, must be strictly increasing or decreasing in all of $[0, 1]$ except for a finite amount of points (Definition \ref{def:F_M}). Therefore, there exists points $c, d$ such that $a < c < d \leq b$ such that $g(\gamma_{\boldsymbol{\beta^*}}(x) ) > g(\gamma_{\boldsymbol{\beta}}(x))$ or $g(\gamma_{\boldsymbol{\beta^*}}(x) ) < g(\gamma_{\boldsymbol{\beta}}(x))$ for all $x \in [c, d]$. By the positivity of $\mathbb{P}_{X}$, $\mathbb{P}_{X}([c, d]) > 0$. Thus, for all $\boldsymbol{\beta} \neq \boldsymbol{\beta^*}$, $\mathbb{E}[(f_{\boldsymbol{\beta}}(x) - f_{\boldsymbol{\beta^*}}(x) )^{2}] > 0$, which implies $\mathbb{K}(\boldsymbol{\beta}, \boldsymbol{\beta^*}) > 0$ for all  $\boldsymbol{\beta} \neq \boldsymbol{\beta^*}$. This proves the lemma. 
\end{proof}
\noindent This final lemma establishes upper bounds on functions of $\boldsymbol{\beta}$ within $\Theta(u_{0})$.
\begin{slemma}
\label{lem:theta_u_upper}
    Suppose Assumptions~\ref{ass:CS}-\ref{ass:Identifiability} are satisfied. Consider sets $\Theta(u)$ where $u \leq u_{0}$. Let $f(\boldsymbol{\beta})$ be a positive, continuous and differentiable function where the gradient is bounded on $\Theta(u_{0})$. In addition, let $f(\boldsymbol{\beta^*}) = 0$. Then, the function $g(u) = \sup\limits_{\boldsymbol{\beta} \in \Theta(u)} f(\boldsymbol{\beta})$ can be bounded from above by $Cu$, where $C$ is some constant greater than 0. 
\end{slemma}

\begin{proof}
Because the gradient is bounded on $\Theta(u_{0})$, $f(\boldsymbol{\beta})$ is Lipschitz on $\Theta(u_{0})$ with Lipschitz constant $L$. Now, if $\boldsymbol{\beta} \in \Theta(u)$, and because $ f(\boldsymbol{\beta}) \geq f(\boldsymbol{\beta^*}) $ \begin{align*}
    f(\boldsymbol{\beta}) - f(\boldsymbol{\beta^*}) \leq L\left\lVert \boldsymbol{\beta} - \boldsymbol{\beta}^* \right\lVert\leq \frac{L}{\sqrt{\lambda_{\min}}}\left\lVert \boldsymbol{v}_{0}(\boldsymbol{\beta} - \boldsymbol{\beta}^*) \right\lVert = \frac{L}{\sqrt{\lambda_{\min}}}u.
\end{align*}
As this upper bound applies for all $\boldsymbol{\beta} \in \Theta(u)$, \begin{align*}
    g(u) = \sup\limits_{\boldsymbol{\beta} \in \Theta(u)} f(\boldsymbol{\beta}) = \sup\limits_{\boldsymbol{\beta} \in \Theta(u)} \left( f(\boldsymbol{\beta}) - f(\boldsymbol{\beta}^*) \right) + f(\boldsymbol{\beta}^*) \leq \frac{L}{\sqrt{\lambda_{\min}}}u + 0.
\end{align*}
This completes the proof. 
\end{proof}

\subsection{Proving the Spokoiny i.i.d Conditions}

\newtheorem{stheorem}{Theorem}
\renewcommand{\thestheorem}{S.\arabic{stheorem}}

The conditions are listed in Section 5.1 of \cite{Spokoiny2012_ParametricEstimation}. We prove all of them in the following theorem:
\begin{stheorem}
    The diffeomorphism regression model satisfies the Spokoiny i.i.d. conditions.
\end{stheorem}
\begin{proof} 

    $\boldsymbol{(ed_{0})}$: Let $v_{0}^{2} =\max \left( 1, \frac{(Cp\ln p)^2}{\sigma^{2}\lambda_{\min}}\right)$. By Lemma~\ref{lem:subexp}, if we let $g_{1} > 0$, then for all $|\lambda| \leq g_{1}$\begin{align*}
       \mathbb{E}\exp\left(\lambda \frac{\boldsymbol{\gamma}^{T} \nabla \zeta_{i}(\boldsymbol{\beta})}{\left\lVert  \boldsymbol{v}_{0}\boldsymbol{\gamma}  \right\lVert}\right) \leq \exp\left(\frac{\lambda^{2}}{2} \frac{(Cp\ln p)^2}{\sigma^{2}\lambda_{\min}} \right) = \exp\left(\frac{v_{0}^{2}\lambda^{2}}{2}\right), \ \ \forall \boldsymbol{\beta} \in \Theta.
    \end{align*}
$\boldsymbol{(ed_{0})}$ follows because $\boldsymbol{\beta^*} \in \Theta$. \\ 

\noindent $\boldsymbol{(ed_{1})}$: Similarly to Lemma \ref{lem:subexp}, \begin{align*}
    \nabla \zeta_{i}(\boldsymbol{\beta}) - \nabla \zeta_{i}(\boldsymbol{\beta^*}) \  | \ x_{i} =  \sigma^{-2}\left(y_{i} - f_{\boldsymbol{\beta^*}}(x_{i})\right)\left(\nabla_{\boldsymbol{\beta}} f_{\boldsymbol{\beta}} (x_{i})- \nabla_{\boldsymbol{\beta}} f_{\boldsymbol{\beta}^*}(x_{i})\right).
\end{align*} 
Thus, for any $\boldsymbol{\gamma} \in \mathbb{R}^{p}$ \begin{align*}
    \boldsymbol{\gamma}^{T}\nabla \zeta_{i}(\boldsymbol{\beta}) \ | \ x_{i} \sim N\left(0, \sigma^{-2}\left(\boldsymbol{\gamma}^T\left[\nabla_{\boldsymbol{\beta}} f(\boldsymbol{\beta}) - \nabla_{\boldsymbol{\beta}} f(\boldsymbol{\beta^*})\right]\right)^2\right).
\end{align*}
Now, consider the function $h(\boldsymbol{\beta}) = \left\lVert \nabla_{\boldsymbol{\beta}} f(\boldsymbol{\beta}) - \nabla_{\boldsymbol{\beta^*}} f(\boldsymbol{\beta^*})\right\lVert$. By Lemma \ref{lem:gradient_lem}, $f(\boldsymbol{\beta})$ has continuous second partial derivatives on $\Theta(u_{0})$. Thus, by the extreme value theorem, $h(\boldsymbol{\beta})$ is a continuous differentiable function where the gradient is bounded. In addition, clearly $h(\boldsymbol{\beta^*}) = 0$. Thus, by Lemma \ref{lem:theta_u_upper}, 
\begin{align}
\label{eq:C1}
    \sup_{\boldsymbol{\beta} \in \Theta(u)} \left\lVert \nabla_{\boldsymbol{\beta}} f(\boldsymbol{\beta}) - \nabla_{\boldsymbol{\beta^*}} f(\boldsymbol{\beta^*})\right\lVert \leq C_{1}u, \quad\text{ for some } C_{1} > 0.
\end{align} 
Now, following similar logic to the proof of Lemma \ref{lem:subexp}, let $\boldsymbol{\beta} \in \Theta(u)$. Then,  
\begin{align*}
     \mathbb{E}\exp\left(\lambda \boldsymbol{\gamma}^T(\nabla \zeta_{i}(\boldsymbol{\beta}) - \nabla \zeta_{i}(\boldsymbol{\beta^*})\right) = & \mathbb{E}\left( \mathbb{E} \left[ \exp\left(\lambda \boldsymbol{\gamma}^T(\nabla \zeta_{i}(\boldsymbol{\beta}) - \nabla \zeta_{i}(\boldsymbol{\beta^*})\right) \ | \ x_{i} \right] \right),
\end{align*}
by the law of iterated expectations. Because the inner expectation is the moment generating function of a normal random variable, 
\begin{align*}
    \mathbb{E}\left( \mathbb{E} \left[ \exp\left(\lambda \boldsymbol{\gamma}^T(\nabla \zeta_{i}(\boldsymbol{\beta}) - \nabla \zeta_{i}(\boldsymbol{\beta^*})\right) \ | \ x_{i} \right] \right) = \mathbb{E} \exp\left( \frac{\lambda^2}{2\sigma^2} \left(\boldsymbol{\gamma}^T\left[\nabla_{\boldsymbol{\beta}} f(\boldsymbol{\beta}) - \nabla_{\boldsymbol{\beta^*}} f(\boldsymbol{\beta^*})\right]\right)^2 \right).
\end{align*}
By the Cauchy-Schwarz Inequality, \begin{align*}
    \mathbb{E} \exp\left( \frac{\lambda^2}{2\sigma^2} \left(\boldsymbol{\gamma}^T\left[\nabla_{\boldsymbol{\beta}} f(\boldsymbol{\beta}) - \nabla_{\boldsymbol{\beta^*}} f(\boldsymbol{\beta^*})\right]\right)^2 \right) \leq \mathbb{E}  \exp\left( \frac{\lambda^2}{2\sigma^2} \left(\left\lVert \boldsymbol{\gamma}  \right\lVert^{2} \left\lVert\nabla_{\boldsymbol{\beta}} f(\boldsymbol{\beta}) - \nabla_{\boldsymbol{\beta^*}} f(\boldsymbol{\beta^*})\right\lVert^{2}\right) \right).
\end{align*}
Next, by result (\ref{eq:C1}), 
\begin{align*}
   \mathbb{E}  \exp\left( \frac{\lambda^2}{2\sigma^2} \left(\left\lVert \boldsymbol{\gamma}  \right\lVert^{2} \left\lVert\nabla_{\boldsymbol{\beta}} f(\boldsymbol{\beta}) - \nabla_{\boldsymbol{\beta^*}} f(\boldsymbol{\beta^*}) \right\lVert^{2}\right) \right)\leq  \exp\left( \frac{\lambda^2}{2\sigma^2} \left( \left\lVert \boldsymbol{\gamma}  \right\lVert^{2} C_{1}^{2}u^2 \right)\right).
\end{align*}
Finally, using the eigenvalue bound from the proof of Lemma \ref{lem:subexp}, we can conclude \begin{align*}
    \exp\left( \frac{\lambda^2}{2\sigma^2} \left( \left\lVert \boldsymbol{\gamma}  \right\lVert^{2} C_{1}^{2}u^2 \right)\right) \leq  &\exp\left( \frac{\lambda^2}{2 }\left( \frac{(Cp\ln p)^{2}}{\sigma^2 \lambda_{\min}} \right) \left( \frac{C_{1}^{2}u^{2}}{(Cp\ln p)^{2}} \right) \left\lVert\boldsymbol{v}_{0} \boldsymbol{\gamma}  \right\lVert^2  \right) \\
    =& \exp\left( \frac{v_{0}^{2} \lambda^{2}}{2} \left( \frac{C_{1}^{2}u^{2}}{(Cp\ln p)^{2}} \right) \left\lVert\boldsymbol{v}_{0} \boldsymbol{\gamma}  \right\lVert ^2\right).
\end{align*}
\noindent With $\omega^* = \frac{C_{1}}{Cp \ln p} > 0$ and $g_{1} > 0$ such that $|\lambda| \leq g_{1}$, \begin{align*}
    \mathbb{E}\exp\left(\lambda \frac{\boldsymbol{\gamma}^T(\nabla \zeta_{i}(\boldsymbol{\beta}) - \nabla \zeta_{i}(\boldsymbol{\beta^*}))}{\omega^* u \left\lVert  \boldsymbol{v}_{0}\boldsymbol{\gamma}  \right\lVert }\right) \leq \exp\left( \frac{v_{0}^{2} \lambda^{2}}{2}\right),
\end{align*}
which shows $\boldsymbol{(ed_{1})}$.  \\

\noindent $\boldsymbol{(i)}$: \cite{Spokoiny2012_ParametricEstimation} defines a matrix $\boldsymbol{f}_{0}$ as the Hessian of the the KL-Divergence at $\boldsymbol{\beta^*}$, $H(\boldsymbol{\beta^*})$. As this equals $\boldsymbol{v}_{0}^{2}$, this condition is satisfied. \\ 

\noindent $\boldsymbol{(\ell_{0})}$: Consider $\boldsymbol{\beta} \in \Theta(u_{0})$ for $u \leq u_{0}$. Performing a Taylor expansion on $2\mathbb{K}(\boldsymbol{\beta}, \boldsymbol{\beta^*})$ centered around $\boldsymbol{\beta^*}$, \begin{align*}
2\mathbb{K}(\boldsymbol{\beta}, \boldsymbol{\beta^*}) = (\boldsymbol{\beta} - \boldsymbol{\beta^*})^{T}\boldsymbol{v}_{0}^{2}(\boldsymbol{\beta} - \boldsymbol{\beta^*}) + R(\boldsymbol{\beta}, \boldsymbol{\beta^*}).
\end{align*}
By Lemma~\ref{lem:gradient_lem}, $\mathbb{K}(\boldsymbol{\beta}, \boldsymbol{\beta^*})$ is a thrice continuously differentiable function for $\boldsymbol{\beta} \in \Theta(u_{0})$. Thus, because $|R(\boldsymbol{\beta}, \boldsymbol{\beta}^*)|$ is $o(\left\lVert\boldsymbol{\beta} - \boldsymbol{\beta}^*\right\lVert^2)$,  \begin{align*}
    \frac{|R(\boldsymbol{\beta}, \boldsymbol{\beta}^*)|}{\left\lVert \boldsymbol{v}_{0}(\boldsymbol{\beta}  - \boldsymbol{\beta}^*)\right\lVert^{2}} \leq \frac{|R(\boldsymbol{\beta}, \boldsymbol{\beta}^*)|}{\lambda_{\min}\left\lVert(\boldsymbol{\beta} - \boldsymbol{ \beta}^*)\right\lVert^{2}} \rightarrow 0 \quad \text{as } \boldsymbol{\beta} \rightarrow \boldsymbol{\beta}^*.
\end{align*}
In addition, by properties of Taylor series remainders, $|R(\boldsymbol{\beta}, \boldsymbol{\beta}^*)|$ is $O(\left\lVert\boldsymbol{\beta} - \boldsymbol{\beta}^*\right\lVert^3)$ and for any closed set that contains $\boldsymbol{\beta}^*$, \begin{align*}
    |R(\boldsymbol{\beta}, \boldsymbol{\beta}^*)| \leq \delta\left\lVert \boldsymbol{\beta} - \boldsymbol{\beta}^*\right\lVert^3 \leq \delta^* \left\lVert\boldsymbol{v}_{0}(\boldsymbol{\beta} - \boldsymbol{\beta}^*)\right\lVert^3,\quad \text{for some } \delta^* > 0.
\end{align*}
Thus, $\frac{ |R(\boldsymbol{\beta}, \boldsymbol{\beta}^*)|}{\left\lVert\boldsymbol{v}_{0}(\boldsymbol{\beta} - \boldsymbol{\beta}^*)\right\lVert^3} \leq \delta^*,$ for $\boldsymbol{\beta} \neq \boldsymbol{\beta}^*$ and in this closed set.  As $\Theta(u_{0})$ is a closed set that contains $\boldsymbol{\beta}^*$ and $\frac{|R(\boldsymbol{\beta}, \boldsymbol{\beta}^*)|}{\left\lVert\boldsymbol{v}_{0}(\boldsymbol{\beta}  -  \boldsymbol{\beta}^*)\right\lVert^{2}} \rightarrow 0 \text{ as } \boldsymbol{\beta} \rightarrow \boldsymbol{\beta}^*$, we can conclude
\begin{align*}
    \left| \frac{2\mathbb{K}(\boldsymbol{\boldsymbol{\beta}, \boldsymbol{\beta}}^*)}{\left\lVert \boldsymbol{v}_{0}(\boldsymbol{\beta} - \boldsymbol{\beta}^*) \right\lVert^{2}} - 1 \right| = \left| \frac{R(\boldsymbol{\boldsymbol{\beta}, \boldsymbol{\beta}}^*)}{\left\lVert \boldsymbol{v}_{0}(\boldsymbol{\beta} - \boldsymbol{\beta}^*) \right\lVert^{2}}  \right| = \frac{|R(\boldsymbol{\beta}, \boldsymbol{\beta}^*)|}{\left\lVert \boldsymbol{v}_{0}(\boldsymbol{\beta} - \boldsymbol{\beta}^*) \right\lVert^{3}} \left\lVert \boldsymbol{v}_{0}(\boldsymbol{\beta} - \boldsymbol{\beta}^*) \right\lVert \leq \delta^* u.
\end{align*}
This shows $\boldsymbol{(\ell_{0})}$. \\ 

\noindent $\boldsymbol{(eu)}$: Let $g_{1}(u) > 0 $ for all $u > 0$ and $|\lambda| \leq g_{1}(u)$. This condition will follow directly from Lemma~\ref{lem:subexp}, as Lemma~\ref{lem:subexp} applies to all $\boldsymbol{\beta} \in \boldsymbol{\Theta}$. \\

\noindent $\boldsymbol{(\ell_{u})}$: First, let us consider $\boldsymbol{\beta} \notin \Theta(u_{0})$.
By Lemma~\ref{lem:kappa}, there exists a positive real number $\kappa(u_{0}) > 0$ such that if $\boldsymbol{\beta} \notin \Theta(u_{0}) $ then $\mathbb{K}(\boldsymbol{\beta}, \boldsymbol{\beta^*})  \geq \kappa(u_{0})$. Now, with $\lambda_{\max}$ being the dominant eigenvalue of $\boldsymbol{v}^{2}_{0}$, \begin{align*}
u^{2} = \left\lVert \boldsymbol{v}_{0}(\boldsymbol{\beta} - \boldsymbol{\beta}^*) \right\lVert^{2} \leq \lambda_{\max}\left\lVert \boldsymbol{\beta} - \boldsymbol{\beta}^* \right\lVert^{2} \leq 4\lambda_{\max}\pi^{2}.
\end{align*}
 Thus, for $\boldsymbol{\beta} \notin \Theta(u_{0})$,  \begin{align*}
    \frac{\kappa(u_{0})}{4\lambda_{\max}\pi^{2}}u^{2} \leq \kappa(u_{0}) \leq \mathbb{K}(\boldsymbol{\beta}, \boldsymbol{\beta^*}).
\end{align*}
Now, let us consider $\boldsymbol{\beta} \in \Theta(u_{0})$. Performing a first order Taylor expansion of $\mathbb{K}(\boldsymbol{\beta}, \boldsymbol{\beta^*})$, we get \begin{align*}
    \mathbb{K}(\boldsymbol{\beta}, \boldsymbol{\beta^*}) = \int_{0}^{1}(1 - t)h^TH(\boldsymbol{\beta}^* + th)h \ dt,
\end{align*}
where $h = \boldsymbol{\beta} - \boldsymbol{\beta^*}$ and $H(\boldsymbol{\beta}^* + th)$ is the Hessian at point $\boldsymbol{\beta}^* + th$. For $\boldsymbol{\beta} \in \Theta(u_{0})$, $H(\boldsymbol{\beta})$ is positive definite, as $\mathbb{K}(\boldsymbol{\beta}, \boldsymbol{\beta}^*)$ is strictly convex on all points in $\Theta(u_{0})$. Thus, all the minimum eigenvalues of Hessians in  $\Theta(u_{0})$ are positive. Now, define \begin{align*}
    \lambda_{\min}(u_{0}) = \inf_{\boldsymbol{\beta} \in \Theta(u_{0})} \lambda_{\min}(H(\boldsymbol{\beta})),
\end{align*}
which is clearly greater than or equal to 0. Now, by Lemma~\ref{lem:gradient_lem}, all the second derivatives within the Hessian are continuous, making the mapping $\lambda_{\min}(H(\boldsymbol{\beta}))$ continuous. In addition, $\Theta(u_{0})$ is a closed set. Thus, the argmin of $\lambda_{\min}(H(\boldsymbol{\beta}))$ must be located in $\Theta(u_{0})$, making $\lambda_{\min}(u_{0}) > 0$. Thus, for $\boldsymbol{\beta} \in \Theta(u_{0})$, \begin{align*}
    \mathbb{K}(\boldsymbol{\beta}, \boldsymbol{\beta^*}) =  \int_{0}^{1}(1 - t)h^TH(\boldsymbol{\beta}^* + th)h \ dt 
    \geq &  \lambda_{\min}(u_{0}) \int_{0}^{1}(1 - t)h^Th \ dt \\
    \geq  & \frac{\lambda_{\min}(u_{0})}{2} \left\lVert \boldsymbol{\beta} - \boldsymbol{\beta}^* \right\lVert^2 \\
    \geq  & \frac{\lambda_{\min}(u_{0})}{2 \lambda_{\max}}\left\lVert \boldsymbol{v}_{0}(\boldsymbol{\beta} - \boldsymbol{\beta}^*) \right\lVert^{2}.
\end{align*}
Therefore, with $b = \min (\frac{\lambda_{\min}(u_{0})}{2 \lambda_{\max}}, \frac{\kappa(u_{0})}{4\lambda_{\max}\pi^{2}}) $, we have shown \begin{align*}
    \mathbb{K}(\boldsymbol{\beta}, \boldsymbol{\beta}^*) \geq b\left\lVert \boldsymbol{v}_{0}(\boldsymbol{\beta} - \boldsymbol{\beta}^*) \right\lVert^2,
\end{align*}
which proves $\boldsymbol{(\ell_{u})}$.  All conditions proven. 
\end{proof}

With all the conditions proven, the large deviation bound of the maximum likelihood estimator $\boldsymbol{\hat{\beta}}$ immediately follows. 

\subsection{Proofs of Theorems~\ref{thm:Consistency}, \ref{thm:Normality} and~\ref{thm:Growing_Parameter} }
We now provide proofs of each of the main theorems from the main text. 
\setcounter{theorem}{1}
 \begin{theorem}
  Suppose Assumptions~\ref{ass:CS}-\ref{ass:Identifiability} hold.  Let $u \leq u_{0}$, $p \geq 2$ be the parameter dimension, $v_{0}$ be a fixed value dependent on $p$ and $0 < \alpha < 1$. Then if 
      $  n^{1/2}ub \geq 6v_{0}\sqrt{-\ln(\alpha) + 2p}$, \begin{equation}
          \mathbb{P}\left(\boldsymbol{\hat{\beta}} \notin \Theta(u)\right) \leq \alpha.
      \end{equation}
\end{theorem}

\begin{proof}
    From Lemma 5.1 of \cite{Spokoiny2012_ParametricEstimation}, we define
    $g = g_{1}(u)\sqrt{n}$, where $g_{1}(u)$ is defined in condition $\boldsymbol{(eu)}$. Now, $g_{1}(u)$ can be any positive number for all $u$, and let $g_{1}(u)$ be large enough such that  $1 + \sqrt{-\ln(\alpha) + 2p} \leq \frac{3v_{0}^{2}g(u)}{b}n^{1/2}$. Then the result immediately follows by Theorem 5.2 of \cite{Spokoiny2012_ParametricEstimation}.
\end{proof}
Next, define the following quantities, where $\omega^*$ and $\delta^*$ are the constants derived from the proof of condition $\boldsymbol{(ed_{0})}$:  \begin{align}
    \rho(u) = Ru = 3v_{0}\omega^* u \quad \text{and} \quad \delta(u) = Du=  \delta^* u.
\end{align}  We choose $u \leq u_{0}$ such that $\rho(u) + \delta(u) < 1$, which exists because $R$ and $D$ are positive constants. We also define the following quantities, which are scaled versions of $\boldsymbol{v}_{0}^{2}$ and the normalized score:  \begin{align}
    \boldsymbol{f}_{\boldsymbol{\epsilon}} & =  (1 - \rho(u) - \delta(u))\boldsymbol{v}_{0}^{2} & \qquad &\boldsymbol{f}_{\boldsymbol{\underline{\epsilon}}} =  (1 + \rho(u) + \delta(u))\boldsymbol{v}_{0}^{2}, \\  \boldsymbol{\xi}_{\boldsymbol{\epsilon}}, & = (n\boldsymbol{f}_{\boldsymbol{\epsilon}})^{-1/2}\sum\limits_{i = 1}^{n} \nabla_{\boldsymbol{\beta}} \ell_{i}(\boldsymbol{\beta}^*),   & \qquad & \boldsymbol{\xi}_{\boldsymbol{\underline{\epsilon}}} = (n\boldsymbol{f}_{\boldsymbol{\underline{\epsilon}}})^{-1/2}\sum\limits_{i = 1}^{n} \nabla_{\boldsymbol{\beta}}  \ell_{i}(\boldsymbol{\beta}^*),
\end{align}
where $\boldsymbol{f}_{\boldsymbol{\epsilon}}$ and $\boldsymbol{f}_{\boldsymbol{\underline{\epsilon}}}$ are  rescalings of $\boldsymbol{v}_{0}^{2}$. In addition, define the \textit{spread} $\Delta_{\boldsymbol{\epsilon}}(u)$ as \begin{align}
    \Delta_{\boldsymbol{\epsilon}}(u) = \Diamond_{\boldsymbol{\epsilon}}(u) + \Diamond_{\underline{\boldsymbol{\epsilon}}}(u) + (\left\lVert \boldsymbol{\xi}_{\boldsymbol{\epsilon}} \right\lVert^{2} - \left\lVert \boldsymbol{\xi}_{\underline{\boldsymbol{\epsilon}}}\right\lVert^{2})/2.
\end{align}
It is worth noting that \cite{Spokoiny2012_ParametricEstimation} defines $\Delta$, $\Diamond_{\boldsymbol{\epsilon}}$ and $\Diamond_{\underline{\boldsymbol{\epsilon}}}$ in terms of $r$, where $r = \sqrt{n}u$. However, these random values apply for $\boldsymbol{\beta}$ such that $\left\lVert \sqrt{n}\boldsymbol{v}_{0}\left( \boldsymbol{\beta} - \boldsymbol{\beta^*}\right) \right\lVert \leq r$. Clearly, this condition applies if and only if $\boldsymbol{\beta} \in \Theta(u)$, and thus we define these terms with respect to $u$. 
By Proposition 3.7 in \cite{Spokoiny2012_ParametricEstimation}, if $0 < \alpha < 1$,  $p \geq 2$ and letting $g = g_{1}(u)\sqrt{n}$ be large enough such that $g_{0} = g v_{0} \geq (1 + \sqrt{-\ln(\alpha) + 2p})^{2} \geq 3$, then \begin{center}
    $P\left(\Diamond_{\boldsymbol{\epsilon}}(u) \geq \rho(u)\left(1 +\sqrt{-\ln(\alpha) + 2p}\right)^{2}\right) \leq \alpha$.
\end{center}
The same applies for $\Diamond_{\underline{\boldsymbol{\epsilon}}}(u)$.  Now, to make the spread arbitrary small, we need to make $u$ as small as possible. However, for the large deviation probability statement to apply, a sufficiently large sample size is required. Finally, define \begin{align}
    C_{\boldsymbol{\epsilon}}(\sqrt{n}u) = \big\{ \boldsymbol{\hat{\beta}} \in \Theta(u) , \ \left\lVert \boldsymbol{\xi}_{\boldsymbol{\underline{\epsilon}}} \right\lVert  \leq \sqrt{n}u \big\}.
\end{align} 
This leads to the following theorem, which establishes an error bound of the asymptotical normal approximation of the MLE. 
\begin{theorem} 
\label{thm:Normality}
Suppose Assumptions~\ref{ass:CS}-\ref{ass:Identifiability} are satisfied. If $u \leq u_{0}$ such that $\rho(u) + \delta(u) < 1$, then it holds on a random set $C_{\boldsymbol{\epsilon}}(\sqrt{n}u)$ and a random value $\Delta_{\boldsymbol{\epsilon}}(u)$ that \begin{equation}
\Bigl\|\sqrt{n\boldsymbol{f_{\epsilon}}}\left(\boldsymbol{\hat{\beta}} - \boldsymbol{\beta^*}\right) - \boldsymbol{\xi_{\epsilon}}\Bigr\|^{2} \leq 2\Delta_{\boldsymbol{\epsilon}}(u),
\end{equation}
where the probability of $C_{\boldsymbol{\epsilon}}(\sqrt{n}u)$ and the distribution of $\Delta_{\boldsymbol{\epsilon}}(u)$ are dependent on $u$.
\end{theorem}
\begin{proof}
    This is a direct result of Theorem 5.3 of \cite{Spokoiny2012_ParametricEstimation}, applied for any $u \leq u_{0}$ such that $\rho(u) + \delta(u) < 1$.
\end{proof}

All the above results also apply in the case the parameter dimension $p$ grows as a function of $n$. As long as $p(n)$ grows at a rate that such that all the above conditions are satisfied, then the large deviation and asymptotic normality results remain.
\begin{theorem}
\label{thm:Growing_Parameter}
   Assume Assumptions~\ref{ass:CS}-\ref{ass:Identifiability} and Assumptions~\ref{ass:u_p}-\ref{ass:p_n} hold. Let $p \geq 2$  and $\boldsymbol{\beta^{*}_{\textit{p}}}$ converge to an infinite dimensional limit $\boldsymbol{\beta}$, where $\left\lVert \boldsymbol{\beta} \right\lVert < \pi $. Let $\boldsymbol{\hat{\beta}}_{p}$ be the MLE in dimension $p$ and let $u_{p}$, $\alpha_{p}$ and a function $p=p(n)$  satisfy all the conditions from Assumptions~\ref{ass:u_p} and~\ref{ass:p_n}. Then  
\begin{equation}             \mathbb{P}\left(\boldsymbol{\hat{\beta}}_{p} \notin \Theta_p(u_{p})\right) \leq \alpha_{p} \rightarrow 0,
    \end{equation}
and on the random set $C_{\boldsymbol{\epsilon}, p}(\sqrt{n}u_{p})$ \begin{equation}
\Bigl\|\sqrt{n(\boldsymbol{f}_{\boldsymbol{\epsilon}})_{p}}\left(\boldsymbol{\hat{\beta}}_{p} -\boldsymbol{\beta^{*}_{\textit{p}}}\right) - \left(\boldsymbol{\xi}_{\boldsymbol{\epsilon}}\right)_{p}\Bigr\|^{2}  \leq 2\Delta_{\boldsymbol{\epsilon}, p}\left(u_{p}\right), 
\end{equation}
where $\mathbb{P}\left(C_{\boldsymbol{\epsilon}, p}(\sqrt{n}u_{p})\right) \rightarrow 1$ and  $\Delta_{\boldsymbol{\epsilon}, p}\left(u_{p}\right)  \rightarrow 0$ in probability as $p \rightarrow \infty$. 
\end{theorem}
\begin{proof}
    The first result follows directly from Assumption~\ref{ass:p_n}.i,
    $$n^{1/2}(u_{p})(b_{p})  \geq  6(v_{0})_{p}\sqrt{-\ln(\alpha_{p}) + 2p(n)},$$ and from Theorem~\ref{thm:Consistency}. As we defined $\alpha_{p} \rightarrow 0$ (Assumption~\ref{ass:p_n}), $\mathbb{P}\left(\boldsymbol{\hat{\beta}}_{p} \in \Theta_p(u_{p})\right) \rightarrow 1.$ Next, the bound $$\Bigl\|\sqrt{n(\boldsymbol{f}_{\boldsymbol{\epsilon}})_{p}}\left(\boldsymbol{\hat{\beta}}_{p} - \boldsymbol{\beta}^*_{p}\right) - \left(\boldsymbol{\xi}_{\boldsymbol{\epsilon}}\right)_{p}\Bigr\|^{2}  \leq 2\Delta_{\boldsymbol{\epsilon}, p}\left(u_{p}\right) $$ is a direct application of Theorem~\ref{thm:Normality}.  Now consider 
    \begin{align*}
        \boldsymbol{\xi}_{\boldsymbol{\underline{\epsilon}}} = (n\boldsymbol{f}_{\boldsymbol{\underline{\epsilon}}})^{-1/2}\sum\limits_{i = 1}^{n} \nabla_{\boldsymbol{\beta}}  \ell_{i}(\boldsymbol{\beta}^*).
    \end{align*}
    As $\boldsymbol{f}_{\boldsymbol{\underline{\epsilon}}} = \left(1 + \rho(u) +\delta(u)\right)\boldsymbol{v}_{0}^{2}$ is positive definite, \begin{align*}
\left\lVert\boldsymbol{\xi}_{\underline{\boldsymbol{\epsilon}}}\right\lVert \leq &n^{-1/2}\lambda_{\max}\left((\boldsymbol{f_{\epsilon}})^{-1/2} \right) \Bigl\| \sum\limits_{i = 1}^{n} \nabla_{\boldsymbol{\beta}} \ell_{i}(\boldsymbol{\beta^*}) \Bigr\|.
    \end{align*}
Using the fact that $$\lambda_{\max}\left((\boldsymbol{f_{\epsilon}})^{-1/2} \right) =\left(1 + \rho(u) +\delta(u)\right)^{-1/2 } \left(\lambda_{\min}\left(\boldsymbol{v}_{0}\right)\right)^{-1}$$ and substituting the explicit form of $ \nabla_{\boldsymbol{\beta}} \ell_{i}(\boldsymbol{\beta^*})$, we obtain \begin{align*}
    \left\lVert\boldsymbol{\xi}_{\underline{\boldsymbol{\epsilon}}}\right\lVert \leq \left(n\left(1 +\rho(u) + \delta(u)\right) \right)^{-1/2} \left(\lambda_{\min}\left(\boldsymbol{v}_{0}\right)\right)^{-1}\Bigl\| \sum\limits_{i = 1}^{n} \sigma^{-2}\left(y_{i} - f_{\boldsymbol{\beta}^*}(x_{i}) \right) \nabla_{\boldsymbol{\beta}}f_{\boldsymbol{\beta}^*}(x_{i})\Bigr\|.
\end{align*}
Because $\left(\lambda_{\min}\left(\boldsymbol{v}_{0}\right)\right)^{-1} = \left(\lambda_{\min}\left(\boldsymbol{v}^{2}_{0}\right)\right)^{-1/2}$  and  $\left\lVert \nabla_{\boldsymbol{\beta}}f_{\boldsymbol{\beta}^*}(x_{i})\right\lVert \leq Cp\ln p$,
\begin{align*}
\left\lVert\boldsymbol{\xi}_{\underline{\boldsymbol{\epsilon}}}\right\lVert \leq \left(n\left(1 + \rho(u) +\delta(u)\right) \right)^{-1/2} \left( \frac{Cp\ln p}{\sigma^2 \sqrt{\lambda_{\min}(\boldsymbol{v}_{0}^2)}}\right) \Bigl\|  \sum\limits_{i = 1}^{n} \left(y_{i} - f_{\boldsymbol{\beta}^*}(x_{i}) \right) \Bigr\|,
\end{align*}
after distributing the $\sigma^{-2}$ out of the norm. Next, as $1 + \rho(u) +\delta(u) > 1$ and $v_0 =\frac{Cp\ln p}{\sigma^2 \sqrt{\lambda_{\min}(\boldsymbol{v}_{0}^2)}}$, \begin{align*}
    \left\lVert\boldsymbol{\xi}_{\underline{\boldsymbol{\epsilon}}}\right\lVert \leq v_0 \Bigl\| n^{-1/2}  \sum\limits_{i = 1}^{n} \left(y_{i} - f_{\boldsymbol{\beta}^*}(x_{i}) \right) \Bigr\|,
\end{align*}
after distributing $n^{-1/2}$ back into the norm. Finally, as  $\left(y_{i} - f_{\boldsymbol{\beta}^*}(x_{i}) \right) \sim N(0, \sigma^2)$ and each $(Y_i, X_i)$ pair are independent, $n^{-1/2}  \sum\limits_{i = 1}^{n} \left(y_{i} - f_{\boldsymbol{\beta}^*}(x_{i}) \right)$ is equal in distribution to a $N(0, \sigma^2)$ random variable. Thus, 
\begin{align*}
\left\lVert\left(\boldsymbol{\xi}_{\underline{\boldsymbol{\epsilon}}}\right)_p\right\lVert \leq (v_0)_p\left\lVert Y \right\lVert.
\end{align*}
where $Y \sim N(0, \sigma^2)$, for all $p \geq 2$. Thus, by Assumption \ref{ass:p_n}.iii, because $\mathbb{P}\left(\left\lVert Y \right\lVert(v_{0})_{p}  \leq \sqrt{n}u_{p}\right) \rightarrow 1 $,  we obtain $\mathbb{P}\left(  \left\lVert \left( \boldsymbol{\xi}_{\underline{\boldsymbol{\epsilon}}}\right)_p\right\lVert\leq \sqrt{n}u_{p} \right) \rightarrow 1$. Thus, as $$C_{\boldsymbol{\epsilon}, p}\left( \sqrt{n}u_{p} \right) = \Big\{ \boldsymbol{\hat{\beta}}_p \in \Theta_p(u_p) , \ \left\lVert \left(\boldsymbol{\xi}_{\boldsymbol{\underline{\epsilon}}}\right)_p\right\lVert  \leq \sqrt{n}u_{p}\Big \},$$  we have $\mathbb{P}\left( C_{\boldsymbol{\epsilon}, p}\left( \sqrt{n}u_{p} \right) \right) \rightarrow 1 $. \

Next, by Assumption \ref{ass:p_n}.ii, $$\rho(u_{p})\left(1 +\sqrt{-\ln(\alpha_{p}) + 2p(n)}\right)  \rightarrow 0.$$ Thus, $\Diamond_{\boldsymbol{\epsilon}}(u_p) + \Diamond_{\underline{\boldsymbol{\epsilon}}}(u_p) $ will converge to $0$ in probability. In addition, as we assumed $\rho(u_{p}) + \delta(u_{p}) \rightarrow 0$ (Assumption \ref{ass:u_p}.iii), that implies $\left(\left\lVert(\boldsymbol{\xi}_{\boldsymbol{\epsilon}})_{p} \right\lVert^{2} - \left\lVert (\boldsymbol{\xi}_{\underline{\boldsymbol{\epsilon}}})_{p} \right\lVert^{2}\right)/2$ also converges to $0$ in probability. As $\Delta_{\boldsymbol{\epsilon}, p}(u_{p})$ is the sum of these two values, this shows $\Delta_{\boldsymbol{\epsilon}, p}(u_{p}) {\rightarrow} 0$ in probability. 
\end{proof}
\subsection{Delta Method Extension and Proof of Theorem~\ref{thm:delta_method}}

We will now derive important properties regarding the mapping of $\boldsymbol{\beta}$ to the stationary point estimates, $\Phi(\boldsymbol{\beta})$. We first establish that $\Phi(\boldsymbol{\beta})$ is a Lipschitz function on the set $D_{m_{0}}(\boldsymbol{\beta})$.

\begin{lemma}
Assume Assumptions \ref{ass:CS}-\ref{ass:Identifiability}, and \ref{ass:pos_deriv_gamma} are satisfied. Then,
    $\Phi(\boldsymbol{\beta}) : \mathbb{R}^{p} \rightarrow \mathbb{R}^{M}$ is a Lipschitz continuous function on the set $D_{m_{0}}(\boldsymbol{\beta})$.
\end{lemma}
\begin{proof}
We will first show the mapping $\Phi_{i}(\boldsymbol{\beta}) = \gamma_{\boldsymbol{\beta}}^{-1}(b_{i})$ from $\mathbb{R}^{p}$ to $\mathbb{R}$ is Lipschitz on $D_{m_{0}}(\boldsymbol{\beta})$. Now, $\Phi_{i}(\boldsymbol{\beta})$ can be viewed as a composition of three functions $\Phi_{i, 1}(\boldsymbol{\beta})$, $\Phi_{i, 2}(\boldsymbol{\beta})$ and $\Phi_{i, 3}(\boldsymbol{\beta})$ where $\Phi_{i, 1}(\boldsymbol{\beta})$ maps $\boldsymbol{\beta}$ to its corresponding diffeomorphism, $\Phi_{i, 2}(\boldsymbol{\beta})$ maps the diffeomorphism to its inverse and and $\Phi_{i, 3}(\boldsymbol{\beta})$ maps the inverse to the real number $\gamma_{\boldsymbol{\beta}}^{-1}(b_{i})$. We will now show each of these individual functions are Lipschitz. For the diffeomorphism space, we will be using the $L^{\infty}$ metric and for $D_{m_{0}}(\boldsymbol{\beta})$, we will be using the standard Euclidean metric. \\ \\
$\boldsymbol{(}\Phi_{i, 1}(\boldsymbol{\beta}) \boldsymbol{)}$: By Lemma~\ref{lem:gradient_lem}, the gradient of $f_{\boldsymbol{\beta}}$ and $\gamma_{\boldsymbol{\beta}}$ is bounded by a function of $p$. Thus, the function that maps $\boldsymbol{\beta}$ to $\gamma_{\boldsymbol{\beta}}(x)$ is Lipschitz with constant $L_{p}$, not dependent on $x$. Thus, for $\boldsymbol{\beta}_{1}, \boldsymbol{\beta}_{2} \in D_{m_{0}}(\boldsymbol{\beta})$,  \begin{align*}
    \sup_{x \in [0, 1]}|\gamma_{\boldsymbol{\beta}_{1}}(x) - \gamma_{\boldsymbol{\beta}_{2}}(x)| \leq L_{p}\left\lVert \boldsymbol{\beta}_{1} - \boldsymbol{\beta}_{2}\right\lVert,
\end{align*}
which shows that $\Phi_{i, 1}(\boldsymbol{\beta})$ is Lipschitz. \\ \\
$\boldsymbol{(}\Phi_{i, 2}(\boldsymbol{\beta}) \boldsymbol{)}$: Let $\boldsymbol{\beta}_{1}, \boldsymbol{\beta}_{2} \in D_{m_{0}}(\boldsymbol{\beta})$.
Now, the mean value inequality states that \begin{align*}
    (b - a) \left(\min_{a \leq x\leq b} f'(x) \right) \leq  (f(b) - f(a)).
\end{align*}
Using this and the fact that for any $x \in [0, 1]$, $\exists y \in [0, 1]$ such that $\gamma_{\boldsymbol{\beta}_{2}}(y) = x$,  we derive \begin{align*}
    |\gamma_{\boldsymbol{\beta}_{1}}^{-1}(x) - \gamma_{\boldsymbol{\beta}_{2}}^{-1}(x)| \leq & \frac{1}{\inf_{x \in [0, 1]} \gamma_{\boldsymbol{\beta}_{1}}'(x)} \left|x - \gamma_{\boldsymbol{\beta}_{1}}\left(\gamma_{\boldsymbol{\beta}_{2}}^{-1}(x)\right)\right| \\
    \leq & \frac{1}{m_{0}}\left|\gamma_{\boldsymbol{\beta}_{2}}(y) - \gamma_{\boldsymbol{\beta}_{1}}\left(\gamma_{\boldsymbol{\beta}_{2}}^{-1}(x)\right)\right| \\ 
    \leq & \frac{1}{m_{0}}\sup_{x \in [0, 1]}\left|\gamma_{\boldsymbol{\beta}_{1}}(x) - \gamma_{\boldsymbol{\beta}_{2}}(x)\right|.
\end{align*}
As this inequality applies to all $x \in [0, 1]$, this shows that \begin{align*}
     \sup_{x \in [0, 1]}|\gamma_{\boldsymbol{\beta}_{1}}^{-1}(x) - \gamma_{\boldsymbol{\beta}_{2}}^{-1}(x)| \leq \frac{1}{m_{0}}\sup_{x \in [0, 1]}\left|\gamma_{\boldsymbol{\beta}_{1}}(x) - \gamma_{\boldsymbol{\beta}_{2}}(x)\right|.
\end{align*}
Thus, the inverse is a Lipschitz function on $D_{m_{0}}(\boldsymbol{\beta})$ using the $L^{\infty}$ metric. \\
\\
$\boldsymbol{(}\Phi_{i, 3}(\boldsymbol{\beta}) \boldsymbol{)}$: Clearly, \begin{align*}
    |\gamma_{\boldsymbol{\beta}_{1}}(b_{i}) - \gamma_{\boldsymbol{\beta}_{2}}(b_{i})| \leq \sup_{x \in [0, 1]}\left|\gamma_{\boldsymbol{\beta}_{1}}(x) - \gamma_{\boldsymbol{\beta}_{2}}(x)\right|.
\end{align*}
showing that $\Phi_{i, 3}(\boldsymbol{\beta})$ is Lipschitz. As the composition of functions is Lipschitz, therefore each component  $\Phi_{i}(\boldsymbol{\beta})$ is a Lipschitz function from $\mathbb{R}^{p}$ to $\mathbb{R}$ on $D_{m_{0}}(\boldsymbol{\beta})$. This proves that the entire multivariate function $\Phi(\boldsymbol{\beta})$ is Lipschitz on $D_{m_{0}}(\boldsymbol{\beta})$. 
\end{proof}
Therefore, by Rademacher’s theorem, $\Phi(\boldsymbol{\beta})$ is differentiable almost everywhere in any open subset of $D_{m_{0}}(\boldsymbol{\beta})$. As long as $\boldsymbol{\beta}$ does not fall within the negligible set within $D_{m_{0}}(\boldsymbol{\beta})$ where $\Phi(\boldsymbol{\beta})$ fails to be differentiable, the following theorem applies. 
\begin{theorem}
\label{thm:delta_method}
Assume Assumptions \ref{ass:CS}-\ref{ass:Identifiability}, and \ref{ass:pos_deriv_gamma} are satisfied. Let $\Phi : \mathbb{R}^{p} \to \mathbb{R}^{m}$ be a differentiable mapping at $\boldsymbol{\beta^*} \in D_{m_{0}}(\boldsymbol{\beta})$ with Jacobian $J_\Phi(\boldsymbol{\beta^*}) \neq 0$.  
Define $d_{\boldsymbol{\epsilon}} = \sqrt{\boldsymbol{f}_{\epsilon}}$ and suppose the conditions of Theorem \ref{thm:Normality} are satisfied. Then, on the random set $C_{\boldsymbol{\epsilon}}(\sqrt{n} u)$, we have
\begin{align}
    \Bigl\| \sqrt{n} \bigl(\Phi(\boldsymbol{\hat{\beta}}) - \Phi(\boldsymbol{\beta^*}) \bigr) 
    - J_\Phi(\boldsymbol{\beta^*}) d_{\boldsymbol{\epsilon}}^{-1} \boldsymbol{\xi}_{\epsilon} \Bigr\|^2
    \;\leq\; 2\, \sigma_1\bigl(J_\Phi(\boldsymbol{\beta^*}) d^{-1}_{\boldsymbol{\epsilon}}\bigr)\, \Delta_{\boldsymbol{\epsilon}}(u) + o_{p}(1),
\end{align}
    where $\sigma_1\left(J_\Phi(\boldsymbol{\beta^*})d^{-1}_{\boldsymbol{\epsilon}}\right)$ is the largest singular value of $J_\Phi(\boldsymbol{\beta^*}) d^{-1}_{\boldsymbol{\epsilon}}$.
\end{theorem}

\begin{proof}
    Performing a first order Taylor expansion, with $r^*$ as the remainder term, and rearranging terms, we obtain
    \begin{align*}
    &\hspace{-1cm}\Bigl\| \sqrt{n} \bigl(\Phi(\boldsymbol{\hat{\beta}}) - \Phi(\boldsymbol{\beta}^*) \bigr) 
    - J_\Phi(\boldsymbol{\beta}^*) d_{\boldsymbol{\epsilon}}^{-1} \boldsymbol{\xi}_{\epsilon} \Bigr\|^2\\
    = & \Bigl\| \sqrt{n}J_{\Phi}(\boldsymbol{\beta^*})(\boldsymbol{\hat{\beta}} - \boldsymbol{\beta^*})  - J_{\Phi}(\boldsymbol{\beta^*})d_{\boldsymbol{\epsilon}}^{-1} \boldsymbol{\xi}_{\boldsymbol{\epsilon}}+ \sqrt{n}r^* \Bigr\|^2  \\
    = & \Bigl\| \left(J_{\Phi}(\boldsymbol{\beta^*})d_{\boldsymbol{\epsilon}}^{-1}\right)(\sqrt{n}d_{\boldsymbol{\epsilon}})(\boldsymbol{\hat{\beta}} - \boldsymbol{\beta^*})  - J_{\Phi}(\boldsymbol{\beta^*})d_{\boldsymbol{\epsilon}}^{-1} \boldsymbol{\xi}_{\boldsymbol{\epsilon}}+ \sqrt{n}r^*  \Bigr\|^2 \\
    = & \Bigl\| J_{\Phi}(\boldsymbol{\beta^*})d_{\boldsymbol{\epsilon}}^{-1}\left(\sqrt{n\boldsymbol{f}_{\boldsymbol{\epsilon}}}(\boldsymbol{\hat{\beta}} - \boldsymbol{\beta^*}) - \boldsymbol{\xi}_{\boldsymbol{\epsilon}}  \right)  + \sqrt{n}r^* \Bigr\|^2.
\end{align*}
Using the triangle inequality and the fact for any matrix $A \in \mathbb{R}^{m \times p}$ and vector $x \in \mathbb{R}^{p}$, $||Ax||^{2} \leq \sigma_{1}(A)||x||^{2}$, \begin{align*}
     \Bigl\| \sqrt{n} \bigl(\Phi(\boldsymbol{\hat{\beta}}) - \Phi(\boldsymbol{\beta}^*) \bigr) 
    - J_\Phi(\boldsymbol{\beta}^*) d_{\boldsymbol{\epsilon}}^{-1} \boldsymbol{\xi}_{\epsilon} \Bigr\|^2 \leq \sigma_1\bigl(J_\Phi(\boldsymbol{\beta^*}) d_{\boldsymbol{\epsilon}}^{-1}\bigr)\Bigl\|\sqrt{n\boldsymbol{f}_{\boldsymbol{\epsilon}}}(\boldsymbol{\hat{\beta}} - \boldsymbol{\beta^*}) - \boldsymbol{\xi}_{\boldsymbol{\epsilon}}  \Bigr\|^2 + \Bigl\|\sqrt{n}r^* \Bigr\|^2.
\end{align*}
By Theorem \ref{thm:Normality}, \begin{align*}
   \Bigl\| \sqrt{n} \bigl(\Phi(\boldsymbol{\hat{\beta}}) - \Phi(\boldsymbol{\beta}^*) \bigr) 
    - J_\Phi(\boldsymbol{\beta}^*) d_{\boldsymbol{\epsilon}}^{-1} \boldsymbol{\xi}_{\epsilon} \Bigr\|^2   \leq 2\sigma_1\bigl(J_\Phi(\boldsymbol{\beta^*})\Delta_{\boldsymbol{\epsilon}}(u)\bigl) + \Bigl\|\sqrt{n}r^* \Bigr\|^2.
\end{align*}
Now, because $\boldsymbol{\beta}^* \in D_{m_{0}}(\boldsymbol{\beta}) \subset \boldsymbol{\Theta}$ is a compact set, $x \in [0, 1]$ and the log likelihood is continuous for all $\boldsymbol{\beta} \in  D_{m_{0}}, x \in [0, 1]$, the uniform law of large numbers applies. Finally, as each $\boldsymbol{\beta}$ defines a unique diffeomorphism, the MLE will be root-$n$ consistent. Thus, $\Bigl\|\sqrt{n}r^* \Bigr\|^2$ will be $o_{p}(1)$, which proves the claim. 
\end{proof}

\end{document}